\documentclass[brochure,12pt,french]{smfbourbaki}

\usepackage[T1]{fontenc}
\usepackage{lmodern,amssymb,bm}

\usepackage{babel}

\usepackage[utf8]{inputenc}

\usepackage[colorlinks=true, linkcolor=blue, citecolor=red, urlcolor=blue]{hyperref}

\usepackage[
backend=bibtex,
style=authoryear, 
citestyle=authoryear-comp,
maxnames=7,
sortcites=false 
]{biblatex}
\usepackage{csquotes}

\usepackage{graphicx}
\def\eps {{\epsilon}}
\def\llangle {{\langle\!\langle}}
\def\rrangle {{\rangle\!\rangle}}

\usepackage{tikz}
\usetikzlibrary{positioning}
\usetikzlibrary{patterns}

\newcommand*{\ctikz}[2][]{\hbox{\mathsurround=3pt$\vcenter{\hbox{\tikz[#1]{#2}}}$}}

\DeclareDelimFormat{nameyeardelim}{\addcomma\space}

\DeclareNameAlias{sortname}{given-family}

\renewcommand{\bibnamedash}{\leavevmode\raise3pt\hbox to3em{\hrulefill}\space}

\AtEveryBibitem{%
  \clearfield{issn} 
  \clearfield{isbn} 
  \clearfield{doi} 

  \ifentrytype{online}{}{
    \clearfield{url}
  }
}

\renewbibmacro{in:}{%
    \ifentrytype{article}{}{\printtext{\bibstring{in}\intitlepunct}}}

\DeclareFieldFormat[article,periodical,inreference]{number}{\mkbibparens{#1}}
\DeclareFieldFormat[article,periodical,inreference]{volume}{\mkbibbold{#1}}
\renewbibmacro*{volume+number+eid}{%
    \printfield{volume}%
    \setunit*{\addthinspace}
    \printfield{number}%
    \setunit{\addcomma\space}%
    \printfield{eid}}

\DeclareFieldFormat[article,inbook,incollection]{title}{\enquote{#1}\addcomma} 

\date{Novembre 2022}
\bbkannee{75\textsuperscript{e} année, 2022--2023}  
\bbknumero{1198}                                      

\title{Validité de la théorie cinétique des gaz : \\ au-delà de l'équation de Boltzmann}
\subtitle{d'après T. Bodineau, I. Gallagher, L. Saint-Raymond, S. Simonella}

\author{François Golse}
\address{École polytechnique, CMLS\\ Route de Saclay, 91128 Palaiseau Cedex}
\email{\\ francois.golse@polytechnique.edu}

\begin{document}

\maketitle

\tableofcontents

\section*{Introduction}

On doit à \textcite{Max1,Max2} et \textcite{Boltz1,Boltz2} la fondation de la théorie cinétique des gaz. Quoique l'idée d'une description particulaire de la matière remonte à l'Antiquité, Maxwell est le premier 
à faire appel à des notions de statistique afin de mettre cette description en équation. Le caractère novateur de cette entreprise apparaît très clairement à la lecture de l'introduction de \parencite{Max2}: 
la liste exhaustive des précurseurs cités par Maxwell est brève --- Lucrèce\footnote{Cf. par exemple {\og  Nam quoniam per inane uagantur, cuncta necessest | Aut gravitate sua ferri primordia rerum | Aut 
ictu forte alterius. Nam  < cum > cita saepe | Obuia conflixere, fit ut diuersa repente | Dissiliant [...]\fg}   (Car puisqu'il errent à travers le vide, il faut que les principes des choses soient tous emportés soit 
par leur propre poids, soit encore par le choc d'un autre atome). De Rerum Natura II, v. 83--87, trad. A. Ernout, Les Belles Lettres, Paris, 2009.  Ces vers décrivent précisément le processus physique 
sous-jacent à l'équation de Boltzmann \eqref{BoltzGen} ci-dessous. Voici ce qu'en dit \textcite{Max2}: {\og  [...] he describes the atoms as all moving downwards with equal velocities, which, at quite uncertain 
times and places, suffer an imperceptible change, just enough to allow of occasional collisions taking place between the atoms.\fg} }(I$^\mathrm{er}$ siècle av. J.-C.), D. Bernoulli (1738), Le Sage (1761), 
Herapath (1847), avant les travaux évidemment fondamentaux de Joule (1848) et Clausius (1857). Mais ce n'est qu'au XX$^\mathrm{e}$ siècle que l'on comprend que la théorie cinétique des gaz peut être 
déduite rigoureusement des équations de la mécanique classique appliquées à un système de particules sphériques identiques interagissant lors de collisions élastiques dans une certaine limite asymptotique. 
Hilbert est le premier à avoir formulé cette question comme un problème mathématique (cité dans son sixième problème sur l'{\og  axiomatisation de la physique\fg}   \parencite{Hilb1900}). Mais le texte 
fondamental identifiant clairement le régime asymptotique à considérer est l'article \parencite{Grad49}. Après les articles précurseurs \parencite{Gal69} et \parencite{Cer72}, \textcite{Lan75} donne la  première 
justification rigoureuse de l'équation de Boltzmann comme conséquence du principe fondamental de la dynamique (c'est-à-dire la deuxième loi de Newton) appliqué à chaque molécule de gaz.

Le théorème de Lanford a été étendu dans diverses directions depuis 1975: voir par exemple \parencite{King,IllnerPulvi,GSRT,Ayi}. Toutefois, un certain nombre de questions essentielles demeuraient ouvertes 
même après les plus récents de ces travaux. Ces questions sont principalement de deux types bien distincts. D'une part, sur le plan mathématique, le théorème de Lanford exprime qu'une certaine quantité 
converge vers une solution de l'équation de Boltzmann. Peut-on alors estimer l'erreur entre la limite, c'est-à-dire la solution de l'équation de Boltzmann, et la quantité qui l'approche? D'autre part, la convergence 
démontrée par Lanford a lieu sur un intervalle de temps assez restreint --- on trouvera une discussion assez détaillée de ce point par Lanford lui-même dans \parencite{LanfordAst} --- alors que l'équation de 
Boltzmann possède des solutions globales en temps pour des données initiales extrêmement générales (voir \parencite{dPL}, ainsi que l'exposé \parencite{PG88} dans ce même séminaire). En pratique, les 
spécialistes de gaz raréfiés utilisent d'ailleurs l'équation de Boltzmann ou certains de ses avatars pour des simulations numériques sur des plages de temps beaucoup plus longues que celle prédite par 
\parencite{Lan75}. Peut-on alors démontrer la validité de l'équation de Boltzmann sur des intervalles de temps arbitrairement longs?

Il semble à peu près évident qu'une réponse complète à cette dernière question devra nécessairement faire appel à des méthodes mathématiques très différentes de celles de \parencite{Lan75}, 
et qui restent probablement à inventer. 

En revanche, quant aux questions du premier type, des progrès spectaculaires ont été accomplis récemment dans une série de travaux importants \parencite{BGSRSToulouse, BGSRSjsp, BoGalSRSim1}. 
Citons également les articles \parencite{BGSRScpam, BGSRSLongT}, qui étendent la théorie des fluctuations étudiée dans \parencite{BGSRSToulouse, BGSRSjsp, BoGalSRSim1} par delà le temps de Lanford, 
mais dans une situation extrêmement particulière, celle des fluctuations autour d'un état d'équilibre. On évoquera brièvement ces derniers travaux dans la section \ref{S->TLanford} de cet exposé.

Il est évidemment impossible de rendre compte en quelques dizaines de pages de l'ensemble des résultats obtenus et du détail des méthodes mathématiques employées dans ces articles, 
qui sont malheureusement d'une grande technicité compte tenu de la difficulté du problème.

Toutefois, dans les contributions de Bodineau, Gallagher, Saint-Raymond et Simonella au problème de la justification rigoureuse de la théorie cinétique des gaz, un rôle essentiel semble dévolu à une nouvelle 
équation de type Hamilton--Jacobi {\og  fonctionnelle\fg}, équation satisfaite par une notion de {\og fonction génératrice des cumulants\fg}   dans la même limite que celle étudiée par Lanford.

C'est donc sur cette équation de Hamilton--Jacobi, et sur ses applications à la description statistique de la dynamique des gaz, que l'on a choisi de centrer cet exposé. On trouvera une présentation sensiblement 
différente de ces mêmes travaux dans la conférence plénière de Laure Saint-Raymond au Congrès international des mathématiciens (2022): voir \parencite{BGSRSicm} --- présentation dont on s'est toutefois 
inspiré ici pour évoquer certaines des méthodes de démonstration utilisées dans les travaux cités plus haut. 

Les questions étudiées dans le présent exposé font évidemment appel à plusieurs notions fondamentales relatives à la justification rigoureuse par Lanford de l'équation de Boltzmann, dont un compte-rendu 
très précis et détaillé se trouve dans \parencite{GSRT} (voir également \parencite{CIP}), et qui a déjà fait l'objet d'une présentation à ce même séminaire \parencite{FG2014}. On a essayé autant que possible 
d'éviter les redites entre ce précédent exposé et le texte qui va suivre. Un certain nombre de questions, comme par exemple les {\og  paradoxes\fg}   liés à l'irréversibilité, ont déjà été décrits et commentés dans 
\parencite{FG2014}; on a délibérément choisi de ne pas y revenir, et d'y renvoyer le lecteur chaque fois que cela était possible.

Je tiens à remercier Thierry Bodineau, Isabelle Gallagher, Laure Saint-Raymond et Sergio Simonella de m'avoir communiqué une première version de leur article \parencite{BGSRSicm}, ainsi que de leurs
suggestions pendant la préparation de cet exposé.

\section{Limite de Boltzmann--Grad et théorème de Lanford}


Commençons par rappeler dans cette section quelques notions de base, déjà présentées dans \parencite{FG2014}, mais indispensables pour la suite.

\subsection{La dynamique moléculaire}


Considérons un gaz monoatomique, que l'on voit comme un système de $N$ molécules qui sont des boules de diamètre $\eps\in]0,\tfrac12[$. La position et la vitesse de la $k$-ième boule à l'instant $t$ sont notées 
respectivement $x_k(t)$ et $v_k(t)\in\mathbf R^3$ où $k=1,\ldots,N$. Dans toute la suite, on supposera pour simplifier que l'évolution du gaz est spatialement périodique, de sorte que 
$x_k(t)\in\mathbf T^3=\mathbf R^3/\mathbf Z^3$. En l'absence de force extérieure agissant sur les molécules, la deuxième loi de Newton écrite pour chaque molécule de gaz est\footnote{Pour $x,y\in\mathbf T^3$,
on note $\text{dist}(x,y)=\min\{|X-Y|\text{ t.q. }X,Y\in\mathbf R^3\,,\,\,X=x\text{ et }Y=y\text{ mod. }\mathbf Z^3\}$.}
\begin{equation}\label{Newton}
\frac{dx_k}{dt}(t)=v_k(t)\,,\quad\frac{dv_k}{dt}(t)=0\,,\qquad\hbox{ si }\text{dist}(x_k(t),x_l(t))>\eps\hbox{ pour tout }k\not=l\,.
\end{equation}
Au cours d'une collision entre la $k$-ième et la $l$-ième molécule à un instant $t^*$, les positions de ces molécules varient continûment en temps, c'est-à-dire que
\begin{equation}\label{Collxk}
\text{dist}(x_k(t^*-0),x_l(t^*-0))=\eps\implies x_k(t^*+0)=x_k(t^*-0)\quad\text{ et }\quad x_l(t^*+0)=x_l(t^*-0)\,,
\end{equation}
tandis que leurs vitesses varient de façon discontinue comme suit:
\begin{equation}\label{Collvkvj}
\begin{aligned}
v_k(t^*+0)&=v_k(t^*-0)-((v_k(t^*-0)-v_l(t^*-0))\cdot n_{kl}(t^*))\,n_{kl}(t^*)\,,
\\
v_l(t^*+0)&=v_l(t^*-0)+((v_k(t^*-0)-v_l(t^*-0))\cdot n_{kl}(t^*))\,n_{kl}(t^*)\,,
\end{aligned}
\end{equation}
en notant\footnote{Soient $x,y\in\mathbf T^3$ tels que $r=\text{dist}(x,y)<\tfrac12$. Il existe un unique vecteur unitaire $n$ dans $\mathbf R^3$ tel que $y=x+rn$. Ce vecteur sera noté $n=(y-x)/r$ ou $(y-x)/|y-x|$ 
dans la suite de cet exposé.} $n_{kl}(t^*):=(x_l(t^*\pm0)-x_k(t^*\pm0))/\eps$. On notera dans la suite de cet exposé
$$
\Lambda^\eps_N:=\{(x_1,\ldots,x_N)\in(\mathbf{T}^3)^N\hbox{ t.q. }\text{dist}(x_k(t),x_l(t))>\eps\hbox{ pour }k,l=1,\ldots,N\,,\,\,k\not=l\}
$$
--- il s'agit de l'ensemble des positions physiquement admissibles pour les molécules, qui ne peuvent s'interpénétrer --- et $\Gamma^\eps_N:=\Lambda^\eps_N\times(\mathbf{R}^3)^N$, l'espace des phases à $N$ 
particules. On suppose connues les positions et les vitesses de chaque molécule à l'instant initial $t=0$, soit
\begin{equation}\label{CondinNewton}
x_k(0)=x_k^{in}\,,\quad v_k(0)=v_k^{in}\,,\qquad\qquad k=1,\ldots,N
\end{equation}
avec $(x_1^{in},v_1^{in},\ldots,x_N^{in},v_N^{in})\in\Gamma^\eps_N$, et on s'intéresse aux solutions 
$$
t\mapsto(x_1(t),v_1(t),\ldots,x_N(t),v_N(t))\in\Gamma^\eps_N
$$
de ce problème de Cauchy (\ref{Newton})-(\ref{Collxk})-(\ref{Collvkvj}) avec la condition initiale (\ref{CondinNewton}). Notons $m_N$ la mesure de Lebesgue sur $(\mathbf T^3\times\mathbf R^3)^N$.

\begin{prop}
Il existe $E\subset\overline{\Gamma^\eps_N}$ tel que $m_N(E)=0$ et vérifiant la propriété suivante: pour tout $(x^{in}_1,v^{in}_1,\ldots,x^{in}_N,v^{in}_N)\in\overline{\Gamma^\eps_N}\setminus E$, le problème 
de Cauchy (\ref{Newton})-(\ref{Collxk})-(\ref{Collvkvj})-(\ref{CondinNewton}) admet une unique solution
$$
t\mapsto(x_1(t),v_1(t),\ldots,x_N(t),v_N(t))=:S_t^{N,\eps}(x^{in}_1,v^{in}_1,\ldots,x^{in}_N,v^{in}_N)
$$
définie pour tout $t\in\mathbf{R}$. Ceci définit $S^{N,\eps}_t$ comme flot sur $\overline{\Gamma^\eps_N}\setminus E$: pour tout $t\in\mathbf{R}$, on a 
$S^{N,\eps}_t(\overline{\Gamma^\eps_N\setminus E})\subset\overline{\Gamma^\eps_N\setminus E}$ et $S^{N,\eps}_{t+s}=S^{N,\eps}_{t}\circ S^{N,\eps}_{s}$. D'autre part, la mesure $m_N$ est invariante sous 
l'action de $S^{N,\eps}_t$, c'est-à-dire que $m_N(S^{N,\eps}_t(A))=m_N(A)$ pour tout $A\subset\overline{\Gamma^\eps_N}$ mesurable et tout $t\in\mathbf{R}$.
\end{prop}

\smallskip
Voir \parencite{Alex}, ou encore le chapitre 4 (en particulier la Proposition 4.1.1) de \parencite{GSRT}. Ce résultat est décrit avec un peu plus de détails dans la section 3.1 de \parencite{FG2014}.

\smallskip
Dans la suite, il sera commode de noter $z^{in}_k:=(x^{in}_k,v^{in}_k)$ et $z_k(t):=(x_k(t),v_k(t))$.

\subsection{Ensemble grand-canonique et loi d'échelle de Boltzmann--Grad}


La théorie cinétique des gaz est obtenue à partir du système \eqref{Newton}-\eqref{Collxk}-\eqref{Collvkvj} dans un régime asymptotique très particulier, connu sous le nom de loi d'échelle de Boltzmann--Grad. 
On y suppose que le diamètre des molécules $\eps$ est très petit par rapport au diamètre du domaine $\mathbf T^3$ où elles sont confinées, tandis que leur nombre $N$ est très grand. La quantité $\eps^2$ 
est donc homogène à une surface --- $\tfrac14\pi\eps^2$ est la surface de la section équatoriale d'une molécule --- de sorte que la quantité $\ell:=\text{Vol}(\mathbf T^3)/(N\eps^2)$ est homogène à une
longueur, proportionnelle au libre parcours moyen --- c'est-à-dire à la longueur moyenne séparant deux collisions subies par la même molécule typique dans le gaz. La loi d'échelle de Boltzmann--Grad postule 
que cette longueur est du même ordre de grandeur que le diamètre du domaine spatial $\mathbf T^3$. En particulier, le volume $(N-1)\tfrac43\pi\eps^3$ de l'espace dans lequel aucun des points $x_k(t)$ ne 
peut pénétrer (dit {\og  volume exclu\fg}, rempli par les $N-1$ autres molécules) est $O(\eps)$ et donc négligeable dans la limite de Boltzmann--Grad. C'est pourquoi la théorie cinétique des gaz de Maxwell et 
Boltzmann, obtenue dans cette limite, ne peut décrire que des gaz parfaits. 

Contrairement au cadre considéré dans \parencite{Lan75,CIP,GSRT} ainsi que dans \parencite{FG2014}, où $N$ est un paramètre entier que l'on fait tendre vers l'infini\footnote{Formalisme dit de l'{\og  ensemble 
canonique\fg}.}, on supposera ici que $N$, ainsi que les positions et les vitesses initiales des $N$ molécules sont des variables aléatoires. Soit une suite $(\phi_n)_{n\ge 0}$ où  
$\phi_n\in C_b((\mathbf T^3\times\mathbf R^3)^n)$ pour tout $n\ge 1$, et où $\phi_0\in\mathbf R$. De façon équivalente, on considère la fonction définie sur l'espace grand-canonique $\Omega$ par
\begin{equation}\label{DefPhi}
\Phi:\,\Omega:=\bigcup_{n\ge 0}\{n\}\times(\mathbf T^3\times\mathbf R^3)^n\ni(N,z_1,\ldots,z_N)\mapsto\sum_{n\ge 0}\mathbf 1_{N=n}\phi_n(z_1,\ldots,z_n)\,,
\end{equation}
et on définit sa moyenne grand-canonique par la formule
$$
\mathbb E_\eps(\Phi):=\frac1{\mathcal Z_\eps}\sum_{n\ge 0}\frac{\mu_\eps^n}{n!}\int_{(\mathbf T^3\times\mathbf R^3)^n}\phi_n(Z_n)\mathbb F^{in}_n(Z_n)dZ_n\,,
$$
où on a noté $Z_n:=(z_1,\ldots,z_n)$, où $\mu_\eps>0$ est un paramètre qui sera précisé plus loin, et où
\begin{equation}\label{DefFnin}
\mathbb F^{in}_n(Z_n):=\prod_{i=1}^nf^{in}(z_i)\prod_{1\le j<k\le n}\mathbf 1_{\text{dist}(x_j,x_k)>\eps}\,.
\end{equation}
On définit de la sorte une mesure de probabilité borélienne $\mathbb P_\eps$ sur $\Omega$, et le terme {\og  ensemble grand-canonique\fg}   désigne le couple $(\Omega,\mathbb P_\eps)$.

Dans cette formule, $f^{in}$ est la fonction de distribution (à une molécule) qui joue le rôle de condition initiale dans l'équation de Boltzmann. Sans perte de généralité, on supposera que $f^{in}$ est une densité 
de probabilité sur l'espace des phases $\mathbf T^3\times\mathbf R^3$. Notons que 
$$
\tfrac43\pi n\eps^3>1\implies\prod_{1\le j<k\le n}\mathbf 1_{\text{dist}(x_j,x_k)>\eps}=0\,,
$$
de sorte que la somme définissant $\mathbb E_\eps(\Phi)$ ne comporte qu'un nombre fini de termes non nuls. Le réel $\mathcal Z_\eps$ est défini par la condition de normalisation $\mathbb E_\eps(1)=1$:
$$
\mathcal Z_\eps:=\sum_{n\ge 0}\frac{\mu_\eps^n}{n!}\int_{(\mathbf T^3\times\mathbf R^3)^n}\mathbb F^{in}_n(Z_n)dZ_n\,.
$$
\begin{lemm}\label{L-EN=}
Supposons que la densité de probabilité $f^{in}$ appartient à l'espace $L^\infty(\mathbf T^3;L^1(\mathbf R^3))$ et que $\eps^3\mu_\eps\to 0$ quand $\eps\to 0^+$. Alors
$$
\mathbb E_\eps(N)\sim\mu_\eps\qquad\text{ lorsque }\eps\to 0^+\,,
$$
où $N$ désigne la fonction $\Phi$ sur l'espace grand-canonique $\Omega$ associée par la formule \eqref{DefPhi} à la suite de fonctions $(\phi_n)_{n\ge 0}$ telle que $\phi_n(z_1,\ldots,z_n)=n$ pour tout $n\ge 0$ et tous $z_1,\ldots,z_n\in\mathbf T^3\times\mathbf R^3$.
\end{lemm}

\begin{proof} Posons $C:=\|f^{in}\|_{L^\infty(\mathbf T^3;L^1(\mathbf R^3))}$, et 
$$
\begin{aligned}
u_n:=&\int_{(\mathbf T^3\times\mathbf R^3)^n}\mathbb F^{in}_n(Z_n)dZ_n
\\
\le&\int_{(\mathbf T^3\times\mathbf R^3)^n}\mathbb F^{in}_{n-1}(Z_{n-1})dZ_{n-1}\int_{\mathbf T^3\times\mathbf R^3}f^{in}(z_n)dz_n&\le u_{n-1}\le 1\,.
\end{aligned}
$$
D'autre part
$$
\begin{aligned}
u_n\ge&\int_{(\mathbf T^3\times\mathbf R^3)^{n-1}}\left(\int_{\mathbf T^3\times\mathbf R^3}\left(1-\sum_{j=1}^{n-1}\mathbf 1_{|x_n-x_j|\le\eps}\right)f^{in}(z_n)dz_n\right)\mathbb F^{in}_{n-1}(Z_{n-1})dZ_{n-1}
\\
\ge& u_{n-1}(1-\tfrac43\pi C(n-1)\eps^3)\,.
\end{aligned}
$$ 
Donc
$$
\mathbb E_\eps(N)\le\frac1{\mathcal Z_\eps}\sum_{n\ge 1}\frac{\mu_\eps^n}{n!}n=\frac{\mu_\eps}{\mathcal Z_\eps}\sum_{n\ge 1}\frac{\mu_\eps^{n-1}}{(n-1)!}=\mu_\eps\,,
$$
tandis que, comme $u_{n-1}\le u_{n-2}$
$$
\begin{aligned}
\mathbb E_\eps(N)\ge&\frac{\mu_\eps}{\mathcal Z_\eps}\sum_{n\ge 0}\frac{\mu_\eps^{n-1}}{n!}nu_{n-1}(1-\tfrac43\pi C(n-1)\eps^3)
\\
\ge&\frac{\mu_\eps}{\mathcal Z_\eps}\sum_{n\ge 1}\frac{\mu_\eps^{n-1}u_{n-1}}{(n-1)!}-\tfrac43\pi C\eps^3\frac{\mu^2_\eps}{\mathcal Z_\eps}\sum_{n\ge 2}\frac{\mu_\eps^{n-2}u_{n-2}}{(n-2)!}
=\mu_\eps(1-\tfrac43\pi C\eps^3\mu_\eps)\,.
\end{aligned}
$$
\end{proof}

\begin{rema}
En supposant que $\mu_\eps\to+\infty$ lorsque $\eps\to 0^+$ et sous les mêmes hypothèses que dans le Lemme \ref{L-EN=}, on a $\mathcal Z_\eps\sim e^{\mu_\eps}$. La démonstration de ce fait, laissée 
au lecteur, suit de près celle du lemme.
\end{rema}

Dans le formalisme grand-canonique, on réalisera donc la loi d'échelle de Boltzmann--Grad en posant
$$
\mu_\eps=\eps^{-2}\,,\quad\text{ de sorte que }\eps^2\mathbb E_\eps(N)\to 1\text{ lorsque }\eps\to 0^+\,.
$$

Observons d'ailleurs que
$$
\mathbb P_\eps(\{N=n\})=\mathbb E_\eps(\mathbf 1_{N=n})=\mathcal Z_\eps^{-1}\frac{\mu_\eps^n}{n!}\int_{(\mathbf T^3\times\mathbf R^3)^n}\mathbb F^{in}_n(Z_n)dZ_n\,,
$$
formule qui évoque bien sûr la loi de Poisson. (Si on fait $\eps=0$ dans cette formule, et que l'on pose $\mu_0=\lambda$ --- ce qui contredit évidemment la loi d'échelle de Boltzmann--Grad --- on trouve 
en effet que $\mathcal Z_0=e^{\lambda}$, puis que $\mathbb P_0(\{N=n\})=e^{-\lambda}\frac{\lambda^n}{n!}$, ce qui est la définition de la loi de Poisson de paramètre $\lambda$).

Le lecteur familier des énoncés usuels du théorème de Lanford dans \parencite{CIP} ou \parencite{GSRT} pourra être surpris du choix de l'ensemble grand-canonique dans cette étude. En effet, le processus 
collisionnel considéré ici, à savoir des collisions binaires entre sphères dures, laisse invariant le nombre de molécules. N'est-il donc pas artificiel de considérer le nombre de molécules comme aléatoire? 
En réalité, le fait de prescrire le nombre total de molécules introduit nécessairement une corrélation entre ces molécules qui ne provient ni de la dynamique, ni du choix de la condition initiale. Le formalisme 
grand-canonique permet précisément d'éviter cela (voir la remarque (6) dans la section 2.4.1 de \parencite{PulviSim}).

\subsection{Mesure(s) empirique(s) et corrélations}


Pour tout $\eps>0$, posons
$$
\rho^\eps[N,Z_N]:=\frac1{\mu_\eps}\sum_{j=1}^N\delta_{z_j}\,,\qquad Z_N:=(z_1,\ldots,z_N)\in\Gamma^\eps_N\,.
$$
Il s'agit d'une fonction sur l'espace grand-canonique $\Omega$ à valeurs dans l'espace des mesures de Radon sur $\mathbf T^3\times\mathbf R^3$. Elle vérifie $\rho^\eps[N,Z_N]\ge 0$ et 
$\|\rho^\eps[N,Z_N]\|_{VT}=N/\mu_\eps$. Pour tout $k\ge 1$, on pose de même
$$
\rho^\eps_k[N,Z_N]:=\frac1{\mu^k_\eps}\sum_{\genfrac{}{}{0pt}{2}{j:\{1,\ldots,k\}\to\{1,\ldots,N\}}{\text{ injective }}}\delta_{z_{j(1)}}\otimes\ldots\otimes\delta_{z_{j(k)}}\ge 0\,,
$$
fonction définie sur $\Omega$ à valeurs dans l'espace des mesures de Radon sur $(\mathbf T^3\times\mathbf R^3)^k$. Elle vérifie $\rho^\eps_k\ge 0$ ainsi que
$\|\rho^\eps_k[N,Z_N]\|_{VT}=N(N-1)\ldots(N-k+1)/\mu^k_\eps$. On notera enfin
\begin{equation}\label{DefRhoepst}
\rho^\eps_t:=\rho^\eps[N,S^{N,\eps}_tZ_N]\,,\qquad\rho^\eps_{k,t}:=\rho^\eps_k[N,S^{N,\eps}_tZ_N]\,.
\end{equation}

A partir de là, on définit la suite des corrélations $F^\eps_k(t,\cdot)$ entre $k$ molécules à l'instant $t$ pour l'ensemble grand-canonique par la formule
\begin{equation}\label{DefCorrel}
\int_{(\mathbf T^3\times\mathbf R^3)^k}F^\eps_k(t,Z_k)h_k(Z_k)dZ_k:=\mathbb E_\eps(\langle\rho^\eps_{k,t},h_k\rangle)\,,\quad k\ge 1\,,
\end{equation}
pour toute fonction test $h_k\in C_b((\mathbf T^3\times\mathbf R^3)^k)$ symétrique en ses $k$ variables. 

Ce formalisme diffère de celui de la {\og  hiérarchie BBGKY\fg}: voir \parencite{FG2014}.  Posons
$$
\mathbb F_{n:k}(t,Z_k):=\int_{(\mathbf T^3\times\mathbf R^3)^{n-k}}\mathbb F^{in}_n(S^{n,\eps}_{-t}Z_n)dz_{k+1}\ldots dz_n\,.
$$
Cette formule vaut pour $n>k\ge 0$; on convient par ailleurs que $\mathbb F_{n:n}=\mathbb F_n=\mathbb F^{in}_n\circ S^{n,\eps}_{-t}$ et que $\mathbb F_{n:k}=0$ lorsque $k>n$.

Une remarque importante s'impose: la fonction de distribution à  $n$ corps initiale $\mathbb F^{in}_n$ est évidemment symétrique en les variables $z_1,\ldots,z_n$ comme le montre la formule \eqref{DefFnin}, 
et cette symétrie est propagée par la dynamique $S^{n,\eps}_t$, de sorte que $\mathbb F_n(t,\cdot)$ est également symétrique en les variables $z_1,\ldots,z_n$. Du point de vue de la physique, cette symétrie
traduit le fait que les molécules sont indistinguables.

Le point de vue de la hiérarchie BBGKY décrit dans \parencite{FG2014} consistait à étudier $\mathbb F_{n:1}$ dans la limite où $n\to+\infty$ avec $\eps\to 0^+$ vérifiant la condition de Boltzmann--Grad 
$n\eps^2=1$. La traduction entre ce formalisme et celui de l'ensemble grand-canonique considéré ici découle du lemme ci-dessous.

\begin{lemm}
Pour tout $\eps>0$ et tout $t\ge 0$, 
$$
F^\eps_k(t,\cdot)=\frac1{\mathcal Z_\eps}\sum_{n\ge k}\frac{\mu_\eps^{n-k}}{(n-k)!}\mathbb F_{n:k}(t,\cdot)\,,\qquad k\ge 1\,.
$$
\end{lemm}

\begin{proof}
En effet, pour toute fonction test $h_k\in C_b((\mathbf T^3\times\mathbf R^3)^k)$ symétrique, 
$$
\begin{aligned}
\mathbb E_\eps(\langle\rho^\eps_{k,t},h_k\rangle)=\frac1{\mathcal Z_\eps}\sum_{n\ge k}\frac{\mu_\eps^n}{n!}\int_{(\mathbf T^3\times\mathbf R^3)^{n-k}}\langle\rho^\eps_k[n,S^{n,\eps}_tZ_n],h_k\rangle
\mathbb F^{in}_n(Z_n)dZ_n
\\
=\frac1{\mathcal Z_\eps}\sum_{n\ge k}\frac{\mu_\eps^n}{n!}\int_{(\mathbf T^3\times\mathbf R^3)^{n-k}}\langle\rho^\eps_k[n,Z_n],h_k\rangle\mathbb F^{in}_n(S^{n,\eps}_{-t}Z_n)dZ_n
\\
=\frac1{\mathcal Z_\eps}\sum_{n\ge k}\frac{\mu_\eps^{n-k}}{n!}\int_{(\mathbf T^3\times\mathbf R^3)^{n-k}}\sum_{\genfrac{}{}{0pt}{2}{j:\{1,\ldots,k\}\to\{1,\ldots,N\}}{\text{ injective }}}h_k(z_{j(1)},\ldots,z_{j(k)})
\mathbb F^{in}_n(S^{n,\eps}_{-t}Z_n)dZ_n
\\
=\frac1{\mathcal Z_\eps}\sum_{n\ge k}\frac{\mu_\eps^{n-k}}{n!}n(n-1)\ldots(n-k+1)\int_{(\mathbf T^3\times\mathbf R^3)^{n-k}}h_k(z_1,\ldots,z_k)\mathbb F^{in}_n(S^{n,\eps}_{-t}Z_n)dZ_n
\\
=\frac1{\mathcal Z_\eps}\sum_{n\ge k}\frac{\mu_\eps^{n-k}}{(n-k)!}\int_{(\mathbf T^3\times\mathbf R^3)^{n-k}}h_k(z_1,\ldots,z_k)\mathbb F_{n:k}(t,Z_n)dZ_n&\,.
\end{aligned}
$$
La mesure de Lebesgue sur $\Gamma^\eps_n$ est invariante par $S^{n,\eps}_t$, d'où la seconde égalité. La quatrième égalité découle de la symétrie des fonctions $h_k$ et $\mathbb F_n(t,\cdot)$.
\end{proof}

\subsection{Équation de Liouville et corrélations}


La formule $\mathbb F_n(t,Z_n)=\mathbb F^{in}_n(S^{\eps,n}_{-t}Z_n)$ montre que $\mathbb F_n$ est constante sur les courbes intégrales de \eqref{Newton}-\eqref{Collxk}-\eqref{Collvkvj}, d'où l'on déduit 
l'équation de Liouville
\begin{equation}\label{LiouvilleBC}
\begin{aligned}
{}&\partial_t\mathbb F_n(t,Z_n)+\sum_{i=1}^nv_i\cdot\nabla_{x_i}\mathbb F_n(t,Z_n)=0\,,\qquad Z_n\!\in\!\Gamma^\eps_n\,,
\\
&\mathbb F_n(t,Z_n)=\mathbb F_n(t,\hat Z_n[i,j])\qquad\text{ si }\text{dist}(x_i(t),x_j(t))=\eps\,,
\end{aligned}
\end{equation}
où 
$$
\begin{aligned}
\hat Z_n[i,j]:=(x_1,v_1,\ldots,x_i,v'_i,\ldots,x_j,v'_j,\ldots,x_n,v_n)\,,\qquad\text{ lorsque }\text{dist}(x_i(t),x_j(t))=\eps\,,
\\
v'_i\!:=\!v_i-((v_i-v_j)\cdot n_{ji})n_{ji}\,,\quad v'_j\!:=\!v_j\!+((v_i-v_j)\cdot n_{ji})n_{ji}\,,\,\,\,\,\,\quad\qquad\text{ où }n_{ji}:=\tfrac{x_i-x_j}\eps\,.
\end{aligned}
$$
En particulier, la condition aux limites ajoutée à l'équation de Liouville découle de \eqref{Collxk}-\eqref{Collvkvj}. Il sera commode de remplacer le problème aux limites ci-dessus pour l'équation de Liouville, 
posé sur le domaine $\Gamma^\eps_n$, par une équation au sens des distributions sur $(\mathbf T^3\times\mathbf R^3)^n$. Par définition $\mathbb F_n(t,Z_n)=0$ dans 
$(\mathbf T^3\times\mathbf R^3)^n\setminus\Gamma^\eps_n$, de sorte que
\begin{equation}\label{LiouvilleDistrib}
\begin{aligned}
\partial_t\mathbb F_n(t,Z_n)+\sum_{i=1}^nv_i\cdot\nabla_{x_i}\mathbb F_n(t,Z_n)=&\sum_{1\le i<j\le n}\mathbb F_n\big|_{\partial^+\Gamma^{\eps}_n}(v_j-v_i)\cdot n_{ij}\delta_{\text{dist}(x_i,x_j)=\eps}
\\
=&\sum_{1\le i<j\le n}\mathbb F_n\big|_{\partial^+\Gamma^{\eps}_n}((v_j-v_i)\cdot n_{ij})_+\delta_{\text{dist}(x_i,x_j)=\eps}
\\
&-\sum_{1\le i<j\le n}\mathbb F_n\big|_{\partial^+\Gamma^{\eps}_n}((v_j-v_i)\cdot n_{ij})_-\delta_{\text{dist}(x_i,x_j)=\eps}
\end{aligned}
\end{equation}
au sens des distributions sur $(\mathbf T^3\times\mathbf R^3)^n$. Pour tout réel $r$, on pose $r_+=\max(r,0)$ et $r_-=\max(-r,0)$, tandis que la trace interne de $\mathbb F_n$ sur $\partial\Gamma^\eps_n$
est notée
$$
\mathbb F_n\big|_{\partial^+\Gamma^{\eps}_n}:=\lim_{\eta\to 0^+}\mathbb F_n\big|_{\partial\Gamma^{\eps+\eta}_n}\,.
$$
Enfin $\delta_{\text{dist}(x_i,x_j)=\eps}$ désigne la distribution de simple couche de densité $1$ portée par l'hypersurface de $(\mathbf T^3)^n$ d'équation $\text{dist}(x_i,x_j)=\eps$. (Pour la justification de \eqref{LiouvilleDistrib}, 
voir la formule (20) de \parencite{FG2014}, ainsi que la formule (II.3.1) de \parencite{Schwartz}). 

On utilise alors la condition aux limites de \eqref{LiouvilleBC} pour exprimer $\mathbb F_n\big|_{\partial^+\Gamma^{\eps}_n}$ aux endroits où $(v_j-v_i)\cdot n_{ij}>0$ (précaution absolument essentielle, 
comme on le verra): 
\begin{equation}\label{LiouvilleDistrib2}
\begin{aligned}
\partial_t\mathbb F_n(t,Z_n)+\sum_{i=1}^nv_i\cdot\nabla_{x_i}\mathbb F_n(t,Z_n)
\\
=\sum_{1\le i<j\le n}\mathbb F_n(t,\hat Z_n[i,j])\big|_{\partial^+\Gamma^{\eps}_n}((v_j-v_i)\cdot n_{ij})_+\delta_{\text{dist}(x_i,x_j)=\eps}
\\
-\sum_{1\le i<j\le n}\mathbb F_n(t,Z_n)\big|_{\partial^+\Gamma^{\eps}_n}((v_j-v_i)\cdot n_{ij})_-\delta_{\text{dist}(x_i,x_j)=\eps}&\,.
\end{aligned}
\end{equation}
En intégrant chaque membre de cette égalité par rapport aux variables $z_2,\ldots,z_n$, on aboutit à 
\begin{equation}\label{BBGKY2}
\partial_t\mathbb F_{n:1}(t,z_1)+v_1\cdot\nabla_{x_1}\mathbb F_{n:1}(t,z_1)=(n-1)\eps^2\mathcal B_\eps^{12}(\mathbb F_{n:2})(t,z_1)
\end{equation}
où
$$
\begin{aligned}
\mathcal B_\eps^{12}(\mathbb F_{n:2})(t,z_1)=&\iint_{\mathbf R^3\times\mathbf S^2}\mathbb F_{n:2}(t,x_1,v'_1,x_1-\eps\omega,v'_2)((v_1-v_2)\cdot\omega)_+dv_2d\omega
\\
&-\!\iint_{\mathbf R^3\times\mathbf S^2}\!\mathbb F_{n:2}(t,x_1,v_1,x_1\!+\!\eps\omega,v_2)((v_1\!-\!v_2)\cdot\omega)_+dv_2d\omega\,,
\end{aligned}
$$
où on a posé
$$
v'_1:=v_1-((v_1-v_2)\cdot\omega)\omega\,,\qquad v'_2:=v_2+((v_1-v_2)\cdot\omega)\omega\,.
$$

On renvoie le lecteur à \parencite{FG2014}, tout particulièrement à l'argument permettant de passer de \eqref{LiouvilleDistrib}, c'est-à-dire de l'égalité (20) de \parencite{FG2014}, à \eqref{BBGKY2}, autrement 
dit à l'égalité (22) et à la formule (23) de \parencite{FG2014}.

Multiplions maintenant chaque membre de \eqref{BBGKY2} par $\mu_\eps^{n-1}/(n-1)!$, et sommons les expressions ainsi obtenues pour $n\ge 1$. Comme $\eps^2\mu_\eps=1$, il vient
\begin{equation}\label{CorrEq1}
\begin{aligned}
(\partial_t+v_1\cdot\nabla_{x_1})F^\eps_1(t,z_1)=&(\partial_t+v_1\cdot\nabla_{x_1})\frac1{\mathcal Z_\eps}\sum_{n\ge 1}\frac{\mu_\eps^{n-1}}{(n-1)!}\mathbb F_{n:1}(t,z_1)
\\
=&\frac1{\mathcal Z_\eps}\sum_{n\ge 2}\frac{\mu_\eps^{n-2}}{(n-2)!}\mathcal B_\eps^{12}(\mathbb F_{n:2})(t,z_1)=\mathcal B_\eps^{12}(F^\eps_2)(t,z_1)\,.
\end{aligned}
\end{equation}

\subsection{L'hypothèse de chaos moléculaire et le théorème de Lanford}


L'égalité \eqref{CorrEq1} n'est pas vraiment une équation pour $F^\eps_1$, puisqu'elle fait intervenir $F^\eps_2$.

L'idée clé de Boltzmann lui permettant d'arriver à l'équation portant son nom est que deux molécules \textit{sur le point d'entrer en collision} sont statistiquement indépendantes, hypothèse dite du {\og chaos 
moléculaire\fg}  (voir par exemple les sections 8 et 11 de \parencite{Grad58}). Évidemment, deux molécules \textit{venant juste d'entrer en collision} ne peuvent pas être statistiquement indépendantes. 
Considérons alors l'intégrande du terme $\mathcal B_\eps^{12}(F^\eps_2)$, à savoir
$$
(F^\eps_2(t,x_1,v'_1,x_1-\eps\omega,v'_2)-F^\eps_2(t,x_1,v_1,x_1+\eps\omega,v_2))((v_1-v_2)\cdot\omega)_+\,.
$$
Dans le second terme de cette différence, on a
$$
(v_1-v_2)\cdot\omega=-\tfrac1\eps(v_1-v_2)\cdot(x_1-x_2)>0\,,
$$
de sorte que la molécule située en $x_1$ de vitesse $v_1$ s'approche de la molécule située en $x_2=x_1+\eps\omega$ et de vitesse $v_2$ (voir la figure \ref{F-Collision}). Comme ces deux molécules sont 
sur le point d'entrer en collision, l'hypothèse du chaos moléculaire entraîne que
$$
F^\eps_2(t,x_1,v_1,x_1+\eps\omega,v_2)\simeq F_1(t,x_1,v_1)F_1(t,x_1,v_2)\,,
$$
où
$$
F_1(t,x_1,v_1):=\lim_{\eps\to 0^+}F^\eps_1(t,x_1,v_1)\,.
$$
De même, dans le premier terme de cette différence, on a
$$
(v_1-v_2)\cdot\omega=\tfrac1\eps(v_1-v_2)\cdot(x_1-x_2)=-\tfrac1\eps(v'_1-v'_2)\cdot(x_1-x_2)>0\,,
$$
de sorte que la molécule située en $x_1$ de vitesse $v'_1$ s'approche de la molécule située en $x_2=x_1-\eps\omega$ et de vitesse $v'_2$ (voir la figure \ref{F-Collision}). Ces deux molécules sont donc 
elles aussi sur le point d'entrer en collision, et l'hypothèse du chaos moléculaire entraîne que
$$
F^\eps_2(t,x_1,v'_1,x_1-\eps\omega,v'_2)\simeq F_1(t,x_1,v'_1)F_1(t,x_1,v'_2)\,.
$$


\begin{figure}\label{F-Collision}

\begin{center}

\includegraphics[width=8cm]{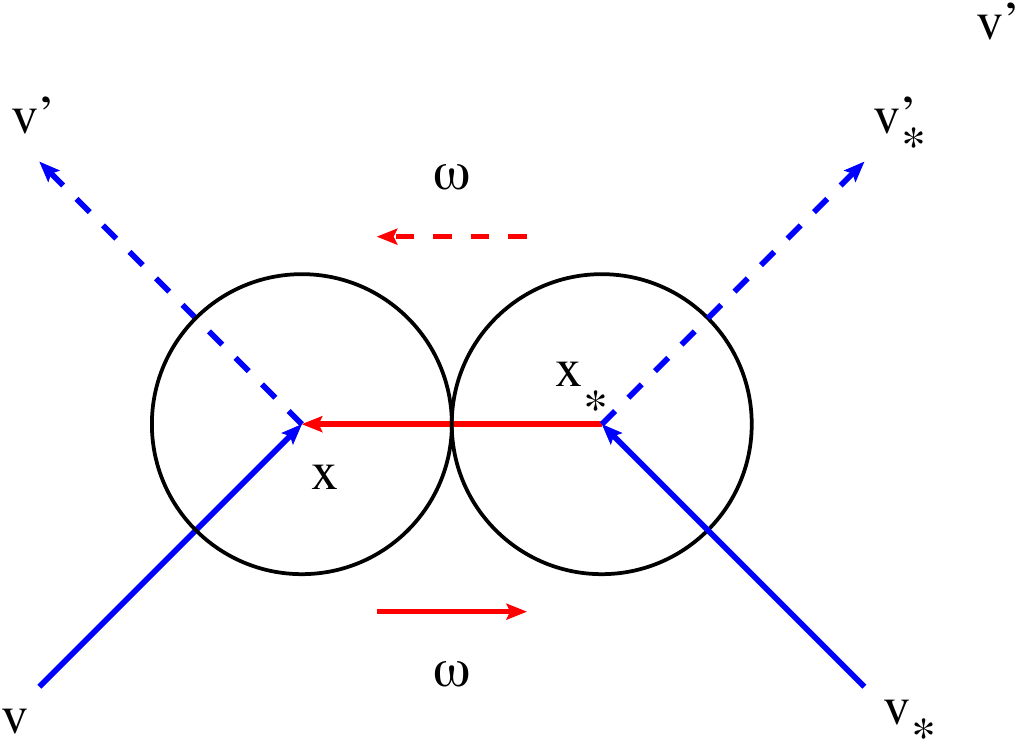}

\caption{Collision binaire. Avant la collision $(v-v_*)\cdot(x-x_*)<0$, après collision $(v'-v'_*)\cdot(x-x_*)=-(v-v_*)\cdot(x-x_*)>0$. On pose $\omega=-\frac{x-x_*}{|x-x_*|}$ avant collision, de sorte que 
$x_*=x+\eps\omega$ et $(v-v_*)\cdot\omega>0$, et $\omega=\frac{x-x_*}{|x-x_*|}$ après collision, de sorte que $x_*=x-\eps\omega$ et $(v-v_*)\cdot\omega>0$.}

\end{center}

\end{figure}


En passant formellement à la limite dans \eqref{CorrEq1}, et en tenant compte des implications de l'hypothèse du chaos moléculaire de Boltzmann mentionnées ci-dessus, on trouve que $F_1$ est 
solution de l'équation de Boltzmann
\begin{equation}\label{BoltzEq}
(\partial_t+v\cdot\nabla_x)F_1(t,z)=\mathcal B(F_1)(t,z)\,,
\end{equation}
où $\mathcal B(F_1)$ est l'intégrale de collision de Boltzmann, soit
\begin{equation}\label{CollInt}
\mathcal B(F_1)(t,z)\!:=\!\int_{\mathbf R^3\times\mathbf S^2}\!(F_1(t,x,v')F_1(t,x,v'_*)\!-\!F_1(t,x,v)F_1(t,x,v_*))((v-v_*)\cdot\omega)_+dv_*d\omega,
\end{equation}
avec la notation
\begin{equation}\label{LoiColl}
v'\equiv v'(v,v_*,\omega):=v-((v-v_*)\cdot\omega)\omega\,,\quad v'_*\equiv v'_*(v,v_*,\omega):=v_*+((v-v_*)\cdot\omega)\omega\,.
\end{equation}

Le raisonnement ci-dessus montre tout l'intérêt d'avoir utilisé la condition aux limites de \eqref{LiouvilleBC} dans la section précédente pour exprimer le terme $\mathbb F_n\big|_{\partial^+\Gamma^{\eps}_n}$ 
là où $(v_j-v_i)\cdot n_{ij}>0$. Si l'on néglige cette étape, on aboutit à l'équation 
$$
(\partial_t+v_1\cdot\nabla_{x_1})F^\eps_1(t,z_1)=\int_{\mathbf R^3\times\mathbf S^2}F^\eps_2(t,x_1,v_1,x_1-\eps\omega,v_2)(v_1-v_2)\cdot\omega dv_2d\omega\,,
$$
au lieu de \eqref{CorrEq1}. Si l'on suppose maintenant que 
$$
F^\eps_2(t,x_1,v_1,x_1-\eps\omega,v_2)\simeq F_1(t,x_1,v_1)F_1(t,x_1,v_2)
$$
dans l'intégrale au membre de droite \textit{indépendamment du signe de} $(v_1-v_2)\cdot\omega$, on trouve, en passant formellement à la limite dans l'équation ci-dessus, que
$$
(\partial_t+v_1\cdot\nabla_{x_1})F_1(t,z_1)=\int_{\mathbf R^3}F_1(t,x_1,v_1)F_1(t,x_1,v_2)\left(\int_{\mathbf S^2}(v_1-v_2)\cdot\omega d\omega\right)dv_2=0\,.
$$
Autrement dit, l'hypothèse de factorisation de $F^\eps_2(t,x_1,v_1,x_1-\eps\omega,v_2)$ \textit{indépendamment du signe de} $(v_1-v_2)\cdot\omega$, et donc en particulier pour des molécules venant juste d'entrer 
en collision, par conséquent fortement corrélées, nous amènerait à la conclusion inintéressante --- et expérimentalement fausse --- que $F_1$ évolue suivant l'équation de transport libre (sans intégrale de collision).

Cette discussion montre la subtilité de l'hypothèse de Boltzmann, ainsi que la difficulté à la formaliser mathématiquement. La limite $F^\eps_1(t,z_1)\to F_1(t,z_1)$ lorsque $\eps\to 0^+$ doit avoir lieu dans une 
topologie permettant de déduire que
$$
F^\eps_2(t,x_1,v_1,x_1+\eps\omega,v_2)((v_1-v_2)\cdot\omega)_+\to F_1(t,x_1,v_1)F_1(t,x_1,v_2)((v_1-v_2)\cdot\omega)_+\,,
$$
mais pas que
$$
F^\eps_2(t,x_1,v_1,x_1-\eps\omega,v_2)((v_1-v_2)\cdot\omega)_+\to F_1(t,x_1,v_1)F_1(t,x_1,v_2)((v_1-v_2)\cdot\omega)_+
$$
lorsque $\eps\to 0^+$. Comme ces convergences ont lieu sur des ensembles de mesure nulle, la convergence p.p. pour la mesure de Lebesgue de $(\mathbf T^3\times\mathbf R^3)^2$ est insuffisante.

On renvoie le lecteur aux textes de \textcite{Grad49,Grad58} et à l'appendice A1 du livre \parencite{So} pour une analyse plus détaillée de ces questions.

\smallskip
Passons maintenant à l'énoncé du théorème de Lanford, qui est la justification rigoureuse de l'équation de Boltzmann \eqref{BoltzEq} à partir du système \eqref{Newton}-\eqref{Collxk}-\eqref{Collvkvj}.

\begin{theo}[Lanford]\label{T-Lanford}
Soit $f^{in}\in C^1(\mathbf T^3\times\mathbf R^3)$, densité de probabilité telle que
\begin{equation}\label{BndCondIni}
f^{in}(x,v)+|\nabla_xf^{in}(x,v)|\le C_0e^{-\beta_0|v|^2}\,,\qquad x\in\mathbf T^3\,,\,\,v\in\mathbf R^3\,,
\end{equation}
où $C_0,\beta_0>0$. Considérons la famille des fonctions de corrélations $(F^\eps_k)_{k\ge 1,\eps>0}$ définies par \eqref{DefCorrel} pour l'ensemble grand-canonique où $\mu_\eps=\eps^{-2}$, selon la loi d'échelle 
de Boltzmann--Grad. Alors, il existe $T_0=T[C_0,\beta_0]>0$ tel que, lorsque $\eps\to 0^+$

\noindent
(1) $F^\eps_1(t,\cdot)$ converge uniformément sur tout compact de $\mathbf T^3\times\mathbf R^3$ vers $f(t,\cdot)$ pour tout $t\in[0,T_0]$, où $f$ est l'unique solution de l'équation de Boltzmann \eqref{BoltzEq} 
sur $\mathbf T^3\times\mathbf R^3$ vérifiant la condition initiale 
\begin{equation}\label{BoltzCondIn}
f(0,x,v)=f^{in}(x,v)\,,\qquad x\in\mathbf T^3\,,\,\,v\in\mathbf R^3\,;
\end{equation}
(2) pour tous $k\ge 2$ et $t\in[0,T_0]$, la famille $F^\eps_k(t,Z_k)$ des fonctions de corrélation à $k$ molécules converge pour presque tout $Z_k\in(\mathbf T^3\times\mathbf R^3)^k$ vers 
$$
f(t,\cdot)^{\otimes k}(Z_k):=\prod_{j=1}^kf(t,z_j)\,.
$$
\end{theo}

\smallskip
Cet énoncé appelle une remarque importante: si l'on compare les points (1) et (2) du théorème de Lanford, on voit que la notion de convergence utilisée pour $F^\eps_k$ avec $k\ge 2$ est plus faible que celle utilisée
pour $F^\eps_1$. Ce point particulier avait été prévu par \textcite{Grad58} (voir section 11, p. 223, dernier paragraphe) bien avant la démonstration de Lanford. 

On trouve notamment dans \parencite{Grad58} la phrase suivante {\og  [...] it is possible to specify the exceptional set on which $F^\eps_2(t,z_1,z_2)$ does not converge to $F_1(t,z_1)F_1(t,z_2)$ rather precisely [...]\fg}  . Ce 
point particulier est étudié en détail par \textcite{BGSRSToulouse}, et nous y renvoyons le lecteur intéressé par cette question.

\subsection{Ce qu'on ne doit pas ignorer à propos de l'équation de Boltzmann}\label{SS-BoltzMath}


Bien que le but de cet exposé soit d'aller {\og  au-delà de l'équation de Boltzmann\fg}, nous utiliserons à plusieurs reprises certaines propriétés mathématiques de base de cette équation, propriétés que nous 
allons rappeler ici.

L'équation de Boltzmann s'écrit en toute généralité
\begin{equation}\label{BoltzGen}
(\partial_t+v\cdot\nabla_x)f(t,x,v)+\mathbf a(t,x)\cdot\nabla_vf(t,x,v)=\mathcal B(f)(t,x,v)\,,\qquad (x,v)\in\mathbf T^3\times\mathbf R^3\,,
\end{equation}
où $\mathbf a$ est un champ d'accélération provenant d'une force extérieure --- comme la gravité par exemple --- tandis que $\mathcal B(f)$ désigne l'intégrale des collisions de Boltzmann et décrit la variation
en temps de la population de molécules de vitesse $v$ due à des collisions avec des molécules de vitesses différentes. L'inconnue $f$ est la {\og  fonction de distribution\fg}  (en vitesse) des molécules, à savoir 
la densité par rapport à la mesure de Lebesgue sur l'espace des phases $\mathbf T^3\times\mathbf R^3$ du nombre de molécules situées au point $x$ à l'instant $t$ et animées de la vitesse $v$. Dans tout cet 
exposé, on négligera systématiquement l'effet de la force extérieure, de sorte que $\mathbf a=0$, et que
$$
(\partial_t+v\cdot\nabla_x)f(t,x,v)=\mathcal B(f)(t,x,v)\,,\qquad (x,v)\in\mathbf T^3\times\mathbf R^3\,.
$$

Nous aurons besoin des propriétés essentielles suivantes de l'intégrale des collisions de Boltzmann. Comme le montre la définition \eqref{CollInt}, cette intégrale des collisions n'agit que sur la dépendance en $v$
de la fonction de distribution. Il suffit donc de l'étudier sur des fonctions constantes en $(t,x)$.

\smallskip
\noindent
\textbf{Formulation faible.} Pour tout $\phi\in C(\mathbf R^3)$ à décroissance rapide, $\mathcal B(\phi)\in C(\mathbf R^3)$ est à décroissance rapide, et, pour toute fonction test $\psi\in C(\mathbf R^3)$ à croissance 
polynomiale à l'infini, on a les identités suivantes (dont la preuve sera esquissée plus loin, et où $v'$ et $v'_*$ sont donnés en fonction de $v,v_*$ et $\omega$ par \eqref{LoiColl})
$$
\begin{aligned}
\int_{\mathbf R^3}\mathcal B(\phi)(v)\psi(v)dv=\int_{\mathbf R^3\times\mathbf R^3\times\mathbf S^2}(\psi(v')-\psi(v))\phi(v)\phi(v_*)((v-v_*)\cdot\omega)_+dvdv_*d\omega
\\
=\int_{\mathbf R^3\times\mathbf R^3\times\mathbf S^2}(\psi(v'_*)-\psi(v_*))\phi(v)\phi(v_*)((v-v_*)\cdot\omega)_+dvdv_*d\omega
\\
=\tfrac12\int_{\mathbf R^3\times\mathbf R^3\times\mathbf S^2}(\psi(v')+\psi(v'_*)-\psi(v)-\psi(v_*))\phi(v)\phi(v_*)((v-v_*)\cdot\omega)_+dvdv_*d\omega
\\
=\tfrac14\int_{\mathbf R^3\times\mathbf R^3\times\mathbf S^2}(\psi(v)+\psi(v_*)-\psi(v')-\psi(v'_*))(\phi(v')\phi(v'_*)-\phi(v)\phi(v_*))
\\
\times((v-v_*)\cdot\omega)_+dvdv_*d\omega&\,.
\end{aligned}
$$
\textbf{Lois de conservation locales.} Pour tout $\phi\in C(\mathbf R^3)$ à décroissance rapide, 
$$
\begin{aligned}
\int_{\mathbf R^3}\mathcal B(\phi)(v)dv=&0\quad\text{ (conservation locale de la masse), }
\\
\int_{\mathbf R^3}\mathcal B(\phi)(v)\tfrac12|v|^2dv=&0\quad\text{ (conservation locale de l'\'energie), }
\end{aligned}
$$
ainsi que
$$
\int_{\mathbf R^3}\mathcal B(\phi)(v)vdv=0\quad\text{ (conservation locale de l'impulsion). }
$$
Ces lois de conservation se déduisent de la formulation faible en observant que 
$$
\begin{aligned}
\psi(v)=\alpha_0+\sum_{i=1}^3\alpha_iv_i+\alpha_4|v|^2\implies
\\
\psi(v')+\psi(v'_*)-\psi(v)-\psi(v_*)=0\quad\text{ pour tous }v,v_*\in\mathbf R^3\text{ et }\omega\in\mathbf S^2\,,
\end{aligned}
$$
où les vitesses $v'$ et $v'_*$ sont données en fonction de $v,v_*,\omega$ par les relations \eqref{LoiColl}.

Évidemment, si $\phi\equiv\phi(t,x,v)$ est une solution de classe $C^1$ de l'équation de Boltzmann sur $[0,T[\times\mathbf T^3\times\mathbf R^3$ à décroissance rapide en $v$ ainsi que ses dérivées premières, 
on déduit des formules ci-dessus que
$$
\begin{aligned}
\partial_t\int_{\mathbf R^3}\phi(t,x,v)dv+\nabla_x\cdot\int_{\mathbf R^3}v\phi(t,x,v)dv=0\,,
\\
\partial_t\int_{\mathbf R^3}v_i\phi(t,x,v)dv+\nabla_x\cdot\int_{\mathbf R^3}vv_i\phi(t,x,v)dv=0\,,
\\
\partial_t\int_{\mathbf R^3}\tfrac12|v|^2\phi(t,x,v)dv+\nabla_x\cdot\int_{\mathbf R^3}v\tfrac12|v|^2\phi(t,x,v)dv=0\,,
\end{aligned}
$$
qui sont bien des lois de conservation locales, puisque chacune se met sous la forme $\mathrm{div}_{t,x}V(t,x)=0$, où $V$ est un champ de vecteurs sur $]0,T[\times\mathbf T^3$.

\smallskip
\noindent
\textbf{Théorème H de Boltzmann.} Soit $\phi\in C(\mathbf R^3)$ à décroissance rapide telle que $\phi>0$ et $\ln\phi$ soit à croissance polynomiale (par exemple $\phi(v)=e^{-P(|v|^2)}$ où $P$ est une fonction
polynomiale de coefficient dominant strictement positif). Alors
$$
\begin{aligned}
\int_{\mathbf R^3}\mathcal B(\phi)(v)\ln\phi(v)dv=-\tfrac14\int_{\mathbf R^3\times\mathbf R^3\times\mathbf S^2}(\phi(v')\phi(v'_*)-\phi(v)\phi(v_*))\ln\left(\frac{\phi(v')\phi(v'_*)}{\phi(v)\phi(v_*)}\right)
\\
\times((v-v_*)\cdot\omega)_+dvdv_*d\omega\le 0&\,,
\end{aligned}
$$
puisque $\ln$ est une fonction croissante. D'autre part
$$
\mathcal B(\phi)=0\iff\int_{\mathbf R^3}\mathcal B(\phi)(v)\ln\phi(v)dv=0\iff\phi\text{ est une maxwellienne,}
$$
c'est-à-dire que $\phi$ est de la forme
\begin{equation}\label{Maxw}
\mathcal M_{\rho,u,\theta}(v):=\frac{\rho}{(2\pi\theta)^{3/2}}e^{-|v-u|^2/2\theta}\,.
\end{equation}

\smallskip
De nouveau, si $\phi\equiv\phi(t,x,v)>0$ est une solution de classe $C^1$ de l'équation de Boltzmann sur $[0,T[\times\mathbf T^3\times\mathbf R^3$ à décroissance rapide en $v$ ainsi que ses dérivées premières, 
telle que $\ln\phi$ soit à croissance polynomiale en $v$, on en déduit que
$$
\frac{d}{dt}\int_{\mathbf T^3\times\mathbf R^3}\phi(t,x,v)\ln\phi(t,x,v)dxdv\le 0\,.
$$
Ceci permet en particulier de majorer la quantité
$$
\int_{\mathbf T^3\times\mathbf R^3}\phi(t,x,v)\ln\phi(t,x,v)dxdv
$$
en fonction de sa donnée initiale. Cette propriété est évidemment l'une des clés pour obtenir des solutions globales de l'équation de Boltzmann sans restriction de taille sur les données initiales \parencite{dPL,PG88}.

\smallskip
Le seul point un peu difficile dans les énoncés ci-dessus est la caractérisation des maxwelliennes par l'équation fonctionnelle $\ln(\phi(v)\phi(v_*))=\ln(\phi(v')\phi(v'_*))$: voir par exemple le Théorème 3.1.1 dans 
\parencite{CIP}. Tous les autres énoncés se déduisent de la formulation faible, dont nous allons maintenant expliquer la preuve.

D'abord, $\phi$ étant continue à décroissance rapide et $\psi$ étant à croissance polynomiale, la fonction
$$
(v,v_*)\mapsto|\phi(v')|+|\phi(v'_*)|
$$
est à décroissance rapide, tandis que la fonction 
$$
(v,v_*)\mapsto|\psi(v')|+|\psi(v'_*)|
$$
est à croissance polynomiale pour tout $\omega\in\mathbf S^2$, sachant que les vitesses $v',v'_*$ sont données en fonction de $v,v_*$ et $\omega$ par \eqref{LoiColl}. 

Puis, pour tout $\omega\in\mathbf S^2$, l'application $(v,v_*)\mapsto (v',v'_*)$  est une isométrie linéaire involutive de $\mathbf R^3\times\mathbf R^3$, et conserve donc la mesure de Lebesgue $dvdv_*$. Donc
$$
\begin{aligned}
\int_{\mathbf R^3\times\mathbf R^3\times\mathbf S^2}\psi(v)\phi(v')\phi(v'_*)((v-v_*)\cdot\omega)_+dvdv_*d\omega
\\
=\int_{\mathbf R^3\times\mathbf R^3\times\mathbf S^2}\psi(v')\phi(v)\phi(v_*)(-(v-v_*)\cdot\omega)_+dvdv_*d\omega
\\
=\int_{\mathbf R^3\times\mathbf R^3\times\mathbf S^2}\psi(v')\phi(v)\phi(v_*)((v-v_*)\cdot\omega)_+dvdv_*d\omega&\,,
\end{aligned}
$$
où la première égalité utilise que $(v'-v'_*)\cdot\omega=-(v-v_*)\cdot\omega$, tandis que la seconde est basée sur le fait que le changement de variables $\omega\mapsto-\omega$ laisse $v'$ et $v'_*$ invariants.
Ceci démontre la première égalité de la formulation faible. La deuxième découle du fait que, pour tout $\omega\in\mathbf S^2$ fixé, la symétrie $v\mapsto v_*$ échange $v'$ et $v'_*$. Les deux dernières égalités 
de la formulation faible en découlent aussitôt. On se reportera à la section 3.1 de \parencite{CIP} pour un traitement plus détaillé de ces propriétés.

\subsection{Le théorème de Lanford comme loi des grands nombres}


Après cette digression sur les propriétés mathématiques de l'équation de Boltzmann, revenons au théorème de Lanford.

Le fait que la fonction de corrélation à deux molécules $F^\eps_2$ se factorise dans la limite de Boltzmann--Grad, ce qui est le point (2) du théorème de Lanford, donne une information sur la famille $\rho^\eps_t$ 
des mesures empiriques définies en \eqref{DefRhoepst}.

\begin{coro}\label{C-Lanford}
Soit $h\in C_c(\mathbf T^3\times\mathbf R^3)$. Sous les hypothèses du théorème de Lanford, pour tout $\eta>0$
$$
\lim_{\eps\to 0^+}\mathbb P_\eps\left(\left\{\left|\langle\rho^\eps_t,h\rangle\!-\!\int_{\mathbf T^3\times\mathbf R^3}\!h(z)f(t,z)dz\right|\!>\!\eta\right\}\right)=0\,.
$$
\end{coro}

\begin{proof}
L'inégalité de Bienaymé--Tchebychev dit que
$$
\begin{aligned}
\mathbb P_\eps\left(\left\{\left|\langle\rho^\eps_t,h\rangle\!-\!\int_{\mathbf T^3\times\mathbf R^3}\!h(z)f(t,z)dz\right|\!>\!\eta\right\}\right)
\\
\le\frac1{\eta^2}\mathbb E_\eps\left(\left|\langle\rho^\eps_t,h\rangle\!-\!\int_{\mathbf T^3\times\mathbf R^3}\!h(z)f(t,z)dz\right|^2\right)&\,.
\end{aligned}
$$
Or
$$
\begin{aligned}
\mathbb E_\eps\left(\left|\langle\rho^\eps_t,h\rangle\!-\!\int_{\mathbf T^3\times\mathbf R^3}\!h(z)f(t,z)dz\right|^2\right)=\mathbb E_\eps(\langle\rho^\eps_t,h\rangle^2)
+\left(\int_{\mathbf T^3\times\mathbf R^3}\!h(z)f(t,z)dz\right)^2
\\
-2\int_{\mathbf T^3\times\mathbf R^3}\!h(z)F^\eps_1(t,z)dz\int_{\mathbf T^3\times\mathbf R^3}\!h(z)f(t,z)dz\,,
\end{aligned}
$$
et un calcul simple montre que
$$
\langle\rho^\eps_t,h\rangle^2=\langle\rho^\eps_t\otimes\rho^\eps_t,h\otimes h\rangle=\frac1{\mu_\eps}\langle\rho^\eps_t,h^2\rangle+\langle\rho^\eps_{2,t},h\otimes h\rangle\,,
$$
de sorte que
$$
\mathbb E_\eps(\langle\rho^\eps_t,h\rangle^2)=\frac1{\mu_\eps}\int_{\mathbf T^3\times\mathbf R^3}\!h(z)^2F^\eps_1(t,z)dz+\int_{(\mathbf T^3\times\mathbf R^3)^2}\!h(z_1)h(z_2)F^\eps_2(t,Z_2)dZ_2\,.
$$
Comme $h$ est à support compact, le point (1) du théorème de Lanford montre que
$$
\begin{aligned}
\frac1{\mu_\eps}\int_{\mathbf T^3\times\mathbf R^3}\!h(z)^2F^\eps_1(t,z)dz-2\int_{\mathbf T^3\times\mathbf R^3}\!h(z)F^\eps_1(t,z)dz\int_{\mathbf T^3\times\mathbf R^3}\!h(z)f(t,z)dz
\\
\to-2\left(\int_{\mathbf T^3\times\mathbf R^3}\!h(z)f(t,z)dz\right)^2&\,,
\end{aligned}
$$
tandis que le point (2) du théorème de Lanford et une borne sur $F_2^\eps$ sur laquelle on reviendra plus tard (voir la Proposition \ref{P-BorneCum} et le paragraphe qui la suit) montrent que
$$
\int_{(\mathbf T^3\times\mathbf R^3)^2}\!h(z_1)h(z_2)F^\eps_2(t,Z_2)dZ_2\to\left(\int_{\mathbf T^3\times\mathbf R^3}\!h(z)f(t,z)dz\right)^2
$$
lorsque $\eps\to 0^+$, d'où le résultat annoncé.
\end{proof}

\smallskip
Rappelons la loi (faible et forte) des grands nombres dans l'énoncé suivant.

\begin{theo}[Loi des grands nombres]\label{T-LLN}
Soit $(Y_n)_{n\ge 1}$, suite de variables aléatoires indépendantes et identiquement distribuées, telles que $\mathbb E(|Y_1|)<+\infty$. Alors 
$$
\tfrac1n(Y_1+\ldots+Y_n)\to \mathbb E(Y_1)\quad\text{ en probabilité et presque sûrement lorsque }n\to+\infty\,.
$$
\end{theo}

Soit $\mathcal Y:=(\mathbf T^3\times\mathbf R^3)^{\mathbf N^*}$ muni de la mesure de probabilité borélienne définie par
$$
\text{Prob}(\{\mathbf z\in \mathcal Y\text{ t.q. }(z_1,\ldots,z_n)\in A_1\times\ldots\times A_n\})=\prod_{i=1}^n\int_{A_i}w(z)dz\,,
$$
où $w$ est une densité de probabilité borélienne sur $\mathbf T^3\times\mathbf R^3$. Posons $Y_n(\mathbf z):=h(z_n)$; on vérifie sans peine que les variables aléatoires $Y_n$ sont indépendantes et de même 
distribution, puisque
$$
\text{Prob}(\{\mathbf z\in\mathcal Y\text{ t.q. }Y_n(\mathbf z)>y\})=\int_{\mathbf T^3\times\mathbf R^3}\mathbf 1_{h(z)>y}w(z)dz\,.
$$
Enfin $\mathbb E(|Y_1|)\le\|h\|_{L^\infty(\mathbf T^3\times\mathbf R^3)}<+\infty$. D'après la loi faible des grands nombres
$$
\lim_{n\to+\infty}\text{Prob}\left(\left\{\mathbf z\in \mathcal Y\text{ t.q. }\left|\left\langle\tfrac1n\sum_{i=1}^n\delta_{z_i},h\right\rangle-\int_{\mathbf T^3\times\mathbf R^3}h(z)w(z)dz\right|>\eta\right\}\right)=0\,.
$$

Cet exemple évoque évidemment la situation décrite dans le Corollaire \ref{C-Lanford}. Pour autant, il s'agit une analogie plutôt que d'une véritable application de la loi faible des grands nombres telle qu'énoncée 
dans le Théorème \ref{T-LLN}. En effet

\noindent
(a) les positions initiales des $N$ particules ne peuvent en aucun cas être considérées comme indépendantes pour $\eps>0$ à cause du facteur d'exclusion $\prod_{1\le i<j\le N}\mathbf 1_{\text{dist}(x_i,x_j)>\eps}$
dans
$$
\mathbb F^{in}_N(Z_N):=\prod_{1\le i<j\le N}\mathbf 1_{\text{dist}(x_i,x_j)>\eps}(f^{in})^{\otimes N}(Z_N)\,;
$$
(b) au fur et à mesure de l'évolution du gaz, le nombre de collisions binaires entre molécules augmente, de sorte que les positions et les vitesses de ces $N$ molécules deviennent {\og  de moins en moins 
indépendantes\fg}.

\smallskip
Néanmoins, il est utile de penser au théorème de Lanford comme à une sorte de loi des grands nombres appliquée à la suite de variables aléatoires $(h(z_n(t)))_{n\ge 1}$, où $z_n(t)\!=\!(x_n(t),v_n(t))$ est 
le couple position-vitesse de la $n$-ième molécule à l'instant $t$.


\section{Cumulants et équation de Hamilton--Jacobi}\label{S-HJ}


Pour aller plus loin, deux voies se présentent:

\noindent
(1) étendre l'intervalle de temps $[0,T[C_0,\beta_0]]$ sur lequel le théorème de Lanford permet de déduire l'équation de Boltzmann \eqref{BoltzEq} de la dynamique moléculaire \eqref{Newton}-\eqref{Collxk}-\eqref{Collvkvj};

\noindent
(2) puisque le théorème de Lanford peut s'interpréter comme une loi des grands nombres, transplanter dans le cadre de la limite de Boltzmann--Grad certains des théorèmes limites de la théorie des probabilités 
précisant la loi des grands nombres.

\smallskip
La démarche (1) présenterait un intérêt considérable. On sait que $T[C_0,\beta_0]$ est de l'ordre d'une fraction ($1/5$ pour être précis) du laps de temps moyen entre deux collisions successives subies par 
une même molécule prise au hasard dans le gaz: voir \parencite{LanfordAst} p. 117, Remarque 3 p. 132. Comme expliqué dans \parencite{FG2014} à la fin de la section 4.1, le temps de validité du théorème 
de Lanford a été étendu à $+\infty$ par \textcite{IllnerPulvi} dans le cas d'un domaine spatial $\mathbf R^3$ au lieu de $\mathbf T^3$, et pour des données initiales correspondant à un degré de raréfaction 
du gaz tel que le libre parcours moyen est très grand, et augmente même au cours de l'évolution grâce à l'effet de dispersion dû à l'opérateur de transport libre $\partial_t+v\cdot\nabla_x$. À ce jour, l'extension 
du temps de validité du théorème de Lanford pour une classe assez générale de données initiales, même au temps d'existence d'une solution classique maximale de l'équation de Boltzmann, demeure un 
problème ouvert, peut-être impossible à résoudre avec les méthodes de démonstrations employées jusqu'ici.

C'est pourquoi une partie importante des travaux de Bodineau, Gallagher, Saint-Raymond et Simonella adoptent la démarche (2). On verra toutefois que cela leur permet d'apporter quelques réponses partielles 
--- mais d'un grand intérêt --- au problème (1).

\subsection{Grandes déviations et théorème central limite}


Revenons au cadre idéal du Théorème \ref{T-LLN}. La loi des grands nombres nous dit que la moyenne {\og  empirique\fg}   des variables aléatoires $Y_1,\ldots,Y_n$ tend vers leur espérance mathématique commune 
$\mathbb E(Y_1)$. Une question naturelle consiste donc à étudier les fluctuations de la moyenne empirique autour de sa limite, c'est-à-dire 
$$
\sqrt{n}\left(\frac{Y_1+\ldots+Y_n}n-\mathbb E(Y_1)\right)=\tfrac1{\sqrt{n}}\sum_{i=1}^n(Y_i-\mathbb E(Y_i))\,.
$$
Le choix de dilater les fluctuations par le facteur $\sqrt{n}$ est justifié par le théorème central limite, rappelé ci-dessous.

\begin{theo}[Théorème central limite]
Soit $(Y_n)_{n\ge 1}$, suite de variables aléatoires réelles indépendantes et identiquement distribuées, telles que $\mathbb E(Y_1^2)<+\infty$. Alors
$$
\tfrac1{\sqrt{n}}\sum_{i=1}^n(Y_i-\mathbb E(Y_i))\to\mathcal N(0,\sigma^2)\quad\text{ en loi lorsque }n\to+\infty\,,
$$
où $\mathcal N(0,\sigma^2)$ est la loi normale centrée de variance $\sigma^2:=\mathbb E(Y_1^2)-\mathbb E(Y_1)^2$, de densité gaussienne $e^{-y^2/2\sigma^2}/\sigma\sqrt{2\pi}$ par rapport à la mesure de 
Lebesgue sur $\mathbf R$.
\end{theo}

\smallskip
(Voir par exemple le chapitre IV, section 4.3 de \parencite{Malliav}, où ce résultat porte le nom de {\og  Théorème de Laplace\fg}  ).

\smallskip
Comme on l'a dit plus haut, le problème de la limite de Boltzmann--Grad ne peut pas se réduire à l'étude d'une suite de variables aléatoires indépendantes. Toutefois, le théorème central limite suggère d'étudier 
les fluctuations de la mesure empirique autour de la fonction $F^\eps_1$, dont le théorème de Lanford montre qu'elle converge vers la solution de l'équation de Boltzmann. Spécifiquement, on considèrera
\begin{equation}\label{DefFluct}
\zeta^\eps_t:=\sqrt{\mu_\eps}(\rho^\eps_t-F^\eps_1(t,\cdot))\,,\qquad t\ge 0\,,\quad\eps>0\,,
\end{equation}
que l'on peut réécrire sous la forme
\begin{equation}\label{DefFluct2}
\langle\zeta^\eps_t,h\rangle:=\sqrt{\mu_\eps}\left(\langle\rho^\eps_t,h\rangle-\mathbb E_\eps(\langle\rho^\eps_t,h\rangle)\right)\,,\qquad t\ge 0\,,\quad\eps>0\,,
\end{equation}
\smallskip
pour toute fonction $h\in C_b(\mathbf T^3\times\mathbf R^3)$. D'une certaine manière, la fluctuation $\zeta^\eps_t$ décrit une correction d'ordre supérieur $1/\sqrt{\mu_\eps}=\eps$ à la limite $f$ décrite par 
le théorème de Lanford, laquelle évolue suivant l'équation de Boltzmann.

\smallskip
Revenons encore au cadre du Théorème \ref{T-LLN}. Comme le théorème central limite porte sur la convergence en loi des fluctuations de la moyenne empirique des variables aléatoires $Y_1,\ldots,Y_n$ 
autour de $\mathbb E(Y_1)$, on va étudier $W_n$, la loi de cette mesure empirique $\tfrac1n(Y_1+\ldots+Y_n)$. Notons $w$ la loi commune aux variables aléatoires $Y_i$ pour $i\ge 1$. L'indépendance 
des $Y_i$ entraîne que $W_n$ est la mesure image du produit de convolution à $n$ termes $w\star\ldots\star w$ par l'homothétie de rapport $1/n$ (voir par exemple le Corollaire 4.2.2 de \parencite{Malliav}).

\begin{theo}[Théorème de Cramér]
Soit $(Y_n)_{n\ge 1}$ une suite de variables aléatoires réelles indépendantes et identiquement distribuées de loi $w$, telles que, pour tout $\tau\in\mathbf R$, l'on ait $\mathbb E(\exp(\tau Y_1))<+\infty$. 
Pour tout entier $n\ge 1$, notons $W_n$ la loi de la moyenne empirique $\frac1n(Y_1+\ldots+Y_n)$. Alors, pour tout ouvert $\mathcal O$ et tout fermé $\mathcal F$ de $\mathbf R$
$$
\varlimsup_{n\to+\infty}\tfrac1n\ln W_n(\mathcal F)\le-\inf_{y\in\mathcal F}K(y)\quad\text{ et }\quad\varliminf_{n\to+\infty}\tfrac1n\ln W_n(\mathcal O)\ge-\inf_{y\in\mathcal O}K(y)\,,
$$
où $K$ est la transformée de Legendre du logarithme de la transformée de Laplace de $w$, à savoir
$$
K(\theta):=\sup_{\tau\in\mathbf R}(\theta\tau-\ln\mathbb E(\exp(\tau Y_1)))=\sup_{\tau\in\mathbf R}\left(\theta\tau-\ln\int_\mathbf R e^{\tau y}w(dy)\right)\,.
$$
\end{theo}

\smallskip
Le théorème de Cramér est l'un des énoncés fondamentaux de la théorie des {\og  grandes déviations\fg}   (voir par exemple \parencite{Varadhan} pour plus de détails, ainsi que pour une preuve de ce théorème). Soit 
$\bar Y:=\mathbb E(Y_1)$; l'inégalité de Jensen montre que 
$$
\ln\mathbb E(\exp(\tau Y_1))\ge\tau\bar Y\quad\text{ pour tout }\tau\in\mathbf R\,,
$$
de sorte que $K(\bar Y)\le 0$. D'autre part la fonction $\tau\mapsto\theta\tau-\ln\mathbb E(\exp(\tau Y_1))$ s'annule pour $\tau=0$, de sorte que $K(\theta)\ge 0$ pour tout $\theta\in\mathbf R$ par définition, 
si bien que 
$$
K(\bar Y)=0=\min_{\theta\in\mathbf R}K(\theta)\,.
$$
Comme la fonction $K$ est convexe par construction, elle décroît sur $]-\infty,\bar Y]$ et croît sur $[\bar Y,+\infty[$. Ainsi
$$
\begin{aligned}
a>\bar Y\implies\lim_{n\to+\infty}\tfrac1n\ln W_n(]a,+\infty[)=\lim_{n\to+\infty}\tfrac1n\ln W_n([a,+\infty[)=-K(a)\,,
\\
a<\bar Y\implies\lim_{n\to+\infty}\tfrac1n\ln W_n(]\!-\!\infty,a[)=\!\lim_{n\to+\infty}\tfrac1n\ln W_n(]\!-\!\infty,a])=-K(a)\,.
\end{aligned}
$$
Autrement dit, la probabilité que la moyenne empirique s'écarte de $\bar Y$ d'une distance $\eta>0$ décroît exponentiellement lorsque $n\to+\infty$ dès que $K(a+\eta)+K(a-\eta)>0$. 

Le théorème de Cramér apporte donc une information quant à la précision de l'approximation de la moyenne empirique par l'espérance dans la loi des grands nombres.

Bien que le problème de la limite de Boltzmann--Grad ne puisse pas se réduire à l'étude d'une suite de variables aléatoires indépendantes comme expliqué plus haut, le théorème de Cramér suggère donc 
cependant de considérer la quantité
\begin{equation}\label{DefK0}
\mathcal K_\eps(t,h):=\frac1{\mu_\eps}\ln\mathbb E_\eps(\exp(\mu_\eps\langle\rho^\eps_t,h\rangle))\,,\quad h\in C_b(\mathbf T^3\times\mathbf R^3)\,,\,\,t\ge 0\,.
\end{equation}

\subsection{Cumulants}


La famille des cumulants d'ordre $k$ est définie à partir de la famille des fonctions de corrélations $F^\eps_k$ comme suit:
\begin{equation}\label{DefCumul}
f^\eps_k:=\mu_\eps^{k-1}\sum_{j=1}^k\sum_{\sigma\in\mathcal P_k^j}(-1)^{j-1}(j-1)!F^\eps_\sigma\,,
\end{equation}
où $\mathcal P_k^j$ désigne l'ensemble des partitions $\sigma:=\{\sigma_1,\ldots,\sigma_j\}$ de $\{1,\ldots,k\}$ en $j$ sous-ensembles, et où
$$
F^\eps_\sigma(t,Z_k):=\prod_{i=1}^jF^\eps_{\sigma_i}(t,Z_{\sigma_i})\quad\text{ avec }F^\eps_{\sigma_i}(t,Z_{\sigma_i}):=F^\eps_{|\sigma_i|}(t,z_{l_1},\ldots,z_{l_{|\sigma_i|}})\,,
$$
les indices $l_1,\ldots,l_{|\sigma_i|}$ étant définis par l'égalité $\sigma_i=\{l_1,\ldots,l_{|\sigma_i|}\}$. On vérifie que
$$
f^\eps_1=F^\eps_1\,,\quad f^\eps_2(t,\cdot)=\mu_\eps(F^\eps_2(t,\cdot)-F^\eps_1(t,\cdot)\otimes F^\eps_1(t,\cdot))
$$
puis que
$$
\begin{aligned}
f^\eps_3(t,\cdot)=&\mu_\eps^2(F^\eps_3(t,\cdot)-F^\eps_{\{1,2\}}(t,\cdot)F^\eps_{\{3\}}(t,\cdot)-F^\eps_{\{2.3\}}(t,\cdot)F^\eps_{\{1\}}(t,\cdot)
\\
&-F^\eps_{\{1,3\}}(t,\cdot)F^\eps_{\{2\}}(t,\cdot)+2F^\eps_1(t,\cdot)\otimes F^\eps_1(t,\cdot)\otimes F^\eps_1(t,\cdot))\,.
\end{aligned}
$$

Le second cumulant $f^\eps_2$ est évidemment crucial pour la limite de Boltzmann--Grad, puisque l'hypothèse du chaos moléculaire de Boltzmann, qui est la clé du raisonnement permettant de déduire l'équation 
de Boltzmann comme approximation de l'identité \eqref{CorrEq1}  reliant $F^\eps_1$ à $F^\eps_2$, s'écrit simplement
$$
\lim_{\eps\to 0^+}\frac1{\mu_\eps}f^\eps_2(t,x,v_1,x+\eps\omega,v_2)=0\quad\text{ lorsque }|\omega|=1\text{ et }(v_2-v_1)\cdot\omega<0\,.
$$

Un travail important de \textcite{PulviSim} montre d'ailleurs comment des estimations sur les cumulants permettent de quantifier l'erreur correspondant à l'hypothèse de chaos moléculaire, et utilise ces estimations 
pour préciser la limite de Boltzmann--Grad (voir les Théorèmes 2.4 et 2.5 de \parencite{PulviSim}).

\smallskip
Dans ce qui va suivre, la notion de cumulant va être utilisée de manière différente. Le lemme ci-dessous permet d'abord de montrer que la notion de cumulant est reliée à la fonctionnelle $\mathcal K^\eps$ 
dont l'étude est suggérée par le théorème de Cramér.

\begin{lemm}\label{L-FnGenCum}
Pour tout $\eps>0$, tout $t\ge 0$ et tout $h\in C_b(\mathbf T^3\times\mathbf R^3)$, on a
$$
\mathcal K^\eps(t,h)=\sum_{k\ge 1}\frac{1}{k!}\left\langle f^\eps_k(t,\cdot),(e^h-1)^{\otimes k}\right\rangle\,.
$$
\end{lemm}

\begin{proof}
D'abord
$$
\begin{aligned}
\mathbb E_\eps(\exp(\mu_\eps\langle\rho^\eps_t,h\rangle))-1=&\sum_{m\ge 1}\frac{\mu_\eps^m}{m!}\mathbb E_\eps(\langle\rho^\eps_t,h\rangle^m)
\\
=&\sum_{m\ge 1}\frac1{m!}\mathbb E_\eps\left(\sum_{j_1,\ldots,j_m=1}^Nh(z_{j_1}(t))\ldots h(z_{j_m}(t))\right)
\\
=&\sum_{m\ge 1}\frac1{m!}\sum_{n=1}^m\mu_\eps^n\sum_{\sigma\in\mathcal P_m^n}\int_{(\mathbf T^3\times\mathbf R^3)^n}\prod_{j=1}^nh(z_j)^{|\sigma_j|}F^\eps_n(t,Z_n)dZ_n\,.
\end{aligned}
$$
Or le nombre de partitions de $\{1,\ldots,m\}$ en $n$ ensembles à $\nu_1,\ldots,\nu_n$ éléments vaut
$$
\frac1{n!}{m\choose n_1}{m-\nu_1\choose \nu_2}\ldots{m-\nu_1-\ldots-\nu_{p-2}\choose \nu_{p-1}}=\frac1{n!}\frac{m!}{\nu_1!\ldots \nu_n!}\,,
$$
de sorte que
$$
\sum_{m\ge n}\frac{1}{m!}\sum_{\sigma\in\mathcal P_m^n}\prod_{j=1}^nh(z_j)^{|\sigma_j|}=\frac1{n!}\prod_{j=1}^n\sum_{\nu_j\ge 1}\frac{h(z_j)^{\nu_j}}{\nu_j!}=\frac1{n!}\prod_{j=1}^n(e^{h(z_j)}-1)\,.
$$
En échangeant l'ordre des sommations en $m$ et en $n$, on trouve donc que
$$
\begin{aligned}
\mathbb E_\eps(\exp(\mu_\eps\langle\rho^\eps_t,h\rangle))-1
=&\sum_{n\ge 1}\mu_\eps^n\int_{(\mathbf T^3\times\mathbf R^3)^n}\sum_{m\ge n}\frac1{m!}\sum_{\sigma\in\mathcal P_m^n}\prod_{j=1}^nh(z_j)^{|\sigma_j|}F^\eps_n(t,Z_n)dZ_n
\\
=&\sum_{n\ge 1}\frac{\mu_\eps^n}{n!}\langle F^\eps_n(t,\cdot),(e^h-1)^{\otimes n}\rangle\,.
\end{aligned}
$$
Puis
$$
\begin{aligned}
\ln\mathbb E_\eps(\exp(\mu_\eps\langle\rho^\eps_t,h\rangle))=&\sum_{p\ge 1}\frac{(-1)^{p-1}}{p}\left(\sum_{n\ge 1}\frac{\mu_\eps^n}{n!}\langle F^\eps_n(t,\cdot),(e^h-1)^{\otimes n}\rangle\right)^p
\\
=&\sum_{p\ge 1}\frac{(-1)^{p-1}}{p}\sum_{n_1,\ldots,n_p\ge 1}\frac{\mu_\eps^{n_1+\ldots+n_p}}{n_1!\ldots n_p!}\prod_{k=1}^p\langle F^\eps_{n_k},(e^h-1)^{\otimes n_k}\rangle\,.
\end{aligned}
$$
Grâce au dénombrement des partitions à $p$ éléments de $\{1,\ldots,n\}$ rappelé ci-dessus 
$$
\begin{aligned}
\ln\mathbb E_\eps(\exp(\mu_\eps\langle\rho^\eps_t,h\rangle))=&\sum_{n\ge 1}\frac{\mu_\eps^n}{n!}\sum_{p=1}^n\frac{(-1)^{p-1}}{p}p!\sum_{\sigma\in\mathcal P_n^p}\prod_{k=1}^p\langle F^\eps_{|\sigma_k|},(e^h-1)^{\otimes |\sigma_k|}\rangle
\\
=&\mu_\eps\sum_{n\ge 1}\frac1{n!}\langle f^\eps_k,(e^h-1)^{\otimes k}\rangle\,,
\end{aligned}
$$
d'où le résultat annoncé.
\end{proof}

\smallskip
Le Lemme \ref{L-FnGenCum} explique donc pourquoi la fonctionnelle $\mathcal K^\eps$ est appelée {\og  fonction génératrice des cumulants\fg}.

\subsection{Equation de Hamilton--Jacobi fonctionnelle: approche formelle}


L'un des résultats principaux de \textcite{BoGalSRSim1} est que, dans la limite de Boltzmann--Grad, c'est-à-dire lorsque $\eps\to 0^+$, la fonction génératrice des cumulants est solution d'une équation de 
Hamilton--Jacobi {\og  fonctionnelle\fg}. Le but de cette section est de présenter un calcul purement formel, inspiré de la section 3.3 de \parencite{BGSRSicm} aboutissant à cette équation. Les énoncés rigoureux 
feront l'objet de la section suivante.

\subsubsection{Dérivées fonctionnelles de $\mathcal K^\eps(t,\cdot)$ d'ordre un et deux}


Soient trois fonctions test $h,\phi,\psi\in C_b(\mathbf T^3\times\mathbf R^3)$. On commence par calculer
$$
\begin{aligned}
\left\langle\frac{\partial\mathcal K^\eps}{\partial h}(t,h),\phi\right\rangle=\lim_{\eta\to 0}\frac{\mathcal K^\eps(t,h+\eta\phi)-\mathcal K^\eps(t,h)}{\eta}
\\
=\frac1{\mu_\eps\mathbb E_\eps(\exp(\mu_\eps\langle\rho^\eps_t,h\rangle))}
\lim_{\eta\to 0}\mathbb E_\eps\left(\frac{\exp(\mu_\eps\langle\rho^\eps_t,h+\eta\phi\rangle-\exp(\mu_\eps\langle\rho^\eps_t,h\rangle))}{\eta}\right)
\\
=\frac{\mathbb E_\eps(\langle\rho^\eps_t,\phi\rangle\exp(\mu_\eps\langle\rho^\eps_t,h\rangle))}{\mathbb E_\eps(\exp(\mu_\eps\langle\rho^\eps_t,h\rangle))}&\,.
\end{aligned}
$$
Passons maintenant au calcul de la dérivée seconde:
$$
\begin{aligned}
\frac{\partial^2\mathcal K^\eps}{\partial h^2}(t,h)\cdot(\phi,\psi)
=\lim_{\eta\to 0}\frac1\eta\left(\left\langle\frac{\partial\mathcal K^\eps}{\partial h}(t,h+\eta\phi),\psi\right\rangle-\left\langle\frac{\partial\mathcal K^\eps}{\partial h}(t,h),\psi\right\rangle\right)
\\
=\lim_{\eta\to 0}\frac1\eta\left(\frac{\mathbb E_\eps(\langle\rho^\eps_t,\psi\rangle\exp(\mu_\eps\langle\rho^\eps_t,h+\eta\phi\rangle))}{\mathbb E_\eps(\exp(\mu_\eps\langle\rho^\eps_t,h+\eta\phi\rangle))}
-\frac{\mathbb E_\eps(\langle\rho^\eps_t,\psi\rangle\exp(\mu_\eps\langle\rho^\eps_t,h\rangle))}{\mathbb E_\eps(\exp(\mu_\eps\langle\rho^\eps_t,h\rangle))} \right)
\\
=\lim_{\eta\to 0}\frac1\eta\left(\frac{\mathbb E_\eps(\langle\rho^\eps_t,\psi\rangle\exp(\mu_\eps\langle\rho^\eps_t,h+\eta\phi\rangle))
	-\mathbb E_\eps(\langle\rho^\eps_t,\psi\rangle\exp(\mu_\eps\langle\rho^\eps_t,h\rangle))}{\mathbb E_\eps(\exp(\mu_\eps\langle\rho^\eps_t,h+\eta\phi\rangle))}\right)
\\
+\mathbb E_\eps(\langle\rho^\eps_t,\psi\rangle\exp(\mu_\eps\langle\rho^\eps_t,h\rangle))\lim_{\eta\to 0}\frac1\eta\left(\frac1{\mathbb E_\eps(\exp(\mu_\eps\langle\rho^\eps_t,h+\eta\phi\rangle))}
-\frac1{\mathbb E_\eps(\exp(\mu_\eps\langle\rho^\eps_t,h\rangle))}\right)
\\
=\mu_\eps\frac{\mathbb E_\eps(\langle\rho^\eps_t,\phi\rangle\langle\rho^\eps_t,\psi\rangle\exp(\mu_\eps\langle\rho^\eps_t,h\rangle))}{\mathbb E_\eps(\exp(\mu_\eps\langle\rho^\eps_t,h\rangle))}
\\
-\mu_\eps\frac{\mathbb E_\eps(\langle\rho^\eps_t,\psi\rangle\exp(\mu_\eps\langle\rho^\eps_t,h\rangle))\mathbb E_\eps(\langle\rho^\eps_t,\phi\rangle\exp(\mu_\eps\langle\rho^\eps_t,h\rangle))}
{\mathbb E_\eps(\exp(\mu_\eps\langle\rho^\eps_t,h\rangle))^2}\,,
\end{aligned}
$$
ce que l'on peut réécrire sous la forme
$$
\begin{aligned}
\frac1{\mu_\eps}\frac{\partial^2\mathcal K^\eps}{\partial h^2}(t,h)\cdot(\phi,\psi)
=&\frac{\mathbb E_\eps(\langle\rho^\eps_t,\phi\rangle\langle\rho^\eps_t,\psi\rangle\exp(\mu_\eps\langle\rho^\eps_t,h\rangle))}{\mathbb E_\eps(\exp(\mu_\eps\langle\rho^\eps_t,h\rangle))}
\\
&-\left\langle\frac{\partial\mathcal K^\eps}{\partial h}(t,h),\phi\right\rangle\left\langle\frac{\partial\mathcal K^\eps}{\partial h}(t,h),\psi\right\rangle\,,
\end{aligned}
$$
ou encore, de manière plus condensée,
$$
\frac1{\mu_\eps}\frac{\partial^2\mathcal K^\eps}{\partial h^2}(t,h)+\frac{\partial\mathcal K^\eps}{\partial h}(t,h)\otimes\frac{\partial\mathcal K^\eps}{\partial h}(t,h)
=\frac{\mathbb E_\eps(\exp(\mu_\eps\langle\rho^\eps_t,h\rangle)\langle\rho^\eps_t\otimes\rho^\eps_t,\cdot\rangle)}{\mathbb E_\eps(\exp(\mu_\eps\langle\rho^\eps_t,h\rangle))}\,.
$$

\subsubsection{Évolution de $\mathcal K^\eps(t,\cdot)$}


Compte-tenu de la définition \eqref{DefK0} de $\mathcal K^\eps$ (qui est un logarithme), il sera commode de calculer
$$
\begin{aligned}
\exp(\mu_\eps\mathcal K^\eps(t,h))\mu_\eps\partial_t\mathcal K^\eps(t,h)=\partial_t\exp(\mu_\eps\mathcal K^\eps(t,h))=\partial_t\mathbb E_\eps(\exp\langle\mu_\eps\rho^\eps_t,h\rangle)
\\
=\frac1{\mathcal Z_\eps}\sum_{n\ge 1}\frac{\mu_\eps^n}{n!}\partial_t\int_{(\mathbf T^3\times\mathbf R^3)^n}\exp\left(\sum_{i=1}^nh(z_i(t))\right)\mathbb F^{in}_n(Z_n)dZ_n
\\
=\frac1{\mathcal Z_\eps}\sum_{n\ge 1}\frac{\mu_\eps^n}{n!}\int_{(\mathbf T^3\times\mathbf R^3)^n}\exp\left(\sum_{i=1}^nh(z_i)\right)\partial_t\mathbb F_n(t,Z_n)dZ_n&\,.
\end{aligned}
$$
On supposera d'autre part que la fonction test $h$ est suffisamment régulière --- par exemple, on pourra évidemment se restreindre au cas où $h\in C^\infty_c(\mathbf T^3\times\mathbf R^3)$, 
mais on verra plus loin qu'il n'est pas nécessaire que $h$ soit indéfiniment différentiable.

On exprime alors $\partial_t\mathbb F_n(t,\cdot)$ au moyen de l'équation de Liouville écrite au sens des distributions sous la forme \eqref{LiouvilleDistrib2}:
$$
\begin{aligned}
\partial_t\mathbb F_n(t,Z_n)=-\sum_{i=1}^nv_i\cdot\nabla_{x_i}\mathbb F_n(t,Z_n)
\\
+\sum_{1\le i<j\le n}\mathbb F_n(t, \hat Z_n[i,j])\big|_{\partial^+\Gamma^{\eps}_n}((v_j-v_i)\cdot n_{ij})_+\delta_{\text{dist}(x_i,x_j)=\eps}
\\
\sum_{1\le i<j\le n}\mathbb F_n(t,Z_n)\big|_{\partial^+\Gamma^{\eps}_n}((v_j-v_i)\cdot n_{ij})_-\delta_{\text{dist}(x_i,x_j)=\eps}&\,.
\end{aligned}
$$
Donc
$$
\begin{aligned}
\exp(\mu_\eps\mathcal K^\eps(t,h))\mu_\eps\partial_t\mathcal K^\eps(t,h)=-\frac1{\mathcal Z_\eps}\sum_{n\ge 1}\frac{\mu_\eps^n}{n!}
\left\langle\sum_{i=1}^nv_i\cdot\nabla_{x_i}\mathbb F_n(t,\cdot),\exp\left(\sum_{i=1}^nh(z_i)\right)\right\rangle
\\
+\frac1{\mathcal Z_\eps}\sum_{n\ge 1}\frac{\mu_\eps^n}{n!}\left\langle\sum_{1\le j<k\le n}\mathbb F_n(t,\hat Z_n[j,k])\big|_{\partial^+\Gamma^{\eps}_n}
((v_k-v_j)\cdot n_{jk})_+\delta_{\text{dist}(x_j,x_k)=\eps},e^{\sum_{i=1}^nh(z_i)}\right\rangle
\\
-\frac1{\mathcal Z_\eps}\sum_{n\ge 1}\frac{\mu_\eps^n}{n!}\left\langle\sum_{1\le j<k\le n}\mathbb F_n(t,\cdot)\big|_{\partial^+\Gamma^{\eps}_n}
((v_k-v_j)\cdot n_{jk})_-\delta_{\text{dist}(x_j,x_k)=\eps},e^{\sum_{i=1}^nh(z_i)}\right\rangle
\\
=:T_1+T_2-T_3&.
\end{aligned}
$$
Le terme $T_1$ se réécrit comme suit:
$$
\begin{aligned}
T_1=&\frac1{\mathcal Z_\eps}\sum_{n\ge 1}\frac{\mu_\eps^n}{n!}\left\langle\mathbb F_n(t,\cdot),\sum_{i=1}^nv_i\cdot\nabla_{x_i}\exp\left(\sum_{i=1}^nh(z_i)\right)\right\rangle
\\
=&\frac1{\mathcal Z_\eps}\sum_{n\ge 1}\frac{\mu_\eps^n}{n!}\left\langle\mathbb F_n(t,\cdot),\exp\left(\sum_{i=1}^nh(z_i)\right)\sum_{i=1}^nv_i\cdot\nabla_xh(z_i)\right\rangle
\\
=&\mathbb E_\eps(\exp(\langle\mu_\eps\rho^\eps_t,h\rangle)\langle\mu_\eps\rho^\eps_t,v\cdot\nabla_xh\rangle)
\\
=&\mathbb E_\eps(\exp(\mu_\eps\langle\rho^\eps_t,h\rangle))\mu_\eps\left\langle\frac{\partial\mathcal K^\eps}{\partial h}(t,h),v\cdot\nabla_xh\right\rangle
\\
=&\exp(\mu_\eps\mathcal K^\eps(t,h))\mu_\eps\left\langle\frac{\partial\mathcal K^\eps}{\partial h}(t,h),v\cdot\nabla_xh\right\rangle\,.
\end{aligned}
$$
Passons au terme $T_3$:
$$
\begin{aligned}
T_3=&\frac1{\mathcal Z_\eps}\sum_{n\ge 1}\frac{\mu_\eps^n}{n!}\left\langle\sum_{1\le j<k\le n}\mathbb F_n(t,\cdot)\big|_{\partial^+\Gamma^{\eps}_n}
((v_k-v_j)\cdot n_{jk})_-\delta_{\text{dist}(x_j,x_k)=\eps},e^{\sum_{i=1}^nh(z_i)}\right\rangle
\\
=&\tfrac12\mu_\eps^2\mathbb E_\eps(\exp(\langle\mu_\eps\rho^\eps_t,h\rangle)\langle\rho^\eps_t\otimes\rho^\eps_t,((v-v_*)\cdot\tfrac{x-x_*}\eps)_+\delta_{\text{dist}(x,x_*)=\eps+0}\rangle)\,.
\end{aligned}
$$
Quant au terme $T_2$, il vient:
$$
\begin{aligned}
T_2=&\frac1{\mathcal Z_\eps}\sum_{n\ge 1}\frac{\mu_\eps^n}{n!}\left\langle\sum_{1\le j<k\le n}\!\!\mathbb F_n(t,\hat Z_n[j,k])\big|_{\partial^+\Gamma^{\eps}_n}
\!((v_k\!-\!v_j)\!\cdot\! n_{jk})_+\delta_{\text{dist}(x_j,x_k)=\eps},e^{\sum_{i=1}^nh(z_i)}\right\rangle
\\
=&\frac1{2\mathcal Z_\eps}\sum_{n\ge 1}\frac{\mu_\eps^n}{n!}\left\langle\sum_{j,k=1}^n\mathbb F_n(t,\cdot)
((v'_k\!-\!v'_j)\!\cdot\! n_{jk})_+\delta_{\text{dist}(x_j,x_k)=\eps+0},e^{\mathbf Dh(z_j,z_k,n_{kj})\!+\!\sum_{i=1}^n\!h(z_i)}\right\rangle
\\
=&\tfrac12\mu_\eps^2\mathbb E_\eps(\exp(\langle\mu_\eps\rho^\eps_t,h\rangle)\langle\rho^\eps_t\otimes\rho^\eps_t,e^{\mathbf Dh(z,z_*,\frac{x-x_*}\eps)}
((v-v_*)\cdot\tfrac{x-x_*}\eps)_-\delta_{\text{dist}(x,x_*)=\eps+0}\rangle)\,,
\end{aligned}
$$
où on a noté
\begin{equation}\label{DefbD}
\mathbf Dh(z,z_*,\omega):=h(x,v-((v-v_*)\cdot\omega)\omega)+h(x_*,v_*+((v-v_*)\cdot\omega)\omega)-h(z)-h(z_*)\,.
\end{equation}
Posons, en observant que $\mathbf Dh(z,z_*,\omega)=\mathbf Dh(z,z_*,-\omega)$,
\begin{equation}\label{Jeps}
J_\eps[h](z,z_*):=\int_{\mathbf S^2}\left(e^{\mathbf Dh(z,z_*,\omega)}-1\right)((v-v_*)\cdot\omega)_+\delta_{\eps\omega}(x-x_*)d\omega\,.
\end{equation}
Observons que
$$
\begin{aligned}
\int_{\mathbf T^3\times\mathbf T^3}\psi(x,x_*)\delta_{\text{dist}(x,x_*)=\eps}=\int_{\mathbf T^3\times\mathbf S^2}\psi(x,x+\eps\omega)\eps^2dxd\omega
\\
=\int_{\mathbf T^3\times\mathbf T^3}\psi(x,x_*)\left(\eps^2\int_{\mathbf S^2}\delta_{\eps\omega}(x-x_*)d\omega\right)dxdx_*&\,.
\end{aligned}
$$
Donc
$$
\begin{aligned}
T_2-T_3
\\
=\!\tfrac12\mu_\eps^2\mathbb E_\eps(\exp(\langle\mu_\eps\rho^\eps_t,h\rangle))\langle\rho^\eps_t\otimes\rho^\eps_t,(e^{\mathbf Dh(z,z_*,\frac{x-x_*}\eps)}\!-\!1)
((v\!-\!v_*)\cdot\tfrac{x\!-\!x_*}\eps)_-\delta_{\text{dist}(x,x_*)=\eps+0}\rangle)
\\
=\!\tfrac12\mu_\eps^2\mathbb E_\eps(\exp(\langle\mu_\eps\rho^\eps_t,h\rangle))\left(\frac1{\mu_\eps}\frac{\partial^2\mathcal K^\eps}{\partial h^2}(t,h)
+\frac{\partial\mathcal K^\eps}{\partial h}(t,h)\otimes\frac{\partial\mathcal K^\eps}{\partial h}(t,h)\right)
\cdot \eps^2J_\eps[h]
\\
=\tfrac12\exp(\mu_\eps\mathcal K^\eps(t,h))\mu_\eps\left(\frac1{\mu_\eps}\frac{\partial^2\mathcal K^\eps}{\partial h^2}(t,h)
+\frac{\partial\mathcal K^\eps}{\partial h}(t,h)\otimes\frac{\partial\mathcal K^\eps}{\partial h}(t,h)\right)\cdot J_\eps[h]&\,.
\end{aligned}
$$
On trouve donc que
$$
\begin{aligned}
\exp(\mu_\eps\mathcal K^\eps(t,h))\mu_\eps\partial_t\mathcal K^\eps(t,h)=\exp(\mu_\eps\mathcal K^\eps(t,h))\mu_\eps\left\langle\frac{\partial\mathcal K^\eps}{\partial h}(t,h),v\cdot\nabla_xh\right\rangle
\\
+\tfrac12\exp(\mu_\eps\mathcal K^\eps(t,h))\mu_\eps\left(\frac1{\mu_\eps}\frac{\partial^2\mathcal K^\eps}{\partial h^2}(t,h)
+\frac{\partial\mathcal K^\eps}{\partial h}(t,h)\otimes\frac{\partial\mathcal K^\eps}{\partial h}(t,h)\right)\cdot J_\eps[h]\,,
\end{aligned}
$$
c'est-à-dire que
\begin{equation}\label{HJeps}
\begin{aligned}
\partial_t\mathcal K^\eps(t,h)=\left\langle\frac{\partial\mathcal K^\eps}{\partial h}(t,h),v\cdot\nabla_xh\right\rangle
\\
+\tfrac12\left(\frac1{\mu_\eps}\frac{\partial^2\mathcal K^\eps}{\partial h^2}(t,h)+\frac{\partial\mathcal K^\eps}{\partial h}(t,h)\otimes\frac{\partial\mathcal K^\eps}{\partial h}(t,h)\right)\cdot J_\eps[h]&\,.
\end{aligned}
\end{equation}
Passons à la limite formellement dans cette équation, en supposant que $\mathcal K^\eps(t,h)$ converge vers $\mathcal K(t,h)$ lorsque $\eps\to 0^+$: il vient
\begin{equation}\label{HJ1}
\partial_t\mathcal K(t,h)=\left\langle\frac{\partial\mathcal K}{\partial h}(t,h),v\cdot\nabla_xh\right\rangle+\tfrac12\frac{\partial\mathcal K}{\partial h}(t,h)\otimes\frac{\partial\mathcal K}{\partial h}(t,h)\cdot J_0[h]\,,
\end{equation}
avec
\begin{equation}\label{J0}
J_0[h]:=\left(\int_{\mathbf S^2}\left(e^{\mathbf Dh(z,z_*,\omega)}-1\right)((v-v_*)\cdot\omega)_+d\omega\right)\delta_0(x-x_*)\,.
\end{equation}
L'équation vérifiée par $\mathcal K$ est une équation aux dérivées partielles {\og  fonctionnelle\fg}   (puisque la {\og  variable\fg}   $h$ appartient à un espace fonctionnel de dimension infinie), de type Hamilton--Jacobi, 
puisque l'équation obtenue fait intervenir une fonctionnelle non linéaire ne dépendant que de la dérivée première $\partial\mathcal K/\partial h$.

\smallskip
Terminons cette section par une mise en garde: la dérivée seconde ${\partial^2\mathcal K^\eps(t,h)}/{\partial h^2}$ multipliée par le petit paramètre $1/\mu_\eps$ dans \eqref{HJeps} pourrait évoquer
une limite à {\og  viscosité\fg}   petite, et faire de la limite (formelle) de \eqref{HJeps} vers \eqref{HJ1} un résultat analogue à la limite des solution $u_\eps$ de l'équation de Burgers
$$
\partial_tu_\eps+u_\eps\partial_xu_\eps=\eps\partial_x^2u_\eps
$$
vers les solutions {\og  entropiques\fg}   de l'équation de Hopf
$$
\partial_tu+\partial_x\left(\frac{u^2}2\right)=0\,.
$$
Voir par exemple \parencite{Hopf50}. La définition \eqref{DefK0} est d'ailleurs formellement tout à fait analogue à la transformation de Cole-Hopf, formule (3) de \parencite{Hopf50}. Mais le {\og  coefficient\fg}   
$J_\eps[h]/\mu_\eps$ de la dérivée seconde ${\partial^2\mathcal K^\eps}/{\partial h^2}(t,h)$ n'est ni positif, ni même de signe constant, puisque $\mathbf Dh$ prend en général des valeurs positives 
et négatives. En effet, on vérifiera facilement que
$$
\mathbf Dh(x,v,x,v_*,\omega)>0\implies\mathbf Dh(x,v',x,v'_*,-\omega)=-\mathbf Dh(x,v,x,v_*,\omega)<0\,,
$$
par le même argument que celui utilisé dans la section \ref{SS-BoltzMath} pour démontrer la formulation faible de l'intégrale de collision de Boltzmann.

\subsection{L'équation de Hamilton--Jacobi: principaux résultats}


Soient $t,\eps>0$. Dans toute la suite, on notera $\mathbf Z_n([0,t])$ la restriction à l'intervalle $[0,t]$ du chemin $\mathbf R\ni\tau\mapsto Z_n(\tau):=S^{n,\eps}_\tau(Z^{in}_n)\in(\mathbf T^3\times\mathbf R^3)^n$. 
On notera $\mathbf z_j([0,t])$ la $j$-ième composante de $\mathbf Z_n([0,t])$. On notera également $\mathbf D_n([0,t])$ l'espace des chemins càdlàg\footnote{Acronyme de {\og  continu à droite, [avec] limite à gauche\fg}.} définis sur $[0,t]$ à valeurs dans $(\mathbf T^3\times\mathbf R^3)^n$, muni de la topologie de Skorokhod (\parencite{Billing}, chapitre 3, section 12).

Commençons par donner une borne sur la fonction génératrice des cumulants. 

\begin{prop}\label{P-BorneCum}
On suppose que la fonction de distribution initiale $f^{in}$ vérifie la borne \eqref{BndCondIni}. Il existe $C,T_0>0$ tels que, pour tout $h:\,D_1([0,+\infty[)\to\mathbf R$ continue vérifiant
$$
h((x,v)([0,t]))\le\alpha+\tfrac14\beta_0\sup_{0\le s\le t}|v(s)|^2\,,
$$
l'on a
$$
\begin{aligned}
\left|\frac1{\mu_\eps}\ln\mathbb E_\eps\left(\exp\left(\sum_{j=1}^Nh(\mathbf z_j([0,t])\right)\right)\right|\le \frac{CC_0e^\alpha}{\beta_0^2}\sum_{n\ge 1}\left(\frac{CC_0e^\alpha}{\beta_0^2}\right)^{n-1}(t+\eps)^{n-1}
\\
=\frac{CC_0e^\alpha}{\beta_0^2-CC_0e^\alpha(t+\eps)}
\end{aligned}
$$
pour $t+\eps<\min\left(T_0,\frac{\beta_0^2e^{-\alpha}}{CC_0}\right):=T_\alpha[C_0,\beta_0]$.
\end{prop} 

Cette borne supérieure est démontrée dans le chapitre 8 de \parencite{BoGalSRSim1} (Théorème 10, voir aussi le Théorème 4 du chapitre 4 de cette même référence). Elle implique une borne sur $F^\eps_2$ 
utilisée dans la preuve du Corollaire \ref{C-Lanford}.

\smallskip
À partir de là, on va modifier légèrement la définition de la fonctionnelle $\mathcal K^\eps$ de façon à y inclure le terme de transport libre $\left\langle\frac{\partial\mathcal K}{\partial h}(t,h),v\cdot\nabla_xh\right\rangle$. 
Pour cela, on considère des fonctions test $h$ de la forme
$$
h(z([0,t])):=g(t,z(t))-\int_0^t(\partial_t+v\cdot\nabla_x)g(s,z(s))ds\,.
$$
Soit
$$
\begin{aligned}
\mathbf B_\alpha:=\{g\in C^1([0,T_\alpha]\times\mathbf T^3\times\mathbf R^3;\mathbf C)\text{ t.q. }|g(t,z)|\le (1-\tfrac{t}{2T_\alpha})(\alpha+\tfrac18\beta_0|v|^2)
\\
\text{ et }\sup_{0\le t\le T_\alpha}|(\partial_t+v\cdot\nabla_x)g(t,z)|\le\tfrac1{2T_\alpha}(\alpha+\tfrac18\beta_0|v|^2)\}&\,.
\end{aligned}
$$
Évidemment, si $g\in\mathbf B_\alpha$, on a
$$
|h(z([0,t]))|\le|g(t,z(t))|+\int_0^t|(\partial_t+v\cdot\nabla_x)g(s,z(s))|ds\le(\alpha+\tfrac18\beta_0|v|^2)\,,
$$
de sorte que la borne uniforme de la proposition ci-dessus s'applique.

\begin{prop}\label{P-DefbK}
On suppose que la fonction de distribution initiale $f^{in}$ vérifie la borne \eqref{BndCondIni}. Soit $\alpha>0$; pour tout $t\in[0,T_\alpha[$ (où $T_\alpha$ est comme dans la Proposition \ref{P-BorneCum}) et tout $g\in\mathbf B_\alpha$, on pose
$$
\mathbf K(t,g):=\lim_{\eps\to 0^+}\frac1{\mu_\eps}\ln\mathbb E_\eps\left(\exp\left(\sum_{j=1}^Nh(\mathbf z_j([0,t])\right)\right)\,,
$$
où 
$$
h(z([0,t]))=g(t,z(t))-\int_0^t(\partial_t+v\cdot\nabla_x)g(s,z(s))ds\,.
$$
Pour tous $t\in[0,T_\alpha[$ et $g\in\mathbf B_\alpha$, la dérivée fonctionnelle $\partial\mathbf K(t,g)/\partial g(t)$ s'identifie à une fonction continue sur $\mathbf T^3$ à valeurs dans l'espace des mesures de Radon 
sur $\mathbf R^3$, et pour tout $t\in[0,T_\alpha[$, il existe $C(t)>0$ tel que
$$
\sup_{0\le s\le t}\sup_{x\in\mathbf T^3}\left\|(1+|v|)\exp(\tfrac18\beta_0|v|^2)\frac{\partial\mathbf K(s,g)}{\partial g(s)}(x,\cdot)\right\|_{VT}\le C(t)\,.
$$
\end{prop}

\smallskip
Les énoncés ci-dessus nous permettent alors de préciser l'équation de Hamilton--Jacobi obtenue à la section précédente par un calcul formel.

\begin{theo}\label{T-HJ}
On suppose que la fonction de distribution initiale $f^{in}$ vérifie la borne \eqref{BndCondIni}. Soit $\alpha>0$; la fonction $\mathbf K$ définie sur $[0,T_\alpha[\times\mathbf B_\alpha$ est une solution de la forme 
intégrée en $t$ de l'équation de Hamilton--Jacobi
\begin{equation}\label{HJ2}
\left\{
\begin{aligned}
{}&\partial_t\mathbf K(t,g)=\mathcal H\left(\frac{\partial\mathbf K(t,g)}{\partial g(t)},g(t)\right)\,,\qquad (t,g)\in[0,T_\alpha[\times\mathbf B_\alpha\,,
\\
&\mathbf K(0,g)=\int_{\mathbf T^3\times\mathbf R^3}(e^{g(0,z)}-1)f^{in}(z)dz\,,
\end{aligned}
\right.
\end{equation}
où le hamiltonien $\mathcal H$ est défini par la formule suivante:
\begin{equation}\label{DefHamilt}
\begin{aligned}
\mathcal H(p,q):=\tfrac12\int_{\mathbf T^3\times(\mathbf R^3)^2\times\mathbf S^2}\left(e^{\mathbf D q(x,v,x_*v_*,\omega)}-1\right)((v-v_*)\cdot\omega)_+ d\omega\,p(x,dv)p(x,dv_*)dx
\end{aligned}
\end{equation}
pour toute fonction $q\in C(\mathbf T^3\times\mathbf R^3)$ vérifiant
$$
\sup_{x\in\mathbf T^3}|q(x,v)|\le c+\tfrac18\beta_0|v|^2\,,\qquad v\in\mathbf R^3\,,
$$
et tout $p\in C(\mathbf T^3;\mathcal M(\mathbf R^3))$ vérifiant
$$
\sup_{x\in\mathbf T^3}\left\|(1+|v|)\exp(\tfrac18\beta_0|v|^2)\frac{\partial\mathbf K(s,g)}{\partial g(s)}(x,\cdot)\right\|_{VT}<+\infty\,.
$$
On rappelle enfin que la notation $\mathbf D q$ est définie dans \eqref{DefbD}.
\end{theo}

\smallskip
Ce théorème est l'un des résultats majeurs de \textcite{BoGalSRSim1}. Nous verrons plus loin qu'il a des applications importantes, d'une part à l'étude des fluctuations autour de l'équation de Boltzmann, et d'autre part 
à celle des grandes déviations dans la limite de Boltzmann--Grad. On pourrait s'interroger sur l'intérêt de remplacer la dynamique à $N$ corps \eqref{Newton}-\eqref{Collxk}-\eqref{Collvkvj}, posée sur un l'espace des 
phases $(\mathbf T^3\times\mathbf R^3)^N$ de dimension $6N$, grande mais finie, par l'équation \eqref{HJ2}, posée sur un espace de dimension infinie. Comme elle est obtenue après passage à la limite $\eps\to 0^+$, 
on peut toutefois espérer que la dynamique sous-jacente à \eqref{HJ2} soit plus simple que la dynamique à $N$ corps.

\subsection{De l'équation de Hamilton--Jacobi à l'équation de Boltzmann}


Une première observation est que l'équation de Boltzmann peut se déduire de l'équation de Hamilton--Jacobi \eqref{HJ2}.

\begin{prop}\label{P-HJBoltz}
On suppose que la fonction de distribution initiale $f^{in}$ vérifie la borne \eqref{BndCondIni}. Soit
$$
f(t,\cdot):=\frac{\partial\mathbf K(t,g)}{\partial g(t)}\Big|_{g=0}\in C(\mathbf T^3;\mathcal M(\mathbf R^3))\,.
$$
Alors $f$ est solution faible de l'équation de Boltzmann sous forme intégrée en temps avec donnée initiale $f^{in}$.
\end{prop}

Commençons par un calcul formel à partir de l'équation \eqref{HJ1}.

\begin{proof}[Argument formel]
Revenons à la fonctionnelle $\mathcal K$. Dérivons chaque membre de \eqref{HJ1} en $h$ dans la direction $\phi$:
$$
\begin{aligned}
\partial_t\left\langle\frac{\partial\mathcal K}{\partial h}(t,h),\phi\right\rangle=&\frac{\partial^2\mathcal K}{\partial h^2}(t,h)\cdot(\phi,v\cdot\nabla_xh)+\left\langle\frac{\partial\mathcal K}{\partial h}(t,h),v\cdot\nabla_x\phi\right\rangle
\\
&+\tfrac12\frac{\partial^2\mathcal K}{\partial h^2}(t,h)\otimes\frac{\partial\mathcal K}{\partial h}(t,h)\cdot(\phi,J_0[h])
\\
&+\tfrac12\frac{\partial\mathcal K}{\partial h}(t,h)\otimes\frac{\partial^2\mathcal K}{\partial h^2}(t,h)\cdot(J_0[h],\phi)
\\
&+\tfrac12\frac{\partial\mathcal K}{\partial h}(t,h)\otimes\frac{\partial\mathcal K}{\partial h}(t,h)\cdot\left(\frac{dJ_0}{dh}[h]\cdot\phi\right)\,,
\end{aligned}
$$
puis faisons $h=0$ dans cette équation. Comme $J_0[0]=0$, il vient
$$
\begin{aligned}
\partial_t\left\langle\frac{\partial\mathcal K}{\partial h}(t,0),\phi\right\rangle=&\left\langle\frac{\partial\mathcal K}{\partial h}(t,0),v\cdot\nabla_x\phi\right\rangle
\\
&+\tfrac12\frac{\partial\mathcal K}{\partial h}(t,0)\otimes\frac{\partial\mathcal K}{\partial h}(t,0)\cdot\left(\frac{dJ_0}{dh}[0]\cdot\phi\right)\,.
\end{aligned}
$$
Un calcul simple montre que 
$$
\frac{dJ_0}{dh}[0]\cdot\phi=\left(\int_{\mathbf S^2}\mathbf D\phi(z,z_*,\omega)((v-v_*)\cdot\omega)_+d\omega\right)\delta_0(x-x_*)\,.
$$
Posons $m(t,\cdot):=\frac{\partial\mathcal K}{\partial h}(t,0)$, que nous allons manipuler comme une fonction continue sur $\mathbf T^3$ à valeurs dans les mesures de Radon sur $\mathbf R^3$. Alors
$$
\begin{aligned}
\left\langle m(t)\otimes m(t),\frac{dJ_0}{dh}[0]\cdot\phi\right\rangle
\\
=\int_{\mathbf T^3\times(\mathbf R^3)^2\times\mathbf S^2}\mathbf D\phi(x,v,x,v_*,\omega)((v-v_*)\cdot\omega)_+d\omega m(t,x,dv)m(t,x,dv_*)dx&\,.
\end{aligned}
$$
On est arrivé ainsi à l'identité satisfaite par $m(t,\cdot)$ pour toute fonction test $\phi$ dans un espace bien choisi (par exemple $C^\infty_c(\mathbf T^3\times\mathbf R^3)$):
$$
\begin{aligned}
\langle\partial_tm(t),\phi\rangle=\langle m(t),v\cdot\nabla_x\phi\rangle
\\
+\tfrac12\int_{\mathbf T^3\times(\mathbf R^3)^2\times\mathbf S^2}\mathbf D\phi(x,v,x,v_*,\omega)((v-v_*)\cdot\omega)_+d\omega m(t,x,dv)m(t,x,dv_*)dx&\,,
\end{aligned}
$$
qui n'est rien d'autre que la formulation faible de l'équation de Boltzmann \eqref{BoltzEq}.
\end{proof}

Passons à l'argument rigoureux basé sur le Théorème \ref{T-HJ}. Commençons par une observation élémentaire, basée sur l'idée suivante: dans la fonctionnelle $\mathbf K(t,g)$, on doit considérer $g(t,\cdot)$ 
et $(\partial_t+v\cdot\nabla_x)g(\tau,\cdot)\big|_{\tau\in[0,t]}$ comme deux {\og  variables\fg}   indépendantes --- on reviendra d'ailleurs plus loin sur ce point précis.

\begin{lemm}\label{L-ChgDerT}
Soient $\alpha>0$ et $\sigma,\tau\in[0,T_\alpha]$. L'application $\mathbf B_\alpha\ni g\mapsto g(s,\cdot)\in\mathbf C$ est de classe $C^1$, et sa dérivée partielle par rapport à $g(\tau,\cdot)$ 
{\og  à $(\partial_t+v\cdot\nabla_x)g$ constant\fg}   est donnée par la formule\footnote{On s'est résolu faute de mieux à employer ici la notation regrettable en usage dans les traités classiques de thermodynamique,
à savoir $(\partial f/\partial x)_y$ pour désigner la dérivée partielle de $f$ par rapport à la variable $x$, la variable $y$ restant constante dans la prise de dérivée.}
$$
\left(\frac{\partial g(\sigma,\cdot)}{\partial g(\tau,\cdot)}\right)_{(\partial_t+v\cdot\nabla_x)g}=A_{\sigma-\tau}\,,
$$
où $A_\theta:=e^{-\theta v\cdot\nabla_x}$ est le groupe à un paramètre engendré par l'opérateur d'advection, donné par la formule $(A_\theta\psi)(x,v):=\psi(x-\theta v,v)$.
\end{lemm}

\begin{proof}
Exprimons $g(\sigma,\cdot)$ en fonction de $g(\tau,\cdot)$ et de $(\partial_t+v\cdot\nabla_x)g(\tau,\cdot)\big|_{t\in[0,T_\alpha]}$ en observant que
$$
g(\sigma,x,v)=g(t,x+(\tau-\sigma)v,v)-\int_\sigma^\tau(\partial_t+v\cdot\nabla_x)g(\theta,x+(\theta-\sigma)v,v)d\theta\,,
$$
c'est-à-dire que
$$
g(\sigma,\cdot)=A_{\sigma-\tau}g(\tau,\cdot)-\int_\sigma^\tau A_{\sigma-\theta}(\partial_t+v\cdot\nabla_x)g(\theta,\cdot)d\theta\,.
$$
La formule annoncée en découle aussitôt.
\end{proof}

\begin{proof}
On rappelle la condition initiale de \eqref{HJ2}:
$$
\mathbf K(0,g)=\int_{\mathbf T^3\times\mathbf R^3}(e^{g(0,z)}-1)f^{in}(z)dz\,,
$$
d'où l'on tire immédiatement que
$$
\left\langle\frac{\partial\mathbf K(0,g)}{\partial g(0)},\psi\right\rangle=\int_{\mathbf T^3\times\mathbf R^3}e^{g(0,z)}\psi(z)f^{in}(z)dz\,.
$$
Puis, d'après le lemme ci-dessus et la règle de dérivation des fonctions composées
\begin{equation}\label{HJBoltzCondIn}
\left\langle\frac{\partial\mathbf K(0,g)}{\partial g(t)},\psi\right\rangle=\left\langle\frac{\partial\mathbf K(0,g)}{\partial g(0)},A_{-t}\psi\right\rangle=\int_{\mathbf T^3\times\mathbf R^3}e^{g(0,z)}A_{-t}\psi(z)f^{in}(z)dz\,.
\end{equation}
Ensuite, on écrit la forme intégrée en $t$ de \eqref{HJ2}:
$$
\mathbf K(t,g)=\mathbf K(0,g)+\int_0^t\mathcal H\left(\frac{\partial\mathbf K(s,g)}{\partial g(s)},g(s)\right)ds\,,
$$
et on dérive chaque membre de cette égalité par rapport à $g(t,\cdot)$ dans la direction $\psi$, en utilisant de nouveau la règle de dérivation des fonctions composées et  le Lemme \ref{L-ChgDerT} pour trouver 
que
\begin{equation}\label{EvolDbK}
\begin{aligned}
\left\langle\frac{\partial\mathbf K(t,g)}{\partial g(t)},\psi\right\rangle=&\int_{\mathbf T^3\times\mathbf R^3}e^{g(0,z)}A_{-t}\psi(z)f^{in}(z)dz
\\
&+\int_0^t\left\langle\frac{\partial\mathcal H}{\partial p}\left(\frac{\partial\mathbf K(s,g)}{\partial g(s)},g(s)\right),\frac{\partial^2\mathbf K(s,g)}{\partial g(t)\partial g(s)}\cdot(\psi,\cdot)\right\rangle ds
\\
&+\int_0^t\left\langle\frac{\partial\mathcal H}{\partial q}\left(\frac{\partial\mathbf K(s,g)}{\partial g(s)},g(s)\right),A_{s-t}\psi\right\rangle ds\,.
\end{aligned}
\end{equation}
La formule \eqref{DefHamilt} implique que
\begin{equation}\label{dHdp}
\left\langle\frac{\partial\mathcal H}{\partial p}(p,q),\tilde p\right\rangle
=\int_{\mathbf T^3\times(\mathbf R^3)^2\times\mathbf S^2}\left(e^{\mathbf D q(x,v,x_*v_*,\omega)}\!-\!1\right)((v\!-\!v_*)\cdot\omega)_+ d\omega\,p(x,dv)\tilde p(x,dv_*)dx
\end{equation}
d'où l'on tire en particulier que
\begin{equation}\label{Hp(p,0)=0}
\left\langle\frac{\partial\mathcal H}{\partial p}(p,0),\tilde p\right\rangle=0\,.
\end{equation}
D'autre part
\begin{equation}\label{dHdq}
\begin{aligned}
2\left\langle\frac{\partial\mathcal H}{\partial q}(p,q),\psi\right\rangle
\\
=\int_{\mathbf T^3\times(\mathbf R^3)^2\times\mathbf S^2}e^{\mathbf D q(x,v,x_*v_*,\omega)}\mathbf D\psi(x,v,x_*v_*,\omega)((v\!-\!v_*)\cdot\omega)_+ d\omega\,p(x,dv)p(x,dv_*)dx&,
\end{aligned}
\end{equation}
d'où
$$
\left\langle\frac{\partial\mathcal H}{\partial q}(p,0),\psi\right\rangle
=\tfrac12\int_{\mathbf T^3\times(\mathbf R^3)^2\times\mathbf S^2}\mathbf D\psi(x,v,x_*v_*,\omega)((v\!-\!v_*)\cdot\omega)_+ d\omega\,p(x,dv)p(x,dv_*)dx\,.
$$
Spécialisons \eqref{EvolDbK} à $g=0$: il vient
$$
\langle f(t,\cdot),\psi\rangle=\left\langle\frac{\partial\mathbf K(t,g)}{\partial g(t)}\Big|_{g=0},\psi\right\rangle+\int_0^t\left\langle\frac{\partial\mathcal H}{\partial q}\left(f(s,\cdot),0\right),A_{s-t}\psi\right\rangle ds\,,
$$
c'est-à-dire
$$
\begin{aligned}
\langle f(t,\cdot),\psi\rangle=\int_{\mathbf T^3\times\mathbf R^3}\psi(z)A_tf^{in}(z)dz+\tfrac12\int_0^t\int_{\mathbf T^3\times(\mathbf R^3)^2\times\mathbf S^2}\mathbf D\psi(x,v,x_*v_*,\omega)
\\
\times A_{t-s}f(s,x,v)A_{t-s}f(s,x,v_*)((v\!-\!v_*)\cdot\omega)_+ d\omega\,dvdv_*dx&\,.
\end{aligned}
$$
On reconnaît dans cette identité une formulation faible intégrée en temps de l'équation de Boltzmann --- cette écriture utilise la troisième formulation faible de l'intégrale des collisions de Boltzmann rappelée 
dans la section \ref{SS-BoltzMath}, ainsi que la formule de Duhamel pour traiter l'opérateur d'advection dans l'équation de Boltzmann \eqref{BoltzEq}.
\end{proof}

\section{Grandes déviations}


Dans cette section, on va voir comment la fonctionnelle $\mathbf K$ et l'équation de Hamilton--Jacobi qu'elle vérifie interviennent dans l'étude des grandes déviations pour un système de sphères dures dans la limite 
de Boltzmann--Grad.

On va considérer une variante à poids de l'équation de Boltzmann:
\begin{equation}\label{BoltzPoids}
\begin{aligned}
(\partial_t+v\cdot\nabla_x)\phi(t,x,v)
\\
=\!\int_{\mathbf R^3\times\mathbf S^2}\!\!\left(\phi(t,x,v'_*)\phi(t,x,v')e^{-\mathbf Dq(t,x,v,x,v_*,\omega)}\!-\!\phi(t,x,v_*)\phi(t,x,v)e^{\mathbf Dq(t,x,v,x,v_*,\omega)}\right)
\\
\times((v-v_*)\cdot\omega)_+ dv_*d\omega&\,,
\end{aligned}
\end{equation}
avec condition initiale
\begin{equation}\label{CondinBoltzPoids}
\phi(0,x,v)=f^{in}(x,v)e^{q(0,x,v)}\,.
\end{equation}
On supposera que la fonction $q\equiv q(t,x,v)$ est lipschitzienne sur $[0,T]\times\mathbf T^3\times\mathbf R^3$. Pour tous $r,T>0$, on considère
$$
\begin{aligned}
\mathcal R_{r,T}:=\{\phi:\,[0,T]\!\times\!\mathbf T^3\!\times\!\mathbf R^3\to\mathbf R_+\text{ solution forte de \eqref{BoltzPoids}-\eqref{CondinBoltzPoids} pour une fonction }
\\
q\text{ telle que }\|q\|_{W^{1,\infty}([0,T]\times\mathbf T^3\times\mathbf R^3)}\le r&\}\,.
\end{aligned}
$$

Maintenant, on va regarder la mesure empirique $t\mapsto\rho^\eps_t$ comme élément de l'espace $D([0,T],\mathcal M^1_+(\mathbf T^3\times\mathbf R^3))$ des chemins càdlàg à valeurs dans l'espace des mesures 
positives sur $\mathbf T^3\times\mathbf R^3$ de masse finie muni de la topologie faible. L'espace  $D([0,T],\mathcal M^1_+(\mathbf T^3\times\mathbf R^3))$ est muni de la topologie de Skorokhod. Soient $\gamma_1$ 
et $\gamma_2$ deux éléments de $D([0,T],\mathcal M^1_+(\mathbf T^3\times\mathbf R^3))$; posons
$$
d_{[0,T]}(\gamma_1,\gamma_2):=\inf_{\alpha\in\mathcal A_T}\max\left(\sup_{0\le t\le T}|\alpha(t)-t|,\sup_{0\le t\le T}\mathbf d(\gamma_1(t),\gamma_2(\alpha(t)))\right)
$$
où $\mathcal A_T$ est l'ensemble des bijections croissantes de $[0,T]$ dans lui-même, et où
$$
\mathbf d(m_1,m_2)=\sum_{n\ge 1}2^{-n}|\langle m_1-m_2,\chi_n\rangle|\,,
$$
sachant que $(\chi_n)_{n\ge 1}$ est une suite dense dans l'espace $C_0(\mathbf T^3\times\mathbf R^3)$ des fonctions continues tendant vers $0$ à l'infini. La distance $d_{[0,T]}$ définit la topologie de Skorokhod sur
l'espace $D([0,T],\mathcal M^1_+(\mathbf T^3\times\mathbf R^3))$ (voir \parencite{Billing}, chapitre 3, section 12).

Le principaux résultats concernant les grandes déviations de la mesure empirique d'un système de sphères dures dans la limite de Boltzmann--Grad sont résumés dans le théorème suivant, qui regroupe
les Théorèmes 3 et 9 de \textcite{BoGalSRSim1}.

\begin{theo}\label{T-GD}
Supposons que $f^{in}$ vérifie \eqref{BndCondIni}. Pour tout $r>0$, il existe $\alpha\equiv\alpha[r,\beta_0,C_0]>0$ et $T_r\in]0,T_\alpha[$ tels que, dans la limite de Boltzmann--Grad

\noindent
(a) pour tout fermé $\mathcal G\subset D([0,T_r],\mathcal M^1_+(\mathbf T^3\times\mathbf R^3))$ pour la distance $d_{[0,T_r]}$
$$
\varlimsup_{\eps\to 0^+}\frac1{\mu_\eps}\ln\mathbb P_\eps(\rho^\eps\in\mathcal G)\le-\inf_{\phi\in\mathcal G}\mathbf K^*(T_r,\phi)\,,
$$
(b) pour tout ouvert $\mathcal O\subset D([0,T_r],\mathcal M^1_+(\mathbf T^3\times\mathbf R^3))$ pour la distance $d_{[0,T_r]}$
$$
\varliminf_{\eps\to 0^+}\frac1{\mu_\eps}\ln\mathbb P_\eps(\rho^\eps\in\mathcal O)\ge-\inf_{\phi\in\mathcal O\cap\mathcal R_{r,T}}\mathbf K^*(T_r,\phi)\,.
$$
Ici $\mathbf K^*$ est la transformée de Legendre de $\mathbf K$: pour tout $\phi\in D([0,T],\mathcal M^1_+(\mathbf T^3\times\mathbf R^3))$
$$
\mathbf K^*(T,\phi):=\sup_{g\in\mathbf B_\alpha}\left(\langle\phi(T,\cdot),g(T,\cdot)\rangle-\llangle\phi,(\partial_t+v\cdot\nabla_x)g\rrangle-\mathbf K(T,g)\right)\,,
$$
où on a noté
$$
\llangle\phi,q\rrangle:=\int_0^T\langle\phi(t,\cdot),q(t,\cdot)\rangle dt\,.
$$
En particulier, pour tout $\phi\in\mathcal R_{r,T_r}$
$$
\begin{aligned}
\lim_{\eta\to 0^+}\varlimsup_{\eps\to 0^+}\frac1{\mu_\eps}\ln\mathbb P_\eps\left(d_{[0,T_r]}(\rho^\eps,\phi)\le\eta\right)=-\mathbf K^*(T_r,\phi)\,,
\\
\lim_{\eta\to 0^+}\varliminf_{\eps\to 0^+}\frac1{\mu_\eps}\ln\mathbb P_\eps\left(d_{[0,T_r]}(\rho^\eps,\phi)<\eta\right)=-\mathbf K^*(T_r,\phi)\,.
\end{aligned}
$$
\end{theo}

Ce résultat fait l'objet du chapitre 7 de \parencite{BoGalSRSim1}. Sa démonstration met en jeu quelques propriétés importantes sur l'équation de Hamilton--Jacobi de la section précédente, dont nous allons donner 
une idée.

Précisons que ce type de résultat avait été conjecturé auparavant par Rezakhanlou dans \parencite{Reza} et par \textcite{Bouchet} --- avec toutefois une formule apparemment différente de celle du Théorème \ref{T-GD}
pour la fonctionnelle des grandes déviations $\mathbf K^*(T,\phi)$: on y reviendra dans la suite.

\subsection{Méthode des caractéristiques pour l'équation de Hamilton--Jacobi}


Soit $t\in[0,T_\alpha]$. Considérons le système hamiltonien
\begin{equation}\label{HamiltSyst}
\begin{aligned}
{}&(\partial_s+v\cdot\nabla_x)\phi_t(s,\cdot)=\frac{\partial\mathcal H}{\partial q}(\phi_t(s,\cdot),q_t(s,\cdot))\,,\qquad\phi_t(0,\cdot)=f^{in}e^{q(0,\cdot)}\,,
\\
&(\partial_s+v\cdot\nabla_x)(q_t-g)(s,\cdot)=-\frac{\partial\mathcal H}{\partial p}(\phi_t(s,\cdot),q_t(s,\cdot))\,,\quad q_t(t,\cdot)=g(t,\cdot)\,.
\end{aligned}
\end{equation}
Grâce à la formule \eqref{dHdq}, on vérifie sans peine que la première de ces deux équations n'est autre que \eqref{BoltzPoids}-\eqref{CondinBoltzPoids}: ce point sera établi plus loin (voir section \ref{SS-Legendre}). 

\begin{prop}
Soit $\alpha>0$ et $g\in\mathbf B_\alpha$. Soit $(\phi_t,q_t)$ solution du système \eqref{HamiltSyst}. Posons
\begin{equation}\label{DefHatK}
\hat{\mathbf K}(t,g):=\langle f^{in},e^{q_t(0,\cdot)}-1\rangle+\llangle\phi_t,(\partial_s+v\cdot\nabla_x)(q_t-g)\rrangle+\int_0^t\mathcal H(\phi_t(s,\cdot),q_t(s,\cdot))ds\,.
\end{equation}
Alors $\hat{\mathbf K}$ est solution de \eqref{HJ2} sur $[0,t]$, et vérifie en outre
$$
\frac{\partial\hat{\mathbf K}(t,g)}{\partial g(t)}=\phi_t(t)\,,\qquad\frac{\partial\hat{\mathbf K}(t,g)}{\partial(\partial_tg+v\cdot\nabla_xg)}=-\phi_t\,.
$$
\end{prop}

La théorie classique de l'équation de Hamilton--Jacobi dit que le graphe de la différentielle de la solution de l'équation de Hamilton--Jacobi en la variable d'espace est invariant par le flot des équations de Hamilton.
C'est bien ce qu'exprime la proposition ci-dessus, à condition de considérer que les soi-disant variables $g(t,\cdot)$ et $(\partial_s+v\cdot\nabla_x)g$ sont non seulement indépendantes comme on l'a dit plus haut, 
mais encore conjuguées au sens de la théorie des systèmes hamiltoniens.

Cette proposition suggère d'étudier (a) l'unicité de la solution de l'équation d'Hamilton--Jacobi \eqref{HJ2}, et (b) l'existence et unicité pour le système des équations de Hamilton \eqref{HamiltSyst} (cf. section 7.2 
dans le chapitre 7 de \parencite{BoGalSRSim1}). Cette étude est largement basée sur une variante du théorème de Cauchy-Kowalevski de \textcite{Nir} et \textcite{Ov} (voir Appendice de \parencite{BoGalSRSim1}), 
déjà utilisée dans la preuve du théorème de Lanford (voir section 5.3 dans \parencite{FG2014}). Au terme de cette étude, on aboutit à l'énoncé suivant.

\begin{prop}
Soit $\alpha>0$. Il existe $T^*_\alpha\in]0,T_\alpha]$ tel que la fonctionnelle $\hat{\mathbf K}$ soit définie sur $[0,T^*_\alpha]\times\mathbf B_\alpha$, et 
$$
\mathbf K(t,g)=\hat{\mathbf K}(t,g)\,,\quad\text{ pour tout }(t,g)\in[0,T^*_\alpha]\times\mathbf B_\alpha\,.
$$
\end{prop}

\subsection{Transformée de Legendre}\label{SS-Legendre}


Soit $\bar\phi$, solution de \eqref{BoltzPoids} avec condition initiale \eqref{CondinBoltzPoids} pour une fonction poids lipschitzienne $\bar q$ telle que 
$\|\bar q\|_{W^{1,\infty}([0,T_0]\times\mathbf T^3\times\mathbf R^3)}<r$. Rappelons que, d'après la formule \eqref{dHdq},
$$
\begin{aligned}
\left\langle\frac{\partial\mathcal H}{\partial q}(p,q),\psi\right\rangle
\\
=\!\tfrac12\int_{\mathbf T^3\times(\mathbf R^3)^2\times\mathbf S^2}e^{\mathbf D q(x,v,x,v_*,\omega)}\mathbf D\psi(x,v,x,v_*,\omega)((v\!-\!v_*)\cdot\omega)_+d\omega p(x,v)p(x,v_*)dvdv_*dx
\\
=\int_{\mathbf T^3\times(\mathbf R^3)^2\times\mathbf S^2}e^{\mathbf D q(x,v,x,v_*,\omega)}(\psi(x,v')\!-\!\psi(x,v))((v\!-\!v_*)\cdot\omega)_+d\omega p(x,v)p(x,v_*)dvdv_*dx
\\
=\int_{\mathbf T^3\times(\mathbf R^3)^2\times\mathbf S^2}\left(e^{-\mathbf D q(x,v,x,v_*,\omega)}p(x,v')p(x,v'_*)-e^{\mathbf D q(x,v,x,v_*,\omega)}p(x,v)p(x,v_*)\right)
\\
\times((v-v_*)\cdot\omega)_+d\omega dv_*\psi(x,v)dvdx&.
\end{aligned}
$$
La deuxième égalité ci-dessus s'obtient par la même démonstration que celle de la troisième égalité dans la formulation faible de l'intégrale des collisions de Boltzmann de la section \ref{SS-BoltzMath}
(par échange de $v$ et $v_*$). La troisième égalité s'obtient comme la seconde égalité dans la formulation faible de l'intégrale des collisions de Boltzmann de la section \ref{SS-BoltzMath} (en utilisant
la transformation $(v,v_*)\mapsto(v',v'_*)$ à $\omega$ fixé définie par \eqref{LoiColl}). Ainsi
$$
\begin{aligned}
\frac{\partial\mathcal H}{\partial q}(p,q)=\int_{\mathbf R^3\times\mathbf S^2}\left(e^{-\mathbf D q(x,v,x,v_*,\omega)}p(x,v')p(x,v'_*)-e^{\mathbf D q(x,v,x,v_*,\omega)}p(x,v)p(x,v_*)\right)
\\
\times((v-v_*)\cdot\omega)_+d\omega dv_*&\,,
\end{aligned}
$$
ce qui montre que l'équation de Boltzmann à poids \eqref{BoltzPoids} avec condition initiale \eqref{CondinBoltzPoids} coïncide avec la première équation du système hamiltonien \eqref{HamiltSyst},
comme annoncé plus haut.

En utilisant de nouveau une variante du théorème de Cauchy-Kowalevski de \textcite{Nir} et \textcite{Ov}, on montre que le problème de Cauchy \eqref{BoltzPoids}-\eqref{CondinBoltzPoids} admet une 
unique solution sur l'intervalle $[0,T_0e^{-5r}]$, vérifiant la borne
$$
\sup_{0\le t\le T_0e^{-5r}}\|\bar\phi(t,\cdot)\exp(\tfrac14\beta_0|v|^2)\|_{L^\infty(\mathbf T^3\times\mathbf R^3)}\le 4C_0e^r\,.
$$
On vérifie de même que pour l'équation de Boltzmann originale \eqref{BoltzEq} que les solutions de l'équation de Boltzmann à poids \eqref{BoltzPoids} sont de masse, d'impulsion et d'énergie constantes:
$$
\langle(\partial_s+v\cdot\nabla_x)\bar\phi(s,\cdot),1\rangle=\langle(\partial_s+v\cdot\nabla_x)\bar\phi(s,\cdot),v_i\rangle=\langle(\partial_s+v\cdot\nabla_x)\bar\phi(s,\cdot),|v|^2\rangle=0
$$
pour $i=1,2,3$. La démonstration suit de près celle des lois de conservation locales établies dans la section \ref{SS-BoltzMath} pour l'équation de Boltzmann originale.

\begin{prop}\label{P-CalculLegendre}
Supposons que $f^{in}$ vérifie \eqref{BndCondIni}. Pour tout $r>0$, il existe $\alpha\equiv\alpha[r,\beta_0,C_0]>0$ et $T_r\in]0,T_\alpha[$ tels que
$$
\begin{aligned}
\mathbf K^*(t,\phi)=&\int_{\mathbf T^3\times\mathbf R^3}\left(\phi(0,z)\ln\left(\frac{\phi(0,z)}{f^{in}(z)}\right)-\phi(0,z)+f^{in}(z)\right)dz
\\
&+\sup_{q\in L^\infty([0,t]\times\mathbf T^3\times\mathbf R^3)}\left(\llangle(\partial_s+v\cdot\nabla_x)\phi,q\rrangle-\int_0^t\mathcal H(\phi(s,\cdot),q(s,\cdot))ds\right)\,.
\end{aligned}
$$
\end{prop}

\smallskip
Or la première équation du système hamiltonien \eqref{HamiltSyst} exprime précisément le fait que la fonction $\bar q$, c'est-à-dire la fonction poids dans \eqref{BoltzPoids}-\eqref{CondinBoltzPoids}, est 
un point critique pour le problème variationnel de la proposition ci-dessus. Comme d'autre part la solution $\bar\phi$ vérifie $\bar\phi\ge 0$, on déduit de la formule \eqref{DefHamilt} que 
$$
q\mapsto\mathcal H(\bar\phi(s,\cdot),q)
$$
est convexe. Par conséquent
$$
\begin{aligned}
\sup_{q\in L^\infty([0,t]\times\mathbf T^3\times\mathbf R^3)}\left(\llangle(\partial_s+v\cdot\nabla_x)\phi,q\rrangle-\int_0^t\mathcal H(\phi(s,\cdot),q(s,\cdot))ds\right)
\\
=\llangle(\partial_s+v\cdot\nabla_x)\bar\phi,\bar q\rrangle-\int_0^t\mathcal H(\bar\phi(s,\cdot),\bar q(s,\cdot))ds&\,.
\end{aligned}
$$
Cette formule permet donc de calculer {\og  explicitement\fg}   la fonctionnelle des grandes déviations intervenant dans le Théorème \ref{T-GD} par la méthode des caractéristiques, c'est-à-dire en résolvant le système 
hamiltonien \eqref{HamiltSyst}.

La Proposition \ref{P-CalculLegendre} montre que la fonctionnelle des grandes déviations obtenue dans le Théorème \ref{T-GD} coïncide bien avec celle conjecturée dans les travaux de Rezakhanlou (Conjecture 4
du texte de Rezakhanlou dans \parencite{Reza}), et de Bouchet (formule (27) de \parencite{Bouchet}).

\smallskip
Terminons par une remarque intéressante --- qui joue un rôle dans la démonstration des résultats de cette section (voir \parencite{BoGalSRSim1}, chapitre 7, section 7.2).

Il est commode de changer les variables $(p,q)$ du hamiltonien \eqref{DefHamilt} en posant $P=pe^{-q}$ et $Q=e^q$, d'où
\begin{equation}\label{DefHamiltPrime}
\begin{aligned}
\mathcal H(p,q)=\tfrac12\int_{\mathbf T^3\times(\mathbf R^3)^2\times\mathbf S^2}\left(Q(x,v')Q(x,v'_*)-Q(x,v)Q(x,v_*)\right)&
\\
\times((v-v_*)\cdot\omega)_+ d\omega\,P(x,dv)P(x,dv_*)dx&=\mathcal H'(P,Q)\,.
\end{aligned}
\end{equation}
Un calcul simple montre que
$$
\begin{aligned}
\frac{\partial H'}{\partial P}(P,Q)=\int_{\mathbf R^3\times\mathbf S^2}\left(Q(x,v')Q(x,v'_*)-Q(x,v)Q(x,v_*)\right)P(x,v_*)((v-v_*)\cdot\omega)_+ dv_*d\omega&\,,
\\
\frac{\partial H'}{\partial Q}(P,Q)
=\int_{\mathbf R^3\times\mathbf S^2}\left(P(x,v')P(x,v'_*)-P(x,v)P(x,v_*)\right)Q(x,v_*)((v-v_*)\cdot\omega)_+ dv_*d\omega&\,.
\end{aligned}
$$
On peut alors réécrire le système \eqref{HamiltSyst} au moyen des nouvelles fonctions inconnues $\gamma_t:=\phi_te^{-q_t}$ et $\eta_t=e^{q_t}$: on trouve sans difficulté\footnote{En notant $f'=f(t,x,v')$ 
et $f'_*=f(t,x,v'_*)$, tandis que $f_*=f(t,x,v_*)$, et que les vitesses $v',v'_*$ sont données par \eqref{LoiColl} en fonction de $v,v_*$ et $\omega$.} que
$$
\begin{aligned}
(\partial_s\!+\!v\cdot\nabla_x)\gamma_t\!+\!\gamma_t(\partial_s\!+\!v\cdot\nabla_x)g=&\frac{\partial H'}{\partial Q}(\gamma_t,\eta_t)
\\
=&+\int_{\mathbf R^3\times\mathbf S^2}(\gamma'_t\gamma'_{t*}-\gamma_t\gamma_{t*})\eta_{t*}((v-v_*\cdot\omega))_+dv_*d\omega\,,
\\
(\partial_s\!+\!v\cdot\nabla_x)\eta_t\!-\!\eta_t(\partial_s\!+\!v\cdot\nabla_x)g=&-\frac{\partial H'}{\partial P}(\gamma_t,\eta_t)
\\
=&-\int_{\mathbf R^3\times\mathbf S^2}(\eta'_t\eta'_{t*}-\eta_t\eta_{t*})\gamma_{t*}((v-v_*\cdot\omega))_+dv_*d\omega\,,
\end{aligned}
$$
système posé sur $[0,t]\times\mathbf T^3\times\mathbf R^3$ avec les conditions aux limites
$$
\gamma_t(0,\cdot)=f^{in}\,,\qquad \eta_t(t,\cdot)=e^{g(t,\cdot)}\,.
$$
La structure de ce nouveau système hamiltonien est intéressante: il s'agit de deux équations de Boltzmann à poids, mais \textit{avec des sens opposés de propagation du temps}. Comme il s'agit de prouver
l'existence et l'unicité de solutions à des équations de type Boltzmann sans utiliser la flèche du temps {\og  naturelle\fg}  , à savoir celle pour laquelle
$$
t\mapsto\iint_{\mathbf T^3\times\mathbf R^3}\phi(t,x,v)\ln\phi(t,x,v)dxdv\text{ est décroissante,}
$$
il n'est pas très étonnant que la démonstration repose de nouveau sur le théorème de Cauchy-Kowalevski abstrait de \textcite{Nir,Ov}: voir la Proposition 7.2.3 de \parencite{BoGalSRSim1}.

\section{Fluctuations}


Dans cette section, on va s'intéresser au champ de fluctuations \eqref{DefFluct} de la mesure empirique $\rho^\eps_t$ dans la limite de Boltzmann--Grad. Rappelons-en la définition:
$$
\begin{aligned}
\langle\zeta^\eps_t,\phi\rangle:=&\sqrt{\mu_\eps}\left(\langle\rho^\eps_t,\phi\rangle-\mathbb E_\eps(\langle\rho^\eps_t,\phi\rangle)\right)
\\
=&\frac1{\sqrt{\mu_\eps}}\left(\sum_{i=1}^N\phi(\mathbf z_i(t))-\mu_\eps\int_{\mathbf T^3\times\mathbf R^3}F^\eps_1(t,z)\phi(z)dz\right)\,.
\end{aligned}
$$

Puisqu'il s'agit de décrire les fluctuations de la mesure empirique autour de $F^\eps_1$ qui converge vers une solution de l'équation de Boltzmann, nous aurons évidemment besoin des opérateurs linéarisés 
(adjoints) associés à l'équation de Boltzmann \eqref{BoltzEq} autour d'une solution $f$ définie sur $[0,T]\times\mathbf T^3\times\mathbf R^3$:
$$
\begin{aligned}
\mathcal L^*_t\phi(z):=&v\cdot\nabla_x\phi(z)+\mathbf L^*_t\phi(z)\,,\quad\text{ où on a posé}
\\
\mathbf L^*_t\phi(z):=&\int_{\mathbf R^3\times\mathbf S^2}\mathbf D\phi(x,v,x,v^*,\omega)f(t,x,v_*)((v-v_*)\cdot\omega)_+dv_*d\omega\,.
\end{aligned}
$$

\begin{theo}\label{T-BFluct}
Supposons que $f^{in}$ vérifie \eqref{BndCondIni}. Il existe $T^*\equiv T^*[C_0,\beta_0]>0$ tel que, dans la limite de Boltzmann--Grad, $\zeta^\eps_t$ converge en loi vers $\zeta_t$ pour tout $t\in[0,T^*]$ 
lorsque $\eps\to 0^+$, où $\zeta_t$ est le processus gaussien centré solution de
\begin{equation}\label{BFluct}
d\zeta_t=\mathcal L_t\zeta_tdt+db_t\,,\qquad 0<t<T^*\,,
\end{equation}
et où $b_t$ est un bruit gaussien delta-corrélé en $t,x$, de covariance
$$
\begin{aligned}
\mathbf{Cov}(t,\phi,\psi):=\tfrac12\int_{\mathbf T^3\times(\mathbf R^3)^2\times\mathbf S^2}\mathbf D\phi(x,v,x,v_*,\omega)\mathbf D\psi(x,v,x,v_*,\omega)
\\
\times f(t,x,v)f(t,x,v_*)((v-v_*)\cdot\omega)_+dxdvdv_*d\omega&\,.
\end{aligned}
$$
De plus $\zeta_0$ est le champ gaussien centré de covariance
$$
\mathbf E(\langle\zeta_0,\phi\rangle\langle\zeta_0,\psi\rangle)=\int_{\mathbf T^3\times\mathbf R^3}\phi(z)\psi(z)f^{in}(z)dz\,.
$$
\end{theo}

Spohn avait conjecturé ce résultat dans \parencite{Spohn81}: voir l'équation (4.1) de cette référence, ainsi que l'article antérieur \parencite{vBLaLeSpohn} dans le cas où la solution de l'équation de Boltzmann 
considérée est un équilibre maxwellien \eqref{Maxw}. 

\subsection{Solutions faibles de \eqref{BFluct}}


Précisons la notion de solution de l'équation de Boltzmann aux fluctuations \eqref{BFluct}. Soit $U(t,s)$ le flot engendré par l'opérateur de Boltzmann linéarisé $\mathcal L_t$: autrement dit
$$
\partial_tU(t,s)\phi-\mathcal L_tU(t,s)\phi=0\,,\qquad U(s,s)\phi=\phi\,,\qquad 0\le s\le t\le T^*\,.
$$
En dérivant par rapport à $t_2$ la relation de flot $U(t_1,t_3)=U(t_1,t_2)U(t_2,t_3)$, il vient
$$
\partial_sU(t,s)\phi+U(t,s)\mathcal L_s\phi=0\,,\qquad U(t,t)\phi=\phi\,,\qquad 0\le s\le t\le T^*\,.
$$
Puis, par passage à l'adjoint, on trouve que
$$
\begin{aligned}
\partial_sU(t,s)^*\phi+\mathcal L^*_sU(t,s)^*\phi&=0\,,\qquad U(t,t)^*\phi=\phi\,,\qquad 0\le s\le t\le T^*\,,
\\
\partial_tU(t,s)^*\phi-U(t,s)^*\mathcal L^*_t\phi&=0\,,\qquad U(s,s)^*\phi=\phi\,,\qquad 0\le s\le t\le T^*\,.
\end{aligned}
$$
Formellement, la solution de \eqref{BFluct} vérifie
$$
\langle\zeta_t,\phi\rangle=\langle\zeta_0,U(t,0)^*\phi\rangle+\int_0^t\langle db_s,U(t,s)^*\phi\rangle
$$
pour toute fonction test $\phi$. Donc
$$
\begin{aligned}
\hat{\mathcal C}(t,s,\psi,\phi):=\mathbb E(\langle\zeta_t,\psi\rangle\langle\zeta_s,\phi\rangle)=&\mathbb E(\langle\zeta_0,U(t,0)^*\psi\rangle\langle\zeta_0,U(s,0)^*\phi\rangle)
\\
&+\mathbb E\left(\int_0^t\int_0^s\langle db_\tau,U(t,\tau)^*\phi\rangle\langle db_\sigma,U(s,\sigma)^*\phi\rangle\right)
\\
&+\mathbb E\left(\langle\zeta_0,U(t,0)^*\psi\rangle\int_0^s\langle db_\sigma,U(s,\sigma)^*\phi\rangle d\sigma\right)
\\
&+\mathbb E\left(\langle\zeta_0,U(s,0)^*\psi\rangle\int_0^t\langle db_\tau,U(t,\tau)^*\phi\rangle d\tau\right)\,,
\end{aligned}
$$
ce qui se réécrit
\begin{equation}\label{DefSolFluct}
\hat{\mathcal C}(t,s,\psi,\phi)\!=\mathbb E(\langle\zeta_0,U(t,0)^*\psi\rangle\langle\zeta_0,U(s,0)^*\phi\rangle)\!+\!\int_0^s\mathbf{Cov}(\sigma,U(t,\sigma)^*\psi,U(s,\sigma)^*\phi)d\sigma\,.
\end{equation}

\begin{defi}
Une solution faible de \eqref{BFluct} est un processus gaussien centré dont la covariance vérifie l'identité ci-dessus. (On rappelle en effet que la loi d'un processus gaussien centré est complètement 
déterminée par sa covariance).
\end{defi}

Dérivons chaque membre de l'identité ci-dessus par rapport à $t$: 
$$
\begin{aligned}
\partial_t\hat{\mathcal C}(t,s,\psi,\phi)\!\!=&\mathbb E(\langle\zeta_0,U(t,0)^*\mathcal L^*_t\psi\rangle\langle\zeta_0,U(s,0)^*\phi\rangle)\!+\!\!\int_0^s\!\!\mathbf{Cov}(\sigma,U(t,\sigma)^*\mathcal L^*_t\psi,U(s,\sigma)^*\phi)d\sigma
\\
=&\hat{\mathcal C}(t,s,\mathcal L^*_t\psi,\phi)\,,\quad t>s\,.
\end{aligned}
$$
D'autre part, pour $t=s$
$$
\begin{aligned}
\frac{d}{dt}\hat{\mathcal C}(t,t,\psi,\phi)=&\left(\partial_t\hat{\mathcal C}(t,s,\psi,\phi)+\partial_s\hat{\mathcal C}(t,s,\psi,\phi)\right)\Big|_{t=s}
\\
=&\mathbb E(\langle\zeta_0,U(t,0)^*\mathcal L^*_t\psi\rangle\langle\zeta_0,U(t,0)^*\phi\rangle)
\\
&+\int_0^t\mathbf{Cov}(\sigma,U(t,\sigma)^*\mathcal L^*_t\psi,U(t,\sigma)^*\phi)d\sigma
\\
&+\mathbb E(\langle\zeta_0,U(t,0)^*\psi\rangle\langle\zeta_0,U(t,0)^*\mathcal L^*_t\phi\rangle)
\\
&+\int_0^t\mathbf{Cov}(\sigma,U(t,\sigma)^*\psi,U(t,\sigma)^*\mathcal L^*_t\phi)d\sigma+\mathbf{Cov}(t,\psi,\phi)\,,
\end{aligned}
$$
de sorte que
$$
\frac{d}{dt}\hat{\mathcal C}(t,t,\psi,\phi)=\hat{\mathcal C}(t,t,\mathcal L^*_t\psi,\phi)+\hat{\mathcal C}(t,t,\psi,\mathcal L^*_t\phi)+\mathbf{Cov}(t,\psi,\phi)\,.
$$
Évidemment, il s'agit là d'un calcul formel (en toute rigueur, il faudrait revenir à la formulation intégrale en temps équivalente). Toutefois, nous allons, dans la section suivante, faire le lien entre ces formules et 
la solution de l'équation de Hamilton--Jacobi fonctionnelle \eqref{HJ2}.

Avant cela, on va modifier très légèrement le calcul ci-dessus de façon à absorber la dérivée en temps le long des caractéristiques de l'opérateur d'advection $v\cdot\nabla_x$. Rappelons la notation $A_t=e^{-tv\cdot\nabla_x}$,
c'est-à-dire que $A_t\phi(x,v)=\phi(x-tv,v)$.

Comme $(\phi,\psi)\mapsto\hat{\mathcal C}(t,s,\psi,\phi)$ est une forme bilinéaire, on calcule facilement
$$
\begin{aligned}
\frac{d}{dt}\hat{\mathcal C}(t,t,A_t\psi,A_t\phi)=&\hat{\mathcal C}(t,t,\mathcal L^*_tA_t\psi,A_t\phi)+\hat{\mathcal C}(t,t,A_t\psi,\mathcal L^*_tA_t\phi)+\mathbf{Cov}(t,A_t\psi,A_t\phi)
\\
&-\hat{\mathcal C}(t,t,v\cdot\nabla_xA_t\psi,A_t\phi)-\hat{\mathcal C}(t,t,A_t\psi,v\cdot\nabla_xA_t\phi)
\\
=&\hat{\mathcal C}(t,t,\mathbf L^*_tA_t\psi,A_t\phi)+\hat{\mathcal C}(t,t,A_t\psi,\mathbf L^*_tA_t\phi)+\mathbf{Cov}(t,A_t\psi,A_t\phi)\,.
\end{aligned}
$$

D'autre part, pour $t>s$, on a
$$
\begin{aligned}
\frac{d}{dt}\hat{\mathcal C}(t,s,A_t\psi,\phi_s)=&\hat{\mathcal C}(t,s,\mathcal L^*_tA_t\psi,\phi_s)-\hat{\mathcal C}(t,s,v\cdot\nabla_xA_t\psi,\phi_s)
\\
=&\hat{\mathcal C}(t,s,\mathbf L^*_tA_t\psi,\phi_s)\,,
\end{aligned}
$$
de sorte que
$$
\hat{\mathcal C}(t,s,A_t\psi,\phi_s)=\hat{\mathcal C}(s,s,A_s\psi,\phi_s\sigma)+\int_s^t\hat{\mathcal C}(\sigma,s,\mathbf L^*_\sigma A_\sigma\psi,\phi_s)\,,
$$
identité dont on intègre ensuite chaque membre par rapport à $s$, pour aboutir à
$$
\int_0^t\hat{\mathcal C}(t,s,A_t\psi,\phi_s)ds=\int_0^t\left(\hat{\mathcal C}(s,s,A_s\psi,\phi_s)+\int_s^t\hat{\mathcal C}(\sigma,s,\mathbf L^*_\sigma A_\sigma\psi,\phi_s)\right)ds\,.
$$
En remplaçant $\psi$ par $A_{-t}\psi$, on trouve finalement que
$$
\int_0^t\hat{\mathcal C}(t,s,\psi,\phi_s)ds=\int_0^t\left(\hat{\mathcal C}(s,s,A_{s-t}\psi,\phi_s)+\int_s^t\hat{\mathcal C}(\sigma,s,\mathbf L^*_\sigma A_{\sigma-t}\psi,\phi_s)d\sigma\right)ds\,.
$$

En résumé, la solution $\zeta_t$ de \eqref{BFluct} est un processus gaussien centré dont la covariance $\hat{\mathcal C}(t,s,\psi,\phi)$ est solution du système d'équations intégrales
\begin{equation}\label{SystChat}
\left\{
\begin{aligned}
{}&\hat{\mathcal C}(t,t,\psi,\phi)=\hat{\mathcal C}(0,0,A_{-t}\psi,A_{-t}\phi)+\int_0^t\mathbf{Cov}(s,A_{s-t}\psi,A_{s-t}\phi)
\\
&\qquad\qquad\qquad+\int_0^t\left(\hat{\mathcal C}(s,s,\mathbf L^*_sA_{s-t}\psi,A_{s-t}\phi)+\hat{\mathcal C}(s,s,A_{s-t}\psi,\mathbf L^*_sA_{s-t}\phi)\right)ds\,,
\\
&\int_0^t\hat{\mathcal C}(t,s,\psi,\phi_s)ds=\int_0^t\left(\hat{\mathcal C}(s,s,A_{s-t}\psi,\phi_s)+\int_s^t\hat{\mathcal C}(\sigma,s,\mathbf L^*_\sigma A_{\sigma-t}\psi,\phi_s)d\sigma\right)ds\,.
\end{aligned}
\right.
\end{equation}
On y ajoute l'information provenant de la condition initiale:
\begin{equation}\label{CondInChat}
\hat{\mathcal C}(0,0,\psi,\phi)=\int_{\mathbf T^3\times\mathbf R^3}\phi(z)\psi(z)f^{in}(z)dz\,.
\end{equation}

Or on va voir que l'équation d'Hamilton--Jacobi permet justement de montrer que la limite de la covariance du champ de fluctuation $\zeta^\eps_t$ lorsque $\eps\to 0^+$ vérifie le même système \eqref{SystChat}
d'équations intégrales que ci-dessus. Ce résultat, que nous allons présenter dans la section suivante, est évidemment une étape fondamentale dans la démonstration du Théorème \ref{T-BFluct}. 

\subsection{Covariance limite}


La covariance du champ de fluctuations est définie comme suit: pour toute paire de fonctions test $\phi,\psi\in C_b(\mathbf T^3\times\mathbf R^3)$, on pose
$$
\mathcal C_\eps(t,s,\psi,\phi):=\mathbb E_\eps(\langle\zeta^\eps_t,\psi\rangle\langle\zeta^\eps_s,\phi\rangle)\,,\qquad 0\le s\le t\,.
$$
L'équation de Hamilton--Jacobi \eqref{HJ2} va nous permettre de caractériser la covariance limite du champ de fluctuations (limite de $\mathcal C_\eps(t,s,\psi,\phi)$ lorsque $\eps\to 0^+$).

\begin{prop}\label{P-EvolCov}
Supposons que $f^{in}$ vérifie \eqref{BndCondIni}, et soit $f$ la solution du problème de Cauchy \eqref{BoltzEq}-\eqref{BoltzCondIn} sur $[0,T^*]\times\mathbf T^3\times\mathbf R^3$ où 
$T^*\equiv T^*[C_0,\beta_0]>0$. Pour toute paire de fonctions test $\phi,\psi\in C_b(\mathbf T^3\times\mathbf R^3)$ et tous $s\le t$ dans $[0,T_0]$,
$$
\mathcal C_\eps(t,s,\psi,\phi)\to\mathcal C(t,s,\psi,\phi)\quad\text{ lorsque }\eps\to 0^+\,.
$$
De plus, la covariance limite $\mathcal C$ est solution du système \eqref{SystChat} d'équations intégrales, et vérifie la condition initiale \eqref{CondInChat}.
\end{prop}

\begin{proof}
Observons pour commencer que
$$
\frac{\partial^2\mathbf K(t,g)}{\partial g(t)^2}\Big|_{g=0}\cdot(\psi,\phi)=\mathcal C(t,t,\psi,\phi)\,.
$$
Partons de \eqref{HJBoltzCondIn}; en dérivant une fois de plus par rapport à $g(t,\cdot)$, il vient
$$
\frac{\partial^2\mathbf K(0,g)}{\partial g(t)^2}\Big|_{g=0}\cdot(\psi,\phi)=\int_{\mathbf T^3\times\mathbf R^3}A_{-t}\psi(z)A_{-t}\phi(z)f^{in}(z)dz=\mathcal C(0,0,A_{-t}\psi,A_{-t}\phi)\,.
$$
Revenons ensuite à l'identité \eqref{EvolDbK}, dont on dérive chaque membre au point $g=0$ par rapport à $g(t,\cdot)$ dans la direction $\psi$, pour trouver
$$
\begin{aligned}
\frac{\partial^2\mathbf K(t,g)}{\partial g(t)^2}\cdot(\phi,\psi)=\frac{\partial^2\mathbf K(0,g)}{\partial g(t)^2}\cdot(\phi,\psi)
\\
+\int_0^t\frac{\partial^2\mathcal H}{\partial p^2}\left(\frac{\partial\mathbf K(s,g)}{\partial g(s)},g(s)\right)
\cdot\left(\frac{\partial^2\mathbf K(s,g)}{\partial g(t)\partial g(s)}\cdot(\phi,\cdot),\frac{\partial^2\mathbf K(s,g)}{\partial g(t)\partial g(s)}\cdot(\psi,\cdot)\right)ds
\\
+\int_0^t\frac{\partial^2\mathcal H}{\partial q\partial p}\left(\frac{\partial\mathbf K(s,g)}{\partial g(s)},g(s)\right)
\cdot\left(A_{s-t}\phi,\frac{\partial^2\mathbf K(s,g)}{\partial g(t)\partial g(s)}\cdot(\psi,\cdot)\right)ds
\\
+\int_0^t\left\langle\frac{\partial\mathcal H}{\partial p}\left(\frac{\partial\mathbf K(s,g)}{\partial g(s)},g(s)\right),\frac{\partial^3\mathbf K(s,g)}{\partial g(t)^2\partial g(s)}\cdot(\phi,\psi,\cdot)\right\rangle ds
\\
+\int_0^t\frac{\partial^2\mathcal H}{\partial p\partial q}\left(\frac{\partial\mathbf K(s,g)}{\partial g(s)},g(s)\right)\cdot\left(\frac{\partial^2\mathbf K(s,g)}{\partial g(t)\partial g(s)}\cdot(\phi,\cdot),A_{s-t}\psi\right) ds
\\
+\int_0^t\frac{\partial^2\mathcal H}{\partial q^2}\left(\frac{\partial\mathbf K(s,g)}{\partial g(s)},g(s)\right)\cdot\left(A_{s-t}\phi,A_{s-t}\psi\right) ds&\,.
\end{aligned}
$$
En utilisant la règle de dérivation des fonctions composées et le Lemme \ref{L-ChgDerT}, on a
$$
\frac{\partial^2\mathbf K(s,g)}{\partial g(t)\partial g(s)}\cdot(\psi,\cdot)=\frac{\partial^2\mathbf K(s,g)}{\partial g(s)^2}\cdot(A_{s-t}\psi,\cdot)\,.
$$
D'autre part
$$
\mathcal H(p,0)=0\,,\quad\text{ d'où }\quad\frac{\partial\mathcal H}{\partial p}(p,0)=0\quad\text{ et }\quad\frac{\partial^2\mathcal H}{\partial p^2}(p,0)=0\,,
$$
de sorte qu'en faisant $g=0$ dans l'identité ci-dessus, il vient
$$
\begin{aligned}
\mathcal C(t,t,\phi,\psi)=\mathcal C(0,0,A_{-t}\phi,A_{-t}\psi)
\\
+\int_0^t\frac{\partial^2\mathcal H}{\partial q\partial p}\left(f(s,\cdot),0\right)\cdot\left(A_{s-t}\phi,\frac{\partial^2\mathbf K(s,0)}{\partial g(s)^2}\cdot(A_{s-t}\psi,\cdot)\right)ds
\\
+\int_0^t\frac{\partial^2\mathcal H}{\partial p\partial q}\left(f(s,\cdot),0\right)\cdot\left(\frac{\partial^2\mathbf K(s,0)}{\partial g(s)^2}\cdot(A_{s-t}\phi,\cdot),A_{s-t}\psi\right) ds
\\
+\int_0^t\frac{\partial^2\mathcal H}{\partial q^2}\left(f(s,\cdot),0\right)\cdot\left(A_{s-t}\phi,A_{s-t}\psi\right) ds
\end{aligned}
$$
puisqu'on sait, d'après la Proposition \ref{P-HJBoltz}, que $\frac{\partial\mathbf K(s,0)}{\partial g(s)}=f(s,\cdot)$, où $f$ est la solution de l'équation de Boltzmann \eqref{BoltzEq}.

Ensuite, la formule \eqref{DefHamilt} montre que
$$
\begin{aligned}
\frac{\partial^2\mathcal H}{\partial q^2}\left(f(s,\cdot),0\right)\cdot(\Phi,\Psi)=\tfrac12\int_{\mathbf T^3\times(\mathbf R^3)^2\times\mathbf S^2}\mathbf D\Phi(s,x,v,x,v_*,\omega)\mathbf D\Psi(s,x,v,x,v_*,\omega) 
\\
\times f(s,x,v)f(s,x,v_*)((v-v_*)\cdot\omega))_+d\omega dvdv_*dx=\mathbf{Cov}(s,\Phi,\Psi)&\,,
\end{aligned}
$$
tandis que
$$
\begin{aligned}
\frac{\partial^2\mathcal H}{\partial p\partial q}\left(f(s,\cdot),0\right)\cdot(\Phi,\Psi)
\\
=\tfrac12\int_{\mathbf T^3\times(\mathbf R^3)^2\times\mathbf S^2}\mathbf D\Psi(t,x,v,x,v_*,\omega)(\Phi(s,x,v)f(s,x,v_*)+\Phi(s,x,v_*)f(s,x,v))
\\
\times((v-v_*)\cdot\omega))_+d\omega dvdv_*dx
\\
=\int_{\mathbf T^3\times(\mathbf R^3)^2\times\mathbf S^2}(\Phi(x,v')f(s,x,v'_*)+\Phi(x,v'_*)f(s,x,v')-\Phi(x,v)f(s,x,v_*)
\\
-\Phi(x,v_*)f(s,x,v))\Psi(x,v)((v-v_*)\cdot\omega))_+d\omega dvdv_*dx
\\
=\int_{\mathbf T^3\times\mathbf R^3}\Psi(x,v)\mathbf L_t\Phi(x,v)dvdx=\int_{\mathbf T^3\times\mathbf R^3}\Phi(x,v)\mathbf L^*_t\Psi(x,v)dvdx&\,.
\end{aligned}
$$
Par conséquent
$$
\begin{aligned}
\mathcal C(t,t,\phi,\psi)=&\,\mathcal C(0,0,A_{-t}\phi,A_{-t}\psi)+\int_0^t\mathcal C(s,s,A_{s-t}\psi,\mathbf L^*_t(A_{s-t}\phi))ds
\\
&+\int_0^t\mathcal C(s,s,A_{s-t}\phi,\mathbf L^*_t(A_{s-t}\psi))ds+\int_0^t\mathbf{Cov}(s,A_{s-t}\phi,A_{s-t}\psi)ds\,.
\end{aligned}
$$
On vient donc de déduire de l'équation de Hamilton--Jacobi \eqref{HJ2} le fait que la limite de $\mathcal C_\eps(t,t,\cdot,\cdot)$ vérifie la première équation du système \eqref{SystChat}, tout comme la covariance 
$\hat{\mathcal C}(t,t,\cdot,\cdot)$ de la solution de \eqref{BFluct}. Que la limite de $\mathcal C_\eps(t,s,\cdot,\cdot)$ vérifie également la seconde équation du système \eqref{SystChat} s'obtient par un calcul 
analogue, quoiqu'un peu plus long, que nous ne ferons donc pas ici, mais pour lequel nous renvoyons le lecteur à la fin de la preuve de la Proposition 5.5.2 dans \parencite{BoGalSRSim1}.
\end{proof}

\smallskip
Comme annoncé plus haut, il suffit maintenant de démontrer l'unicité des solutions de \eqref{SystChat} pour en déduire que la corrélation limite $\mathcal C$ du champ de fluctuations de la mesure empirique 
$\zeta^\eps_t$ autour de la solution de l'équation de Boltzmann coïncide avec la corrélation $\hat{\mathcal C}$ de la solution de \eqref{BFluct}. 

Après quoi, une fois démontré que $\zeta^\eps_t$ converge en loi vers un champ gaussien centré, on en déduira que ce champ gaussien est précisément la solution $\zeta_t$ de \eqref{BFluct}, puisqu'un champ 
gaussien centré est caractérisé par sa covariance.

\smallskip
Ces deux étapes bien distinctes de la démonstration du Théorème \ref{T-BFluct} (unicité et convergence vers un champ gaussien centré) ne reposent pas sur l'équation de Hamilton--Jacobi, et nous renvoyons 
le lecteur intéressé à en connaître les démonstrations aux sections 6.2 et 6.3, chapitre 6 de \parencite{BoGalSRSim1}.

\subsection{Formule de Spohn pour la covariance limite}


\textcite{Spohn81} propose encore une autre formule pour la covariance limite $\mathcal C(t,s,\psi,\phi)$. 

\begin{prop}\label{P-Spohn}
Sous les mêmes hypothèses qu'à la Proposition \ref{P-EvolCov}, et pour tous $0\le s\le t\le T^*$, l'on a
$$
\begin{aligned}
\mathcal C(t,s,\psi,\phi)=\int_{\mathbf T^3\times\mathbf R^3}U(t,s)^*\psi(z)\phi(z)f(s,z)dz
\\
+\int_0^t\int_{(\mathbf T^3\times\mathbf R^3)^2}R^{1,2}(f(\tau,\cdot),f(\tau,\cdot))(z,z_*)(U(t,\tau)^*\psi)(z)(U(s,\tau)^*\phi)(z_*) dzdz_*&\,,
\end{aligned}
$$
où $R^{1,2}$ est l'opérateur de recollision, défini comme suit:
$$
R^{1,2}(g,g)(z_1,z_2):=\delta_0(x_1-x_2)\int_{\mathbf S^2}(g(z'_1)g(z'_2)-g(z_1)g(z_2))((v_1-v_2)\cdot\omega)_+d\omega\,.
$$
\end{prop}

\smallskip
Avant de donner la démonstration de cette proposition, expliquons la terminologie d'{\og  opérateur de recollision\fg}   désignant $R^{1,2}$. Pour cela, revenons à l'équation de Liouville \eqref{LiouvilleDistrib}, et écrivons
l'équation vérifiée par $\mathbb F_{n:2}$. On intègre donc chaque membre de \eqref{LiouvilleDistrib} par rapport aux variables $z_3,\ldots,z_n$, pour trouver que
$$
\begin{aligned}
(\partial_t+\sum_{i=1}^2v_i\cdot\nabla_{x_i})\mathbb F_{n:2}(t,z_1,z_2)
\\
=(n-2)\eps^2\mathcal B_\eps^{1,3}(\mathbb F_{n:3})(t,z_1,z_2)+(n-2)\eps^2\mathcal B_\eps^{2,3}(\mathbb F_{n:3})(t,z_1,z_2)
\\
+\mathbb F_{n:2}(t,z'_1,z'_2)\Big|_{\text{dist}(x_1,x_2)=\eps+0}((v_1-v_2)\cdot n_{12})_+\delta_0(x_1-x_2)
\\
-\mathbb F_{n:2}(t,z_1,z_2)\Big|_{\text{dist}(x_1,x_2)=\eps+0}((v_1-v_2)\cdot n_{12})_-\delta_0(x_1-x_2)&\,.
\end{aligned}
$$
Les intégrales de collision sur la seconde ligne sont données par 
$$
\begin{aligned}
\mathcal B^{1,3}_\eps(\mathbb F_{n:2})(t,z_1,z_2)
\\
=\int_{\mathbf R^3\times\mathbf S^2}(\mathbb F_{n:3}(t,z'_1,z_2,x_1-\eps\omega,v'_3)-\mathbb F_{n:3}(t,z_1,z_2,x_1+\eps\omega,v_3))((v_1-v_3)\cdot\omega)_+d\omega dv_3&\,,
\\
\mathcal B^{2,3}_\eps(\mathbb F_{n:2})(t,z_1,z_2)
\\
=\int_{\mathbf R^3\times\mathbf S^2}(\mathbb F_{n:3}(t,z'_1,z_2,x_2-\eps\omega,v'_3)-\mathbb F_{n:3}(t,z_1,z_2,x_2+\eps\omega,v_3))((v_1-v_3)\cdot\omega)_+d\omega dv_3&\,.
\end{aligned}
$$
Les deux derniers termes sur les troisièmes et quatrièmes lignes traduisent les collisions entre les molécules n$^\textbf{os}$ $1$ et $2$. Comme ces molécules sont celles que l'on a choisies parmi les $n$ 
en considérant la marginale à deux corps $F_{n:2}(t,z_1,z_2)$, ces collisions correspondent à des corrélations entre les molécules en question, ce qui semble contredire l'hypothèse de chaos moléculaire
de Boltzmann, utilisée pour factoriser 
$$
\mathbb F_{n:2}(t,z'_1,z'_2)\Big|_{\text{dist}(x_1,x_2)=\eps+0}\mathbf 1_{(v_1-v_2)\cdot n_{12}>0}\quad\text{ et }\quad\mathbb F_{n:2}(t,z_1,z_2)\Big|_{\text{dist}(x_1,x_2)=\eps+0}\mathbf 1_{(v_1-v_2)\cdot n_{12}<0}
$$ 
dans le formalisme canonique utilisé dans \parencite{Lan75,CIP,GSRT,FG2014}, ou, de façon équivalente, la corrélation $F^\eps_2$ dans le formalisme grand-canonique utilisé ici, afin d'arriver à l'équation 
de Boltzmann. 

Intégrons donc ces deux derniers termes contre une fonction test $\phi\equiv\phi(t,z_1,z_2)$ et faisons tendre $\eps$ vers $0$ en appliquant l'hypothèse de chaos moléculaire en supposant que $\mathbb F_{n:2}$
se factorise en $g(t,\cdot)\otimes g(t,\cdot)$ pour les particules sur le point de collisionner: il vient
$$
\begin{aligned}
\frac1{\eps^2}\int_{(\mathbf T^3\times\mathbf R^3)^2}\mathbb F_{n:2}(t,z'_1,z'_2)\Big|_{\text{dist}(x_1,x_2)=\eps+0}((v_1-v_2)\cdot n_{12})_+\delta_0(x_1-x_2)\phi(z_1,z_2)dz_1dz_2
\\
\to\int_{\mathbf T^3\times(\mathbf R^3)^2\times\mathbf S^2}g(t,x,v'_1)g(t,x,v'_2)((v_1-v_2)\cdot \omega)_+\phi(x,v_1,x,v_2)dxdv_1dv_2d\omega&\,,
\end{aligned}
$$
et
$$
\begin{aligned}
\frac1{\eps^2}\int_{(\mathbf T^3\times\mathbf R^3)^2}\mathbb F_{n:2}(t,z_1,z_2)\Big|_{\text{dist}(x_1,x_2)=\eps+0}((v_1-v_2)\cdot n_{12})_-\delta_0(x_1-x_2)\phi(z_1,z_2)dz_1dz_2
\\
\to\int_{\mathbf T^3\times(\mathbf R^3)^2\times\mathbf S^2}g(t,x,v_1)g(t,x,v_2)((v_1-v_2)\cdot \omega)_+\phi(x,v_1,x,v_2)dxdv_1dv_2d\omega&\,.
\end{aligned}
$$
Donc
$$
\begin{aligned}
\frac1{\eps^2}\int_{(\mathbf T^3\times\mathbf R^3)^2}\mathbb F_{n:2}(t,z'_1,z'_2)\Big|_{\text{dist}(x_1,x_2)=\eps+0}((v_1-v_2)\cdot n_{12})_+\delta_0(x_1-x_2)\phi(z_1,z_2)dz_1dz_2
\\
-\frac1{\eps^2}\int_{(\mathbf T^3\times\mathbf R^3)^2}\mathbb F_{n:2}(t,z_1,z_2)\Big|_{\text{dist}(x_1,x_2)=\eps+0}((v_1-v_2)\cdot n_{12})_-\delta_0(x_1-x_2)\phi(z_1,z_2)dz_1dz_2
\\
\to\int_{(\mathbf T^3\times\mathbf R^3)^2}R^{1,2}(g,g)(t,z_1,z_2)\phi(z_1,z_2)dz_1dz_2&\,,
\end{aligned}
$$
ce qui montre le lien entre les recollisions des molécules n$^\textbf{os}$ $1$ et $2$ et l'opérateur $R^{1,2}$.

\smallskip
La formule de Spohn dans la proposition ci-dessus est donc particulièrement intéressante  car elle met en évidence la contribution des recollisions, qui sont des événements rares (noter qu'on a dû diviser leur contribution
par $\eps^2$ pour les mettre à une échelle observable dans la limite de Boltzmann--Grad), à la covariance limite. Quoique rares, ces recollisions ont un effet sur l'instabilité de la dynamique et sur les fluctuations,
puisqu'elles correspondent à des collisions entre des sphères de rayon négligeable --- rappelons que dans la dynamique du billard, plus le rayon des molécules est petit, plus grande est l'instabilité de la
dynamique.

\begin{proof}[Démonstration de la Proposition \ref{P-Spohn}]
On va se limiter au cas où $s=t$ --- le cas $s<t$ se traitant de façon analogue. La formule est évidemment correcte pour $t=s=0$, puisqu'on sait que
$$
\mathcal C(0,0,\phi,\psi)=\mathbb E(\langle\zeta_0,\phi\rangle\langle\zeta_0,\psi\rangle)=\int_{\mathbf T^3\times\mathbf R^3}\phi(z)\psi(z)f^{in}(z)dz\,.
$$
Considérons l'opérateur
$$
\begin{aligned}
\Sigma_t\psi(z_1):=-\delta_0(x_1-x_2)\int_{\mathbf R^3\times\mathbf R^3}(f(t,z_1)f(t,z_2)+f(t,z'_1)f(t,z'_2))\mathbf D\psi(z_1,z_2,\omega)
\\
\times((v_1-v_2)\cdot\omega)_+dz_1d\omega&\,.
\end{aligned}
$$

\begin{lemm}\label{L-Sigmat}
L'opérateur $\Sigma_t$ représente la covariance du bruit $\mathbf{Cov}(t,\cdot,\cdot)$, au sens où
$$
\mathbf{Cov}(t,\phi,\psi)=\int_{\mathbf T^3\times\mathbf R^3}\phi(z)\Sigma_t\psi(z)dz\,.
$$
Cet opérateur vérifie
$$
\begin{aligned}
\Sigma_t\phi(z)=&-(f(t,\cdot)\mathcal L^*_t+\mathcal L_tf(t,\cdot))\phi(z)+\phi(z)\partial_tf(t,z)
\\
&+\int_{\mathbf T^3\times\mathbf R^3}R^{1,2}(f(t,\cdot),f(t,\cdot))(z,z_*)\phi(z_*)dz_*\,.
\end{aligned}
$$
\end{lemm}

\smallskip
Admettons ce lemme, dont la démonstration résulte d'un calcul élémentaire laissé au lecteur. On part de la formule \eqref{DefSolFluct} avec $s=t$, soit
$$
\mathcal C(t,t,\phi,\psi)=\mathbb E(\langle\zeta_0,U(t,0)^*\phi\rangle\langle\zeta_0,U(t,0)^*\psi\rangle)+\int_0^t\mathbf{Cov}(\tau,U(t,\tau)^*\phi,U(t,\tau)^*\psi)d\tau\,,
$$
formule que l'on transforme comme suit:
$$
\begin{aligned}
\mathcal C(t,t,\phi,\psi)=&\int_{\mathbf T^3\times\mathbf R^3}U(t,0)^*\phi(z)U(t,0)^*\psi(z)f^{in}(z)dz
\\
&+\int_0^t\left(\int_{\mathbf T^3\times\mathbf R^3}\phi(z)U(t,\tau)\Sigma_tU(t,\tau)^*\psi(z)dz\right)d\tau\,.
\end{aligned}
$$
Utilisons maintenant la deuxième formule du lemme, que l'on écrit sous la forme
$$
\begin{aligned}
U(t,\tau)\Sigma_tU(t,\tau)^*\psi(z)=\partial_\tau\left(U(t,\tau)f(\tau,\cdot)U(t,\tau)^*\right)\psi(z)
\\
+U(t,\tau)\int_{\mathbf T^3\times\mathbf R^3}R^{1,2}(f(t,\cdot),f(t,\cdot))(z,z_*)U(t,\tau)^*\phi(z_*)dz_*&\,,
\end{aligned}
$$
d'où l'on tire que
$$
\begin{aligned}
\mathcal C(t,t,\phi,\psi)=\int_{\mathbf T^3\times\mathbf R^3}U(t,0)^*\phi(z)U(t,0)^*\psi(z)f^{in}(z)dz
\\
+\int_0^t\int_{\mathbf T^3\times\mathbf R^3}\phi(z)\partial_\tau\left(U(t,\tau)f(\tau,\cdot)U(t,\tau)^*\right)\psi(z)dzd\tau
\\
+\int_0^t\left(\int_{(\mathbf T^3\times\mathbf R^3)^2}(U(t,\tau)^*\phi)(z)R^{1,2}(f(t,\cdot),f(t,\cdot))(z,z_*)(U(t,\tau)^*\psi)(z_*)dz\right)d\tau&\,.
\end{aligned}
$$
Par intégration en $\tau$, la seconde intégrale au membre de droite se combine avec le premier terme au membre de droite pour donner ce qui est bien la formule annoncée dans le cas $s=t$.
\end{proof}

\section{Esquisse de preuve pour le Théorème \ref{T-HJ}}


On se limitera dans cette section à quelques trop brèves indications. On espère qu'elles inciteront les lecteurs à étudier en détail \parencite{BoGalSRSim1}, et leur serviront de guide de lecture.

Une première remarque s'impose: comme pour la preuve du théorème de Lanford, la démonstration du Théorème \ref{T-HJ} passe par une représentation {\og  assez explicite\fg}   des cumulants de tous ordres, 
plutôt que par une analyse des solutions de l'équation \eqref{HJ2} --- sans même parler d'un passage à la limite pour $\eps\to 0^+$ dans \eqref{HJeps} pour en déduire l'équation \eqref{HJ1}.

\subsection{Représentation des corrélations}


Commençons par la solution de l'équation de Boltzmann \eqref{BoltzEq}
$$
(\partial_t+v_1\cdot\nabla_{x_1})F_1(t,z_1)=\mathcal B_+(F_1)(t,z_1)-\mathcal B_-(F_1)(t,z_1)\,,
$$
où l'intégrale des collisions est décomposée en termes de perte $\mathcal B_-$ et de gain $\mathcal B_+$:
$$
\begin{aligned}
\mathcal B_-(F_1)(t,z_1)\!=\!F_1(t,x_1,v_1)\int_{\mathbf R^3\times\mathbf S^2}F_1(t,x_1,v_2)((v_1-v_2)\cdot\omega)_+dv_2d\omega\,,
\\
\mathcal B_+(F_1)(t,z_1)=\int_{\mathbf R^3\times\mathbf S^2}\!F_1(t,x_1,v'_1)F_1(t,x_1,v'_2)((v_1-v_2)\cdot\omega)_+dv_2d\omega\,,
\end{aligned}
$$
avec 
$$
v'_1=v_1-((v_1-v_2)\cdot\omega)\omega\,,\qquad v'_2=v_2+((v_1-v_2)\cdot\omega)\omega\,.
$$
On exprime ensuite $F_1$ par la formule de Duhamel
$$
\begin{aligned}
F_1(t,z_1)=F_1(0,x_1-tv_1,v_1)+\int_0^t\mathcal B_+(F_1)(t-s,x-sv_1,v_1)ds
\\
-\int_0^t\mathcal B_-(F_1)(t-s,x-sv_1,v_1)ds&\,,
\end{aligned}
$$
puis l'on itère cette formule en exprimant à nouveau $F_1$ par cette même formule dans les termes de perte $\mathcal B_-$ et de gain $\mathcal B_+$. On aboutit ainsi à représenter $F_1$ par une série 
de la forme
\begin{equation}\label{SerF1}
F_1(t,z_1)=\sum_{n\ge 0}\sum_{\mathbb A_{1,n}}\int_{([0,t_1]\times\mathbf R^3\times\mathbf S^2)^{n-1}}(f^{in})^{\otimes n}(\Psi_{1,n}(0))S(\Psi_{1,n})dT_{2,n}dV_{2,n}d\Omega_{2,n}\,,
\end{equation}
où $V_{2,n}:=(v_2,\ldots,v_n)$, avec des définitions similaires de $\Omega_{2,n}$ et de $T_{2,n}$, à ceci près que l'on suppose en plus $0<t_n<t_{n-1}<\ldots<t_2<t_1$. Chaque terme dans cette série est 
représenté par un arbre\footnote{Pour ce qui est des arbres et des graphes, on utilisera la terminologie de l'annexe du chapitre IV de \parencite{BourbakiLie46}.} $\mathbb A_{1,n}$ enraciné\footnote{C'est-à-dire
avec un sommet distingué.} en $1$ avec $n-1$ points de ramification $a_i\in\{2,\ldots,n\}$ correspondant à des instants de collision $t_2>\ldots>t_n>0$ pour $i=2,\ldots,n$, comme 
sur la figure \ref{F-PsiTraj}. La pseudo-trajectoire $\Psi_{1,n}$ est définie par la règle suivante:

\noindent
(a) sur $]t_i,t_{i-1}[$, on fait évoluer un groupe de $i-1$ particules par la dynamique du transport libre $Z_{i-1}\mapsto(X_{i-1}-(t_{i-1}-t)V_{i-1},V_{i-1})$;

\noindent
(b) au temps $t_i$, on ajoute la particule n$^\mathrm{o}\,i$ à la position $x_{a_i}(t_i)$ avec la vitesse $v_i$;

\noindent
(c) si $v_{a_i}(t_i+0)-v_i)\cdot\omega_i>0$, on effectue la collision en remplaçant $(v_{a_i}(t_i+0),v_i)$ par $(v'_{a_i}(t_i+0),v'_i)$ au moyen des relations de collision paramétrées par $\omega_i$, i.e.
$$
(v'_{a_i}(t_i+0),v'_i)=(v_{a_i}(t_i+0)-((v_{a_i}(t_i+0)-v_i)\cdot\omega_i)\omega_i,v_i+((v_{a_i}(t_i+0)-v_i)\cdot\omega_i)\omega_i)
$$
et on itère le processus ci-dessus en posant
$$
(v_{a_i}(t_i-0),v_i(t_i-0))=(v'_{a_i}(t_i+0),v'_i)\,;
$$
si $(v_{a_i}(t_i+0)-v_i)\cdot\omega_i>0$, on poursuit de même sans changer les vitesses $v_{a_i}(t_i+0)$ et $v_i$ (ce dernier cas correspondant à la contribution des termes de perte $\mathcal B_-$).

Ces règles définissent ainsi une pseudo-trajectoire $\Psi_n(t)$ issue de $z_1$ à l'instant $t$, et dont la valeur à $t=0$ est $\Psi_n(0)\in(\mathbf T^3\times\mathbf R^3)^n$.

Dans la série ci-dessus, on pose enfin
$$
S(\Psi_{1,n}):=\prod_{i=1}^n\pm((v-v_{a_i}(t_i+0))\cdot\omega_i)_+\,,
$$
où le préfacteur $\pm$ vaut $+$ en cas de collision avec changement de vitesse, et $-$ dans le cas contraire.

Une représentation analogue existe pour $F^\eps_1$ avant le passage à la limite $\eps\to 0^+$:
\begin{equation}\label{SerF1eps}
F^\eps_1(t,z_1)=\sum_{n\ge 0}\sum_{\mathbb G_{1,n}}\int_{\mathcal E^\eps}F^\eps_n(\Psi^\eps_{1,n}(0))S^\eps(\Psi_{1,n})dT_{2,n}dV_{2,n}d\Omega_{2,n}\,.
\end{equation}

Les différences avec le calcul de $F_1$ sont les suivantes:

\noindent
(a') sur $]t_i,t_{i-1}[$, on fait évoluer un groupe de $i-1$ particules par la dynamique du billard à $i-1$ corps \eqref{Newton}-\eqref{Collxk}-\eqref{Collvkvj}, au lieu du transport libre; il y a donc en général
des recollisions, comme expliqué dans la section précédente, de sorte que cette représentation est paramétrée par un graphe $\mathbb G_{1,n}$ pouvant présenter des circuits (chaque recollision créant
un circuit: voir figure \ref{F-PsiTraj}), et non plus par un arbre;

\noindent
(b') au temps $t_i$, on ajoute la particule n$^\mathrm{o}\,i$ à la position $x_{a_i}(t_i)+\eps\omega_i$ au lieu de $x_{a_i}(t_i)$ avec la vitesse $v_i$;

\noindent
(c') on intègre sur les paramètres admissibles $\mathcal E^\eps$ au lieu de $([0,t_1]\times\mathbf R^3\times\mathbf S^2)^{n-1}$ avec la seule contrainte $0<t_n<\ldots<t_2<t_1$, car les particules ne doivent 
se recouvrir à aucun instant. En particulier $\Psi^\eps_n(0)\in\Gamma^\eps_n$.


\begin{figure}\label{F-PsiTraj}

\begin{center}
\includegraphics[width=5cm]{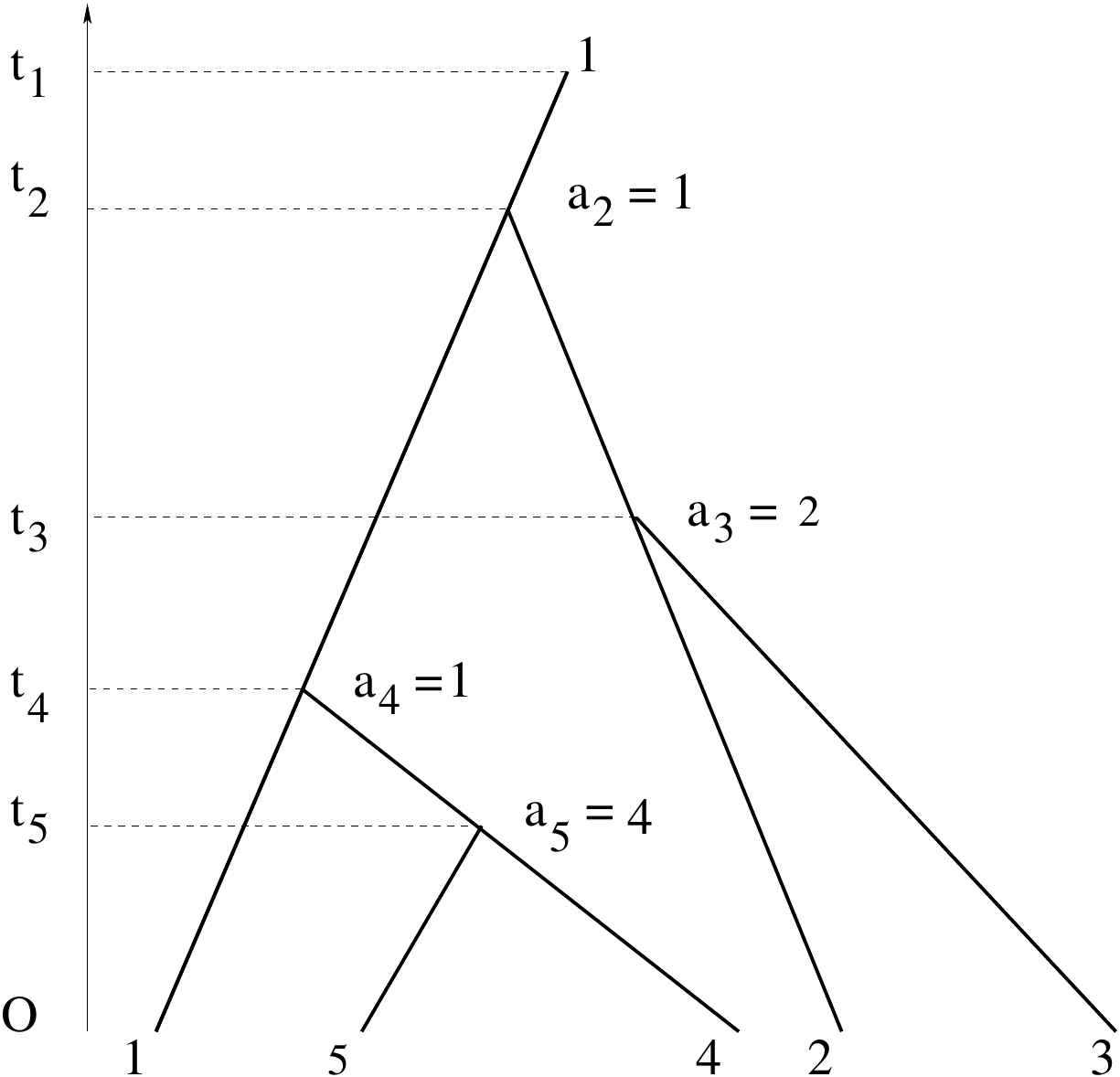}\hskip 0.5cm\includegraphics[width=5cm]{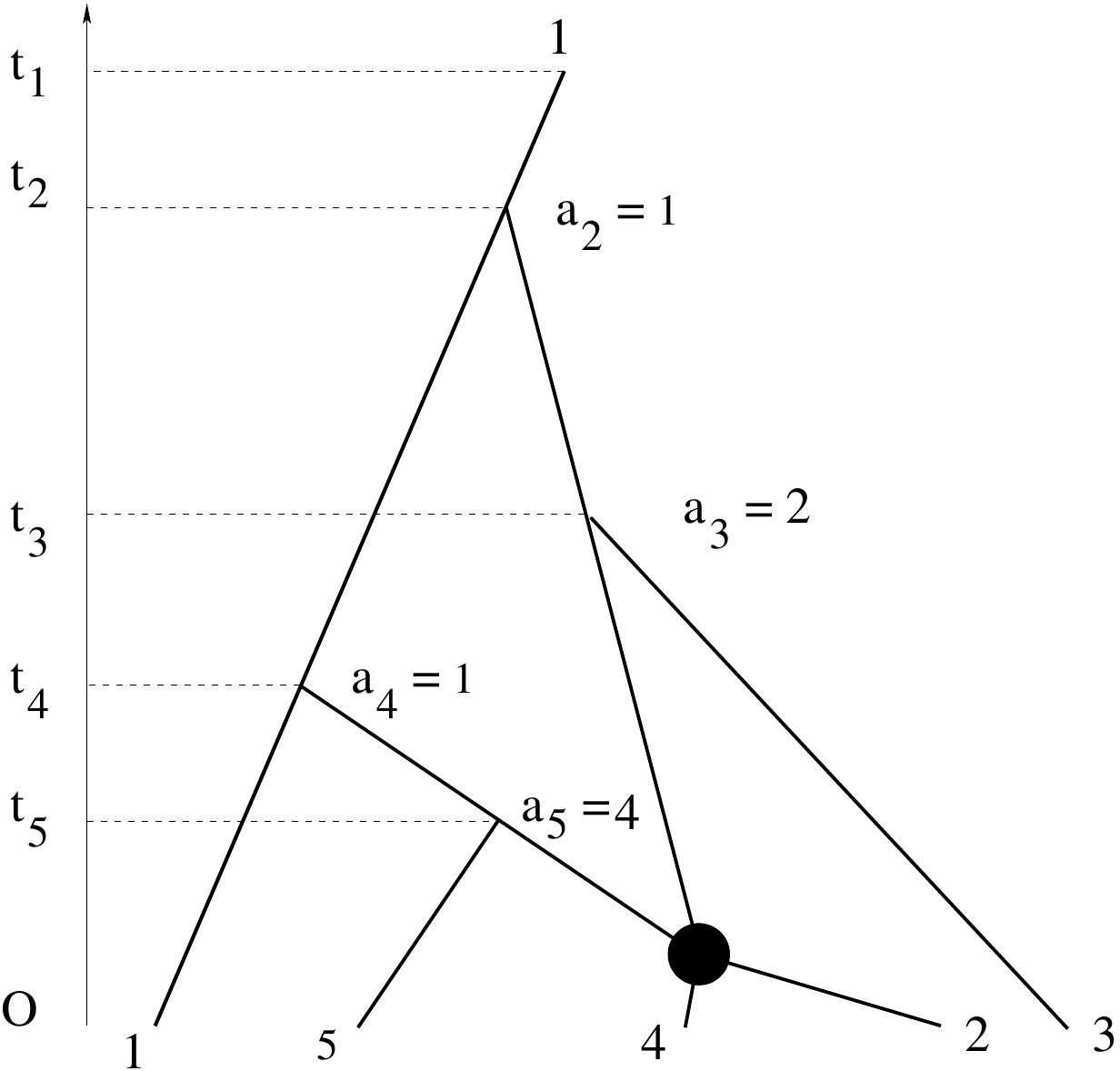}

\caption{À gauche un arbre de type $\mathbb A_{1,5}$ à $5$ molécules correspondant à la dynamique de l'équation de Boltzmann. À droite, un graphe de type $\mathbb G_{1,5}$ à $5$ molécules correspondant 
à la dynamique \eqref{Newton}-\eqref{Collxk}-\eqref{Collvkvj}. L'intersection signalée par un disque noir désigne une {\og  recollision\fg}   entre les molécules n$^\mathrm{os}2$ et $4$ avant que celles-ci n'entrent en collision 
avec la molécule n$^\mathrm{o}1$. Lorsqu'à l'instant $t_4-0$ la particule n$^\mathrm{o}4$ est ajoutée aux particules n$^\mathrm{os}1,2,3$, elle n'en est pas indépendante car elle a déjà subi une collision avec la 
molécule n$^\mathrm{o}2$ auparavant (dans le cas présent, avant l'instant $t_5$). Le triangle de sommets les points de ramification $a_2, a_4$ et le disque noir correspondant à la recollision entre molécules n$^\mathrm{os} \,2$ 
et $4$ est un circuit dans le graphe $\mathbb G_{1,5}$.}

\end{center}

\end{figure}


Évidemment, il existe une représentation analogue à \eqref{SerF1eps} pour les corrélations d'ordre supérieur $F^\eps_k$. Au lieu d'un seul arbre, elle met en jeu $k$ graphes de sommets marqués $1,2,\ldots,k$, 
avec $n_1,\ldots,n_k$ arêtes, et ces graphes peuvent bien sûr interagir.

\subsection{Représentation des cumulants}


Examinons le cas du second cumulant $f^\eps_2:=\mu_\eps(F^\eps_2-F^\eps_1\otimes F^\eps_1)$. Comme expliqué plus haut, la représentation de $F^\eps_2$ met en jeu deux graphes avec chacun un sommet
marqué, notés $1$ et $2$. Deux cas se présentent alors: (i) ces deux graphes sont disconnectés, c'est-à-dire restent à une distance supérieure à $\eps$ sur l'intervalle de temps $[0,t]$, (ii) à un (ou plusieurs) 
instants dans l'intervalle de temps $[0,t]$, ces deux graphes subissent une recollision \textit{externe}, mettant en jeu une arête de chacun des deux graphes. 

Les graphes disconnectés ne sont pas indépendants, car ils sont corrélés par la relation d'exclusion (le fait d'être à distance supérieure à $\eps$ sur $[0,t]$). De même que, d'après \eqref{DefFnin}
$$
\mathbb F^{in}_2(z_1,z_2)=f^{in}\otimes f^{in}(z_1,z_2)-f^{in}(z_1)f^{in}(z_2)\mathbf 1_{|x_1-x_2|\le\eps}\,,
$$
(où on peut penser aux points $z_1$ et $z_2$ comme à deux arbres triviaux à un seul sommet), on peut représenter la contribution des graphes disconnectés comme différence entre le carré tensoriel de $F^\eps_1$
et celle où les deux graphes de sommets marqués $1$ et $2$ se trouvent à une distance inférieure à $\eps$ à un (ou plusieurs) instants dans l'intervalle de temps $[0,t]$ sans transformation de vitesse, situation qu'on 
désigne du nom de {\og  recouvrement\fg}.


\begin{figure}

\begin{center}

\includegraphics[width=3cm]{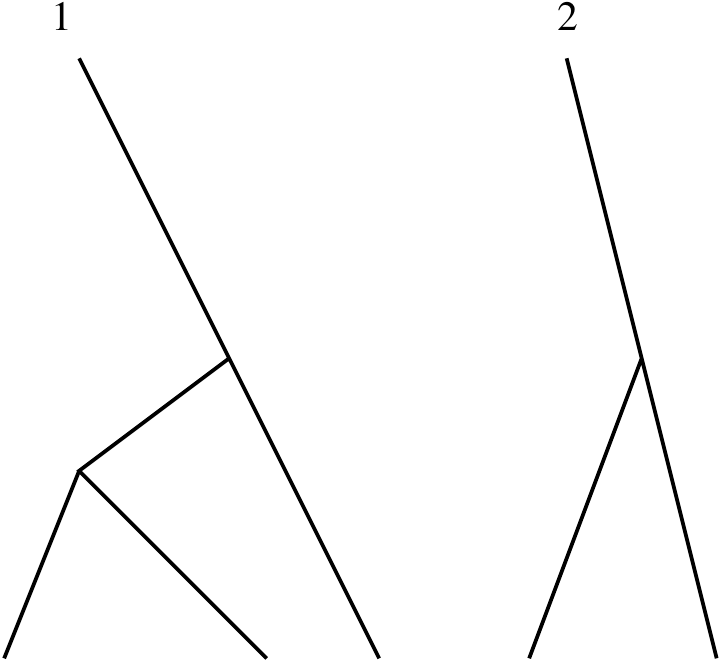}\hskip .5cm\includegraphics[width=3cm]{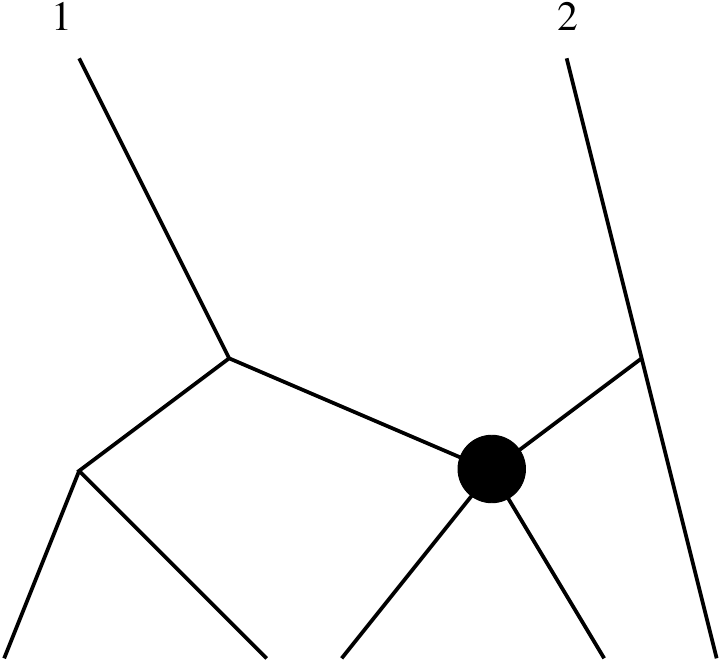}\hskip .5cm\includegraphics[width=3cm]{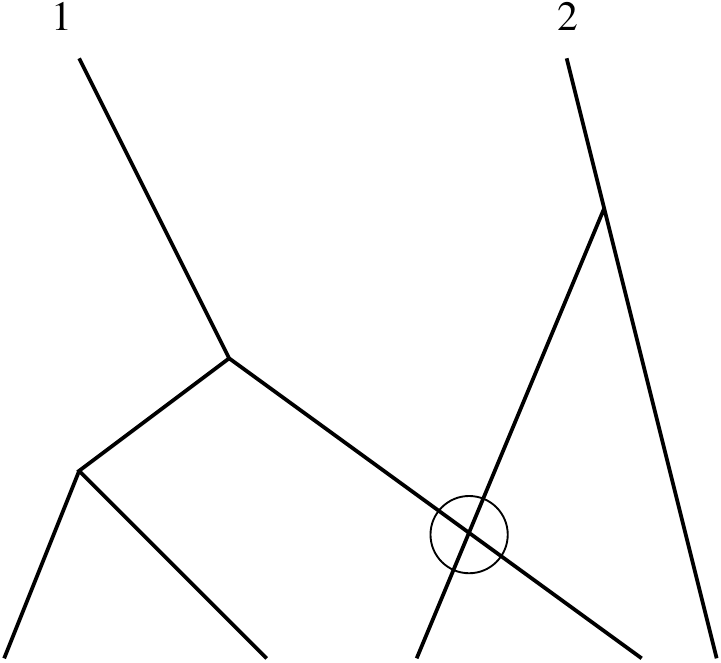}

\caption{À gauche les arbres $(1)$ et $(2)$ sont disconnectés, ce qu'on note $(1)\not\sim(2)$; au milieu, les arbres $(1)$ et $(2)$ subissent une recollision externe (disque noir), ce qu'on note $(1)\sim_r(2)$;
enfin, à droite, les arbres $(1)$ et $(2)$ sont en situation de recouvrement (cercle), ce qu'on note $(1)\sim_o(2)$.}

\end{center}

\end{figure}


À partir de là, on arrive à une représentation du second cumulant de la forme

\[
\begin{aligned}
f^\eps_2=\mu_\eps(F^\epsilon_2-F^\epsilon_1\otimes F^\epsilon_1)
\\
\\
=\mu_\eps\sum_{(1)\sim_r(2)}\ctikz{\draw[-] (-1,-1)--(1,1);\draw[-] (0,0)--(0.5,-0.5); \draw[-] (0.2,-1)--(0.5,-0.5);\draw[-] (0.7,-1)--(0.5,-0.5); \draw[-] (0.5,-0.5)--(2,1);\draw[-] (1,0)--(2,-1); \fill(0.5,-0.5) circle[radius=4pt];}
-\mu_\eps\sum_{(1)\sim_o(2)}\ctikz{ \draw[-] (-1,-1)--(1,1); \draw[-] (0,0)--(1,-1); \draw[-] (0,-1)--(2,1); \draw[-] (1,0)--(2,-1);\draw(0.5,-0.5) circle[radius=4pt]}
\\
\\
+\text{terme obtenu par propagation du second cumulant à }t=0
\end{aligned}
\]
(voir dans la section 4.4.2 de \parencite{BoGalSRSim1} la formule suivant (4.4.1), ou encore la figure 10 de \parencite{BGSRSicm}, où le terme provenant de la condition initiale n'est pas mentionné, comme
expliqué dans la Remarque 3.1 de cette dernière référence). 

Partons de $z_1,z_2$ à l'instant $t$ tels que $|(x_1-x_2)-(t-s)(v_1-v_2)|\le\eps$ pour au moins une valeur de $s\in[0,t]$, et considérons le cas très particulier de deux arbres sans point de ramification 
(autrement dit sans collisions entre $[0,t]$). La contribution d'un tel terme à $F^\eps_2-F^\eps_1\otimes F^\eps_1$ sera génériquement d'ordre $1$ en norme $L^\infty$ (sauf évidemment dans le cas d'une 
maxwellienne, où les termes de perte et de gain se compensent exactement pour chaque vecteur unitaire $\omega$ dans la loi de collision \eqref{LoiColl}). Ce qu'on donc peut espérer, dans le meilleur des 
cas, est que  $F^\eps_2-F^\eps_1\otimes F^\eps_1$ soit petit en variation totale, comme terme d'ordre un intégré sur un ensemble de mesure petite --- dans ce cas précis très particulier, sur un cylindre de 
section $\pi\eps^2\times t|v_1-v_2|$. Ceci justifie heuristiquement la mise à l'échelle de $F^\eps_2-F^\eps_1\otimes F^\eps_1$, que l'on doit multiplier par $\mu_\eps$ pour compenser l'évanescence de
ce terme dans la limite de Boltzmann--Grad.

Ces remarques très fragmentaires donnent un aperçu de comment  calculer et estimer le cumulant d'ordre $2$. On comprend qu'il existe une représentation analogue pour les cumulants d'ordre quelconque. 
En gros, les contributions au cumulant d'ordre $k$ sont d'ordre $\mu_\eps^{k-1}$ (provenant de la mise à l'échelle dans la définition \eqref{DefCumul} de $f^\eps_k$), multiplié par le nombre de graphes sur 
lesquels on doit sommer, terme que l'on multiplie à son tour par $O((\eps^2\times t(|v_1|+\ldots+|v_k|))^{k-1})$, qui est la mesure du domaine d'intégration, multipliée par la taille de la fonction à intégrer, pour 
laquelle on utilise la décroissance gaussienne de la condition initiale, afin de compenser la contribution de la norme des vitesses relatives à la taille du domaine d'intégration. Quant au nombre de graphes 
sur lesquels sommer, on se ramène à $k^{k-2}$ (nombre d'arbres non orientés à $k$ sommets étiquetés, d'après la formule de Cayley --- cf. \parencite{PftB}, chapitre 30), multiplié par $2^{k-1}$ (correspondant 
au choix d'un signe $\pm$ à chaque sommet excepté la racine). Ce raisonnement conduit à la formule (3.4) de \parencite{BGSRSicm}.

Le lecteur intéressé à connaître le détail de ces calculs et des estimations qui en découlent est invité à lire les chapitres 3 et 4 de \parencite{BoGalSRSim1}, qui précisent les formules évoquées ci-dessus 
de représentation des cumulants de tous ordres en termes de graphes, ainsi que le chapitre 8, pour la démonstration des bornes sur les cumulants de tous ordres (contenant le résultat énoncé plus haut
comme Proposition \ref{P-BorneCum}). La limite en $\eps\to 0^+$ conduisant à la fonction $\mathbf K(t,h)$ et au Théorème \ref{T-HJ} est traitée dans le chapitre 9 de \parencite{BoGalSRSim1}.

\section{Épilogue: par-delà le temps de Lanford}\label{S->TLanford}


Tous les résultats dont il a été question jusqu'ici dans cet exposé portaient sur des intervalles de temps courts (de l'ordre du {\og  temps de Lanford\fg}, c'est-à-dire de $T_0=T_0[C_0,\beta_0]$ dans l'énoncé
du Théorème \ref{T-Lanford}). Rappelons que cette restriction avait été annoncée dans le troisième paragraphe de la section \ref{S-HJ}. Mais la dernière phrase de ce paragraphe faisait également espérer
que l'on puisse également obtenir des informations sur la limite de Boltzmann--Grad après le temps de Lanford.

Rappelons que \textcite{BGSRInvent} ont réussi à pousser la limite de Boltzmann--Grad par-delà le temps de Lanford dans des régimes particuliers (conduisant à l'équation de Boltzmann linéaire  puis à 
une asymptotique de diffusion, ou au voisinage d'un équilibre maxwellien, puis dans une limite hydrodynamique conduisant aux équations de Stokes dépendant du temps \parencite{BGSRapde}). Sur le 
premier de ces deux résultats, on pourra consulter également \parencite{FG2014}.

L'idée de pousser l'analyse des fluctuations de la mesure empirique autour de la solution de l'équation de Boltzmann s'impose alors comme particulièrement naturelle. On peut en effet espérer que ce
cadre particulier offre la possibilité de justifier une équation aux fluctuations sur des plages de temps tendant vers l'infini avec $1/\eps$. Voici le théorème obtenu par \textcite{BGSRScpam,BGSRSLongT}
dans cette direction (Théorème 4.1 de \parencite{BGSRSicm}).

Notons $M(v):=\tfrac1{(2\pi)^{3/2}}e^{-|v|^2/2}$ la maxwellienne de densité $1$, de vitesse moyenne nulle et de température $1$, autrement dit la gaussienne centrée réduite. Notons également
$$
\begin{aligned}
\mathbf L_M\phi(t,x,v):=\int_{\mathbf R^3\times\mathbf S^2}&(M(v')\phi(t,x,v'_*)+M(v'_*)\phi(t,x,v')
\\
&-M(v)\phi(t,x,v_*)-M(v_*)\phi(t,x,v))((v-v_*)\cdot\omega)_+dv_*d\omega\,,
\\
\mathbf L^*_M\phi(t,x,v):=\int_{\mathbf R^3\times\mathbf S^2}&(\phi(t,x,v'_*)+\phi(t,x,v')
\\
&-\phi(t,x,v_*)-\phi(t,x,v))((v-v_*)\cdot\omega)_+M(v_*)dv_*d\omega\,,
\end{aligned}
$$
les intégrales de collision directe et adjointe linéarisées autour de $M$. 

\begin{theo}
Supposons que $f^{in}=M$. Dans la limite de Boltzmann--Grad $\mu_\eps=\eps^{-2}$ pour le formalisme grand-canonique, le champ de fluctuations $\zeta^\eps_t$ de \eqref{DefFluct}  converge en loi
sur tout intervalle de temps de la forme $[0,T_\eps]$ avec $T_\eps=O(\ln\ln\ln\mu_\eps)$ vers le processus gaussien centré solution de l'équation de Boltzmann aux fluctuations
$$
d\zeta_t=(-v\cdot\nabla_x+\mathbf L_M)\zeta^\eps_t+db_t\,,
$$
où $b_t$ est un bruit gaussien delta-corrélé en $t,x$, de covariance
$$
\begin{aligned}
\mathbf{Cov}(t,\phi,\psi):=\tfrac12\int_{\mathbf T^3\times(\mathbf R^3)^2\times\mathbf S^2}\mathbf D\phi(x,v,x,v_*,\omega)\mathbf D\psi(x,v,x,v_*,\omega)
\\
\times M(v)M(v_*)((v-v_*)\cdot\omega)_+dxdvdv_*d\omega&\,,
\end{aligned}
$$
et où $\zeta_0$ est le champ gaussien centré de covariance
$$
\mathbf E(\langle\zeta_0,\phi\rangle\langle\zeta_0,\psi\rangle)=\int_{\mathbf T^3\times\mathbf R^3}\phi(z)\psi(z)M(v)dz\,.
$$
\end{theo}

\smallskip
On laissera de côté la démonstration de ce dernier résultat. Ce qui en est dit dans la section 4.3 de \parencite{BGSRSicm} fait comprendre qu'elle est particulièrement complexe. Signalons la formule
suivante pour la covariance dans ce cas:
$$
\mathcal C(t,0,h,g^{in})=\int_{\mathbf T^3\times\mathbf R^3}g(t,x,v)h(x,v)M(v)dxdv\,,
$$
où $g(t,\cdot)=e^{t(-v\cdot\nabla_x+\mathbf L^*_M)}g^{in}$. Cette formule se déduit facilement de celle de Spohn (Proposition \ref{P-Spohn}), au moins sur le laps de temps  bref (de l'ordre du temps de 
Lanford) sur lequel elle a été établie jusqu'ici. Il suffit en effet d'utiliser l'identité immédiate $M\mathbf L^*_M(g)=L_M(Mg)$, et d'observer que $R^{1,2}(M,M)=0$.

Enfin, comme la longueur de l'intervalle de temps sur lequel l'équation de Boltzmann aux fluctuations autour de l'équilibre maxwellien est établie tend vers l'infini avec $1/\eps$, on peut également étudier
ces fluctuations dans un régime hydrodynamique convenablement linéarisé, pour obtenir un système de Stokes-Fourier aux fluctuations --- sur ce dernier point, voir l'article très récent \parencite{BGSRSGawed}.


\printbibliography


\end{document}